\documentclass[a4paper]{amsart}
\usepackage{amsmath}
\usepackage{amsthm}
\usepackage{amssymb}
\usepackage[square,sort,comma,numbers]{natbib}
\usepackage{cases}
\usepackage{tcolorbox}
\usepackage{tikz}
\usepackage{tkz-graph}
\usepackage{indentfirst}
\usepackage{float}
\usepackage{enumitem}
\usepackage[cal=boondoxupr,calscaled=1]{mathalfa}
\usepackage{ragged2e}
\usepackage{subcaption}
\usepackage{graphicx}
\usepackage{dirtytalk}

\let\temp\rmdefault
\usepackage{kmath}
\let\rmdefault\temp

\DeclareMathSizes{10}{9.5}{6}{6}

\renewcommand{\familydefault}{\rmdefault}

\tikzstyle{vertex1}=[auto=left,circle,draw,fill=black!25,minimum size=10pt,inner sep=2pt]
\tikzstyle{Dot}=[auto=left,circle,fill=black!100,minimum size=2pt,inner sep=1.5pt]
\tikzstyle{dot}=[auto=left,circle,fill=black!100,minimum size=0.5pt,inner sep=0.5pt]
\tikzstyle{Dot2}=[auto=left,circle,fill=black!100,minimum size=1pt,inner sep=1pt]

\newtheorem{Theorem}{Theorem}[section]
\newtheorem{Lemma}{Lemma}[section]

\newtheorem{Definition}{Definition}[section]
\newtheorem{Example}{Example}[section]

\newcounter{cases}
\newcounter{subcases}
\newcounter{subsubcases}

\newenvironment{mycases}
{%
	\setcounter{cases}{0}%
	\def\case
	{%
		\par\noindent
		\refstepcounter{cases}%
		\textbf{Case \thecases.}
	}%
}
{%
	\par
}
\newenvironment{subcases}
{%
	\setcounter{subcases}{0}%
	\def\subcase
	{%
		\par\noindent
		\refstepcounter{subcases}%
		\textit{\textbf{Subcase (\thesubcases):}}
	}%
}
{%
}
\newenvironment{subsubcases}
{%
	\setcounter{subsubcases}{0}%
	\def\subsubcase
	{%
		\par\noindent
		\refstepcounter{subsubcases}%
		\textit{\textbf{\thesubcases-(\thesubsubcases):}}
	}%
}
{%
}
\renewcommand*\thecases{\arabic{cases}}
\renewcommand*\thesubcases{\Roman{subcases}}
\renewcommand*\thesubsubcases{\roman{subsubcases}}
\counterwithin{figure}{section}

\newcommand*{\mybox}[1]{\framebox{#1}}

\newcounter{boxlblcounter}  
\newcommand{\makeboxlabel}[1]{\fbox{#1.}\hfill}

\theoremstyle{definition}

\makeatletter
\@addtoreset{proofpart}{Theorem}

\title[Fault Tolerant Metric Dimension of Leafless Cacti Graphs]{Fault Tolerant Metric Dimensions of Leafless Cacti Graphs with Application in Supply Chain Management}

\author[T. Asif]{Tauseef Asif}
\address{Department of Mathematics and Statistics, The University of Haripur, Pakistan}
\email{tasif@uoh.edu.pk}
\author[G. Haidar]{Ghulam Haidar}
\address{Department of Mathematics and Statistics, The University of Haripur, Pakistan}
\email{haidersehani012@gmail.com}
\author[F. Yousafzai]{Faisal Yousafzai}
\address{Department of Basic Sciences and Humanities, National University of Sciences and Technology, Islamabad, Pakistan}
\email{yousafzaimath@gmail.com}
\author[M. U. I. Khan]{Murad ul Islam Khan*}\thanks{*Corresponding author}
\address{Department of Mathematics and Statistics, The University of Haripur, Pakistan}
\email{muradulislam@uoh.edu.pk}
\author[Q. Khan]{Qaisar Khan}
\address{Department of Mathematics and Statistics, The University of Haripur}
\email{qaisar.khan@uoh.edu.pk}
\author[R. Fatima]{Rakea Fatima}
\address{Department of Mathematics and Statistics, The University of Haripur}
\email{}

\newcommand{\C}{\ensuremath{\mathcal{C}}}
\newcommand{\E}{\ensuremath{\mathcal{E}}}
\newcommand{\G}{\ensuremath{\mathcal{G}}}
\newcommand{\PP}{\ensuremath{\mathcal{P}}}
\newcommand{\A}{\ensuremath{\mathcal{A}}}
\newcommand{\B}{\ensuremath{\mathcal{B}}}
\newcommand{\SSS}{\ensuremath{\mathcal{S}}}
\newcommand{\W}{\ensuremath{\mathcal{W}}}
\newcommand{\V}{\ensuremath{\mathcal{V}}}
\newcommand{\BB}{\ensuremath{\mathcal{B}}}
\newcommand{\uu}{\ensuremath{\mathcal{u}}}
\newcommand{\vv}{\ensuremath{\mathcal{v}}}
\newcommand{\ww}{\ensuremath{\mathcal{w}}}
\newcommand{\aaa}{\ensuremath{\mathcal{a}}}
\newcommand{\bb}{\ensuremath{\mathcal{b}}}
\newcommand{\dd}{\ensuremath{\mathcal{d}}}
\newcommand{\n}{\ensuremath{\mathcal{n}}}
\newcommand{\R}{\ensuremath{\mathcal{R}}}
\newcommand{\m}{\ensuremath{\mathcal{m}}}
\newcommand{\rr}{\ensuremath{\mathcal{r}}}
\newcommand{\kk}{\ensuremath{\mathcal{k}}}
\newcommand{\el}{\ensuremath{\mathcal{l}}}
\newcommand{\ii}{\ensuremath{\mathcal{i}}}
\newcommand{\jj}{\ensuremath{\mathcal{j}}}
\newcommand{\q}{\ensuremath{\mathcal{q}}}
\newcommand{\p}{\ensuremath{\mathcal{p}}}
\newcommand{\s}{\ensuremath{\mathcal{s}}}

\begin{document}

\begin{abstract}
	A resolving set for a simple graph $\G$ is a subset of vertex set of $\G$ such that it distinguishes all vertices of $\G$ using the shortest distance from this subset. This subset is a metric basis if it is the smallest set with this property. A resolving set is a fault tolerant resolving set if the removal of any vertex from the subset still leaves it a resolving set. The smallest set satisfying this property is the fault tolerant metric basis, and the cardinality of this set is termed as fault tolerant metric dimension of $\G$, denoted by $\beta'(\G)$. In this article, we determine the fault tolerant metric dimension of bicyclic graphs of type-I and II and show that it is always $4$ for both types of graphs. We then use these results to form our basis to consider leafless cacti graphs, and calculate their fault tolerant metric dimensions in terms of \textit{inner cycles} and \textit{outer cycles}. We then consider a detailed real world example of supply and distribution center management, and discuss the application of fault tolerant metric dimension in such a scenario. We also briefly discuss some other scenarios where leafless cacti graphs can be used to model real world problems.
\end{abstract}

	\maketitle

\section{Introduction}
A vertex $\ww$ of a simple connected graph $\G=(\V,\E)$ is said to resolve the vertices $\uu,\vv$ of $\G$, if $\dd_{\G}(\uu,\ww)\neq \dd_{\G}(\vv,\ww)$, where $\dd_{\G}$ denotes the length of the shortest path between the given vertices. The distance vector of a vertex $\vv \in \V(\G)$ with respect to an ordered subset $\W=\{\ww_{1},\ww_{2}, \cdots, \ww_{{\kk}}\} \subset \G$ is defined by $\R(\vv|\W)=\{\dd(\vv,\ww_1), \dd(\vv, \ww_2), \cdots, \dd(\vv,\ww_{{\kk}})\}$. $\W$ is called a resolving set of $\G$ if all vertices of $\G$ have distinct distance vectors with respect to $\W$. If $\W$ is the smallest resolving set, it is termed as metric basis, and the cardinality of $\W$ is called the metric dimension of $\G$. These ideas were first introduced as \textit{\say{locating set}} and \textit{\say{location number}} by Slater \cite{sla1}. Hararay and Melter \cite{melter1976metric} independently presented these ideas as \textit{\say{resolving set}} and \textit{\say{metric basis}}. 

Since the inception, these concepts have been studied extensively. Exact values for metric dimensions of graph families have been considered in \cite{CHARTRAND200099,javaid2008families,shao2018metric,mohamed2023metric,math11040869, 10142411}, while the bounds on metric dimension for some other graph families are solved in \cite{caceres2007metric,CACERES20122618,GENESON2022123,sedlar2021bounds, wang2024metric}. Complexity problem for finding the metric dimensions of graphs has been discussed in \cite{diaz2011planar,diaz2017complexity}.

Many other graph invariants related to metric dimension have been defined in the literature e.g., local metric dimensions \cite{okamoto2010local}, $\mathcal{\kk}$-metric dimension \cite{estrada2013k}, and mixed metric dimensions \cite{kelenc2017mixed} etc.

Hernado et al. \cite{hernando2008fault} considered the problem for a simple connected graph $\G$ where a vertex of metric basis set $\W \subset \V(\G)$ is unavailable. Since the unique identification of all landmarks (vertices) is not possible due to this unavailability of $\ww \in \W$, an error or fault is introduced in the system. To overcome this phenomenon, they introduced the concept of fault tolerant metric basis. It is defined as the smallest resolving set $\W'$, such that the set $\W'-\{\ww : \ww \in \W'\}$ is again a resolving set for all $\ww \in \W'$. The cardinality of such a $\W'$ is called fault tolerant metric dimension, denoted by $\beta'$.

Many researchers considered this concept and studied it for different graphs. Hayat et al. \cite{9162022} studied the fault tolerant metric dimension of some families of interconnection networks. Basak et al \cite{BASAK202066} determined the exact values of the fault tolerant metric dimension for the circulant graphs $\C_{\n}(1,2,3)$ for all finite values of $\n$. Arulperumjothi et al. calculated $\beta'(\G)$ where $\G$ is a fractal cubic network. Prabhu et al. \cite{PRABHU2022126897} investigated this problem for Butterfly, Benes and silicates networks. For further studies, the articles \cite{Faheem2022OnTF, Liu2021ftmd,10.1371/journal.pone.0290411, Sharma2023, SharmaRazaBhat+2023+177+187} and the references therein may be consulted.

\section{Preliminaries}

Let $\G$ be a simple connected graph and $\W$ be a fault tolerant metric basis of $\G$. A subgraph $\G'$ of $\G$ is an induced subgraph, if $\V(\G') \subseteq \V(\G)$ and $\E(\G')$ contains all edges present in $\E(\G)$, among vertices of $\V(\G')$. We define the set $\W \cap \V(\G')$ as a fault tolerant metric generator of $\G'$ with respect to $\G$, and denote it by $\SSS_{\G'/\G}$. When the graph $\G$ is understood, we will just use the notion, fault tolerant metric generator $\SSS_{\G'}$. In other words, we say that the induced sub-graph $\G'$, contributes $|\SSS_{\G'}|$ vertices to the fault tolerant metric basis $\W$. 

A subset $\W_{\G'}$ of $\V(\G)$, is a fault tolerant distance producer for an induced subgraph $\G'$ of $\G$, if $\rr\left(\uu|\W_{\G'}-\{\ww'\}\right) \neq \rr(\vv|\W_{\G'}-\{\ww'\})$, for all $\uu, \vv \in \G'$ and $\ww' \in \W_{\G'}$. From these definitions, we can easily see that $|\W_{\G'}| \geq |\SSS_{\G'}|$.

\begin{Definition}
A simple connected graph $\G$ with $|\V(\G)|=\n$ is said to be bicyclic, if $%
|\E(\G)|=\n+1$.
\end{Definition}

There are three types of bicyclic graphs containing no leaves. We will only deal with the bicyclic graphs of type-I and II. These are defined as follows.
\begin{enumerate}[label=\Roman*.]
	\item $\C_{\n,\m}$ obtained by joining two cycles $\C_{\n}$ and $\C_{\m}$ through a single vertex. Let us consider that the cycle $\C_{\m}$ is attached to the vertex $\vv_{\n}$ of first cycle.
	
	\begin{figure}[H]
		\centering
		\tikzset{every picture/.style={line width=0.95pt}} 
		\resizebox{7.5cm}{3cm}{
		\begin{tikzpicture}[x=0.75pt,y=0.75pt,yscale=-1,xscale=1]
			
			\draw  [fill={rgb, 255:red, 0; green, 0; blue, 0 }  ,fill opacity=1 ] (310.59,117.97) .. controls (310.59,115.12) and (312.79,112.81) .. (315.51,112.81) .. controls (318.23,112.81) and (320.43,115.12) .. (320.43,117.97) .. controls (320.43,120.82) and (318.23,123.13) .. (315.51,123.13) .. controls (312.79,123.13) and (310.59,120.82) .. (310.59,117.97) -- cycle ;
			\draw   (158,117.97) .. controls (158,94.24) and (193.26,75) .. (236.75,75) .. controls (280.25,75) and (315.51,94.24) .. (315.51,117.97) .. controls (315.51,141.7) and (280.25,160.93) .. (236.75,160.93) .. controls (193.26,160.93) and (158,141.7) .. (158,117.97) -- cycle ;
			\draw   (315,117.97) .. controls (315,94.24) and (350.26,75) .. (393.75,75) .. controls (437.25,75) and (472.51,94.24) .. (472.51,117.97) .. controls (472.51,141.7) and (437.25,160.93) .. (393.75,160.93) .. controls (350.26,160.93) and (315,141.7) .. (315,117.97) -- cycle ;
			\draw  [fill={rgb, 255:red, 0; green, 0; blue, 0 }  ,fill opacity=1 ] (283.59,85.97) .. controls (283.59,83.12) and (285.79,80.81) .. (288.51,80.81) .. controls (291.23,80.81) and (293.43,83.12) .. (293.43,85.97) .. controls (293.43,88.82) and (291.23,91.13) .. (288.51,91.13) .. controls (285.79,91.13) and (283.59,88.82) .. (283.59,85.97) -- cycle ;
			\draw  [fill={rgb, 255:red, 0; green, 0; blue, 0 }  ,fill opacity=1 ] (338.59,85.97) .. controls (338.59,83.12) and (340.79,80.81) .. (343.51,80.81) .. controls (346.23,80.81) and (348.43,83.12) .. (348.43,85.97) .. controls (348.43,88.82) and (346.23,91.13) .. (343.51,91.13) .. controls (340.79,91.13) and (338.59,88.82) .. (338.59,85.97) -- cycle ;
			\draw  [fill={rgb, 255:red, 0; green, 0; blue, 0 }  ,fill opacity=1 ] (283.59,149.97) .. controls (283.59,147.12) and (285.79,144.81) .. (288.51,144.81) .. controls (291.23,144.81) and (293.43,147.12) .. (293.43,149.97) .. controls (293.43,152.82) and (291.23,155.13) .. (288.51,155.13) .. controls (285.79,155.13) and (283.59,152.82) .. (283.59,149.97) -- cycle ;
			\draw  [fill={rgb, 255:red, 0; green, 0; blue, 0 }  ,fill opacity=1 ] (338.59,149.97) .. controls (338.59,147.12) and (340.79,144.81) .. (343.51,144.81) .. controls (346.23,144.81) and (348.43,147.12) .. (348.43,149.97) .. controls (348.43,152.82) and (346.23,155.13) .. (343.51,155.13) .. controls (340.79,155.13) and (338.59,152.82) .. (338.59,149.97) -- cycle ;
			\draw  [fill={rgb, 255:red, 0; green, 0; blue, 0 }  ,fill opacity=1 ] (233.59,74.97) .. controls (233.59,72.12) and (235.79,69.81) .. (238.51,69.81) .. controls (241.23,69.81) and (243.43,72.12) .. (243.43,74.97) .. controls (243.43,77.82) and (241.23,80.13) .. (238.51,80.13) .. controls (235.79,80.13) and (233.59,77.82) .. (233.59,74.97) -- cycle ;
			\draw  [fill={rgb, 255:red, 0; green, 0; blue, 0 }  ,fill opacity=1 ] (390.59,75.97) .. controls (390.59,73.12) and (392.79,70.81) .. (395.51,70.81) .. controls (398.23,70.81) and (400.43,73.12) .. (400.43,75.97) .. controls (400.43,78.82) and (398.23,81.13) .. (395.51,81.13) .. controls (392.79,81.13) and (390.59,78.82) .. (390.59,75.97) -- cycle ;
			\draw  [draw opacity=0][dash pattern={on 4.5pt off 4.5pt}] (262.05,145.03) .. controls (254.3,146.6) and (245.75,147.47) .. (236.75,147.47) .. controls (201.68,147.47) and (173.25,134.26) .. (173.25,117.97) .. controls (173.25,102.74) and (198.1,90.2) .. (229.98,88.63) -- (236.75,117.97) -- cycle ; \draw [dash pattern={on 4.5pt off 4.5pt}] [dash pattern={on 4.5pt off 4.5pt}]  (258.99,145.61) .. controls (252.07,146.81) and (244.58,147.47) .. (236.75,147.47) .. controls (201.68,147.47) and (173.25,134.26) .. (173.25,117.97) .. controls (173.25,102.74) and (198.1,90.2) .. (229.98,88.63) ;  \draw [shift={(262.05,145.03)}, rotate = 171.56] [fill={rgb, 255:red, 0; green, 0; blue, 0 }  ][dash pattern={on 3.49pt off 4.5pt}][line width=0.08]  [draw opacity=0] (8.93,-4.29) -- (0,0) -- (8.93,4.29) -- cycle    ;
			\draw  [draw opacity=0][dash pattern={on 4.5pt off 4.5pt}] (412.02,89.71) .. controls (438.19,93.35) and (457.25,104.62) .. (457.25,117.97) .. controls (457.25,134.26) and (428.82,147.47) .. (393.75,147.47) .. controls (384.96,147.47) and (376.59,146.64) .. (368.97,145.14) -- (393.75,117.97) -- cycle ; \draw [dash pattern={on 4.5pt off 4.5pt}] [dash pattern={on 4.5pt off 4.5pt}]  (412.02,89.71) .. controls (438.19,93.35) and (457.25,104.62) .. (457.25,117.97) .. controls (457.25,134.26) and (428.82,147.47) .. (393.75,147.47) .. controls (386.06,147.47) and (378.69,146.83) .. (371.86,145.67) ; \draw [shift={(368.97,145.14)}, rotate = 8.42] [fill={rgb, 255:red, 0; green, 0; blue, 0 }  ][dash pattern={on 3.49pt off 4.5pt}][line width=0.08]  [draw opacity=0] (8.93,-4.29) -- (0,0) -- (8.93,4.29) -- cycle    ; 
			
			\draw (283,53.4) node [anchor=north west][inner sep=0.75pt]  [font=\footnotesize]  {$\vv_{1}$};
			\draw (233,43.4) node [anchor=north west][inner sep=0.75pt]  [font=\footnotesize]  {$\vv_{2}$};
			\draw (282,173.4) node [anchor=north west][inner sep=0.75pt]  [font=\footnotesize]  {$\vv_{\n-1}$};
			\draw (283,112.4) node [anchor=north west][inner sep=0.75pt]  [font=\footnotesize]  {$\vv_{\n}$};
			\draw (338,53.4) node [anchor=north west][inner sep=0.75pt]  [font=\footnotesize]  {$\vv_{\n+1}$};
			\draw (390,43.4) node [anchor=north west][inner sep=0.75pt]  [font=\footnotesize]  {$\vv_{\n+2}$};
			\draw (337,173.4) node [anchor=north west][inner sep=0.75pt]  [font=\footnotesize]  {$\vv_{\n+\m-1}$};
			\draw (182,164.4) node [anchor=north west][inner sep=0.75pt]  [font=\Large]  {$\C_{\n}$};
			\draw (412,165.4) node [anchor=north west][inner sep=0.75pt]  [font=\Large]  {$\C_{\m}$};

		\end{tikzpicture}
		
		}
		\caption{bicyclic graph of type-I}
		\label{bicyclictypei}
	\end{figure}
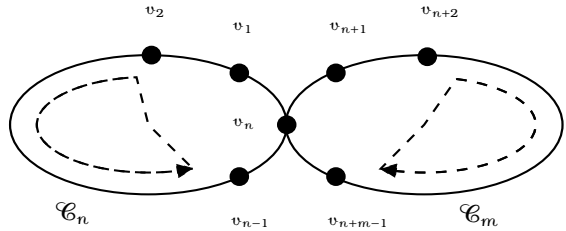
	Note that the total number of vertices is $\n+\m -1$, while total number of edges in $\n+\m$. Also, it is worthwile to mention that the vertices of $\C_{\n}$ are indexed in an anti-clockwise direction, while vertices of $\C_{\m}$ are indexed clockwise. These graphs are also called \say{\textit{infinity-graphs}} \cite{you2014maximal}.
	\item $\C_{\n,\rr,\m}$ obtained after joining two disjoint cycles, $\C_{\n}$ and $\C_{\m}$, by a path $\PP_{\rr}( \rr \geq 1)$ of length $\rr$. Let us consider the following figure. 
	
	\begin{figure}[H]
		\centering
		\tikzset{every picture/.style={line width=0.95pt}} 
		\resizebox{7.5cm}{3.3cm}{%
				\begin{tikzpicture}[x=0.75pt,y=0.75pt,yscale=-1,xscale=1]
				
				\draw  [fill={rgb, 255:red, 0; green, 0; blue, 0 }  ,fill opacity=1 ] (170.59,67.97) .. controls (170.59,65.12) and (172.79,62.81) .. (175.51,62.81) .. controls (178.23,62.81) and (180.43,65.12) .. (180.43,67.97) .. controls (180.43,70.82) and (178.23,73.13) .. (175.51,73.13) .. controls (172.79,73.13) and (170.59,70.82) .. (170.59,67.97) -- cycle ;
				\draw  [fill={rgb, 255:red, 0; green, 0; blue, 0 }  ,fill opacity=1 ] (225.91,123.92) .. controls (225.91,121.07) and (228.11,118.76) .. (230.83,118.76) .. controls (233.55,118.76) and (235.76,121.07) .. (235.76,123.92) .. controls (235.76,126.77) and (233.55,129.08) .. (230.83,129.08) .. controls (228.11,129.08) and (225.91,126.77) .. (225.91,123.92) -- cycle ;
				\draw  [fill={rgb, 255:red, 0; green, 0; blue, 0 }  ,fill opacity=1 ] (213.59,88.97) .. controls (213.59,86.12) and (215.79,83.81) .. (218.51,83.81) .. controls (221.23,83.81) and (223.43,86.12) .. (223.43,88.97) .. controls (223.43,91.82) and (221.23,94.13) .. (218.51,94.13) .. controls (215.79,94.13) and (213.59,91.82) .. (213.59,88.97) -- cycle ;
				\draw   (119,123.92) .. controls (119,93.03) and (144.03,68) .. (174.92,68) .. controls (205.8,68) and (230.83,93.03) .. (230.83,123.92) .. controls (230.83,154.8) and (205.8,179.83) .. (174.92,179.83) .. controls (144.03,179.83) and (119,154.8) .. (119,123.92) -- cycle ;
				\draw  [fill={rgb, 255:red, 0; green, 0; blue, 0 }  ,fill opacity=1 ] (208.59,163.97) .. controls (208.59,161.12) and (210.79,158.81) .. (213.51,158.81) .. controls (216.23,158.81) and (218.43,161.12) .. (218.43,163.97) .. controls (218.43,166.82) and (216.23,169.13) .. (213.51,169.13) .. controls (210.79,169.13) and (208.59,166.82) .. (208.59,163.97) -- cycle ;
				\draw  [fill={rgb, 255:red, 0; green, 0; blue, 0 }  ,fill opacity=1 ] (468.17,68.67) .. controls (468.17,71.52) and (465.97,73.83) .. (463.25,73.83) .. controls (460.53,73.83) and (458.33,71.52) .. (458.33,68.67) .. controls (458.33,65.83) and (460.53,63.52) .. (463.25,63.52) .. controls (465.97,63.52) and (468.17,65.83) .. (468.17,68.67) -- cycle ;
				\draw  [fill={rgb, 255:red, 0; green, 0; blue, 0 }  ,fill opacity=1 ] (407.84,123.72) .. controls (407.84,126.57) and (405.64,128.88) .. (402.92,128.88) .. controls (400.2,128.88) and (398,126.57) .. (398,123.72) .. controls (398,120.88) and (400.2,118.57) .. (402.92,118.57) .. controls (405.64,118.57) and (407.84,120.88) .. (407.84,123.72) -- cycle ;
				\draw  [fill={rgb, 255:red, 0; green, 0; blue, 0 }  ,fill opacity=1 ] (420.17,158.67) .. controls (420.17,161.52) and (417.97,163.83) .. (415.25,163.83) .. controls (412.53,163.83) and (410.33,161.52) .. (410.33,158.67) .. controls (410.33,155.83) and (412.53,153.52) .. (415.25,153.52) .. controls (417.97,153.52) and (420.17,155.83) .. (420.17,158.67) -- cycle ;
				\draw   (514.76,123.72) .. controls (514.76,154.61) and (489.72,179.64) .. (458.84,179.64) .. controls (427.96,179.64) and (402.92,154.61) .. (402.92,123.72) .. controls (402.92,92.84) and (427.96,67.81) .. (458.84,67.81) .. controls (489.72,67.81) and (514.76,92.84) .. (514.76,123.72) -- cycle ;
				\draw  [fill={rgb, 255:red, 0; green, 0; blue, 0 }  ,fill opacity=1 ] (425.17,83.67) .. controls (425.17,86.52) and (422.97,88.83) .. (420.25,88.83) .. controls (417.53,88.83) and (415.33,86.52) .. (415.33,83.67) .. controls (415.33,80.83) and (417.53,78.52) .. (420.25,78.52) .. controls (422.97,78.52) and (425.17,80.83) .. (425.17,83.67) -- cycle ;
				\draw  [draw opacity=0][dash pattern={on 4.5pt off 4.5pt}] (171.39,165.19) .. controls (150.17,163.4) and (133.5,145.6) .. (133.5,123.92) .. controls (133.5,106.83) and (143.85,92.16) .. (158.62,85.83) -- (174.92,123.92) -- cycle ; \draw [dash pattern={on 4.5pt off 4.5pt}] [dash pattern={on 4.5pt off 4.5pt}]  (168.25,164.8) .. controls (148.54,161.61) and (133.5,144.52) .. (133.5,123.92) .. controls (133.5,106.83) and (143.85,92.16) .. (158.62,85.83) ;  \draw [shift={(171.39,165.19)}, rotate = 193.55] [fill={rgb, 255:red, 0; green, 0; blue, 0 }  ][dash pattern={on 3.49pt off 4.5pt}][line width=0.08]  [draw opacity=0] (8.93,-4.29) -- (0,0) -- (8.93,4.29) -- cycle    ;
				\draw  [draw opacity=0][dash pattern={on 4.5pt off 4.5pt}] (486.52,93.04) .. controls (495,100.63) and (500.33,111.65) .. (500.33,123.92) .. controls (500.33,146.79) and (481.79,165.33) .. (458.92,165.33) .. controls (457.93,165.33) and (456.96,165.3) .. (455.99,165.23) -- (458.92,123.92) -- cycle ; \draw [dash pattern={on 4.5pt off 4.5pt}] [dash pattern={on 4.5pt off 4.5pt}]  (486.52,93.04) .. controls (495,100.63) and (500.33,111.65) .. (500.33,123.92) .. controls (500.33,146.79) and (481.79,165.33) .. (458.92,165.33) ; \draw [shift={(455.99,165.23)}, rotate = 355.75] [fill={rgb, 255:red, 0; green, 0; blue, 0 }  ][dash pattern={on 3.49pt off 4.5pt}][line width=0.08]  [draw opacity=0] (8.93,-4.29) -- (0,0) -- (8.93,4.29) -- cycle    ; 
				\draw  [fill={rgb, 255:red, 0; green, 0; blue, 0 }  ,fill opacity=1 ] (270.91,116) .. controls (270.91,113.15) and (273.11,110.84) .. (275.83,110.84) .. controls (278.55,110.84) and (280.76,113.15) .. (280.76,116) .. controls (280.76,118.85) and (278.55,121.16) .. (275.83,121.16) .. controls (273.11,121.16) and (270.91,118.85) .. (270.91,116) -- cycle ;
				\draw    (230.83,123.92) .. controls (260.83,104) and (243.83,137) .. (275.83,116) ;
				\draw  [fill={rgb, 255:red, 0; green, 0; blue, 0 }  ,fill opacity=1 ] (350.91,131) .. controls (350.91,128.15) and (353.11,125.84) .. (355.83,125.84) .. controls (358.55,125.84) and (360.76,128.15) .. (360.76,131) .. controls (360.76,133.85) and (358.55,136.16) .. (355.83,136.16) .. controls (353.11,136.16) and (350.91,133.85) .. (350.91,131) -- cycle ;
				\draw    (355.83,131) .. controls (388.83,112) and (368.83,144) .. (402.92,123.72) ;
				\draw  [dash pattern={on 2.25pt off 2.25pt}]  (275.83,116) .. controls (315.83,86) and (315.83,161) .. (355.83,131) ;
				
				\draw (223,69.4) node [anchor=north west][inner sep=0.75pt]  [font=\footnotesize]  {$\vv_{1}$};
				\draw (179,46.4) node [anchor=north west][inner sep=0.75pt]  [font=\footnotesize]  {$\vv_{2}$};
				\draw (219,170.4) node [anchor=north west][inner sep=0.75pt]  [font=\footnotesize]  {$\vv_{\n-1}$};
				\draw (205,117.4) node [anchor=north west][inner sep=0.75pt]  [font=\footnotesize]  {$\vv_{\n}$};
				\draw (264,88.4) node [anchor=north west][inner sep=0.75pt]  [font=\footnotesize]  {$\vv_{\n+1}$};
				\draw (340,103.4) node [anchor=north west][inner sep=0.75pt]  [font=\footnotesize]  {$\vv_{\n+\rr-1}$};
				\draw (413,117.4) node [anchor=north west][inner sep=0.75pt]  [font=\footnotesize]  {$\vv_{\n+\rr}$};
				\draw (378,64.4) node [anchor=north west][inner sep=0.75pt]  [font=\footnotesize]  {$\vv_{\n+\rr+1}$};
				\draw (425,46.4) node [anchor=north west][inner sep=0.75pt]  [font=\footnotesize]  {$\vv_{\n+\rr+2}$};
				\draw (380,167.4) node [anchor=north west][inner sep=0.75pt]  [font=\footnotesize]  {$\vv_{\n+\rr+\m-1}$};
				\draw (165,211.4) node [anchor=north west][inner sep=0.75pt]  [font=\Large]  {$\C_{\n}$};
				\draw (453,212.4) node [anchor=north west][inner sep=0.75pt]  [font=\Large]  {$\C_{\m}$};

			\end{tikzpicture}
			
		}
		\caption{bicyclic graph of type-II}
		\label{bicyclicgraphtypeii}
	\end{figure}
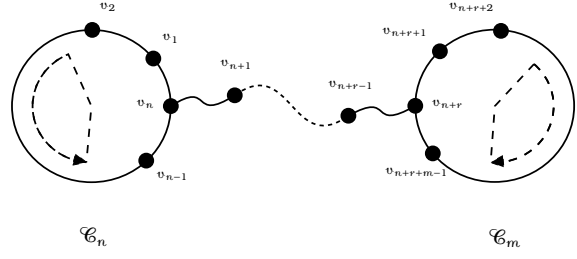
	We see that the path $\PP_{\rr}$ attaches the vertex $\vv_{\n}$ of $\C_{\n}$ to $\vv_{\n+\rr}$ of $\C_{\m}$. Note that usually a path of length $\rr$ has $\rr+1$ vertices but here, the number of vertices on the path are $\rr-1$ due to the inclusion of vertices $\vv_{\n}$ and $\vv_{\n+\rr}$. Moreover, the vertices of $\C_{\n}$ are indexed anticlockwise, while that of $\C_{\m}$ are indexed clockwise. These graphs are also known as \say{\textit{Kayak-Paddle graph}} \cite{ahmad2020computing}.
\end{enumerate}

In fact, Infinity-graph and Kayak Paddle Graph are special cases of \say{leafless cactus graph} defined as follows.

\begin{Definition}
	A graph $\G$ is a cactus graph if no two cycles of $\G$ share a common edge. It is a leafless cactus graph if it does not contain a vertex of degree one.
\end{Definition}

Metric dimension of bicyclic graphs of type-I and type-II have already been discussed in the literature.

\begin{Theorem} \cite{math11040869} \label{ownthm1}
\label{odd}Let $\C_{\n,\m}$ be a base bicyclic graph of type 1, $\n,\m\geq 3$.
Then,
\begin{equation*}
\beta(\C_{\n,\m})=
\begin{cases}
3 & \text{ when }\n,\m\text{ are even} \\
2 & \text{ otherwise}
\end{cases}
\end{equation*}
\end{Theorem} 
It is also worthwhile to mention that a metric basis set of $\C_{\n,\m}$ must contain a vertex from each cycle, since otherwise, we can easily find vertices which are not resolved by the metric basis set $\W$.
\begin{Theorem} \cite{math11040869}\label{ownthm2}
\label{oldthm2}Let $\C_{\n,\rr,\m}$ be a base bicyclic graph of type-II, $\n,\m\geq 3$
and $\rr\geq 1$. Then, $\beta \left( \C_{\n,\rr,\m}\right) =2.$
\end{Theorem}

\begin{Definition}
A resolving set $\SSS$ for a graph $\G$ is fault-tolerant resolving set if $\SSS \backslash \{\vv\}$ is also a resolving set for each $\vv$ in $\SSS$. The fault-tolerant metric dimension of $\G$ is the minimum cardinality of a fault-tolerant resolving set, and it will be denoted by $\beta ^{\prime }(\G)$. A fault-tolerant resolving set of order $\beta^{\prime }(\G)$ is called a fault-tolerant metric basis.
\end{Definition}


From the definition, we see that $\beta'(\G) \geq \beta(\G)+1$. Hernado et. al \cite{hernando2008fault} gave an upper bound for $\beta'(\G)$. More specifically:
\begin{Theorem}\label{oldthm3}
	Fault-tolerant metric dimension is bounded by a function of the metric dimension (independent of the graph). In particular, $\beta'(\G) \leq \beta(\G)\big(1+2 \cdot 5^{\beta(\G)-1} \big)$ for every graph $\G$.
\end{Theorem}

Using Theorem \ref{oldthm3} with Theorems \ref{ownthm1} and \ref{ownthm2}, we see that

\begin{equation*}
	\beta'(\G) \leq
	\begin{cases}
		153 & \text{ when }\G \simeq \C_{\n,\m} \text{ and }\n,\m\text{ are even} \\
		22 & \text{ when }\G \simeq \C_{\n,\m} \text{ and }\n,\m\text{ are any other values} \\
		22 & \text{ when }\G \simeq \C_{\n,\rr,\m} \text{ for all values of }\n,\rr, \m. 
	\end{cases}
\end{equation*}
These values are in fact highly inflated because of the generalized nature of the result. The actual values are much smaller and are discussed in the coming sections.

Javaid et. al. \cite{javaid2009fault} proved that $\beta'(\C_{\n})=3$. They also gave the condition under which the three vertices from $\C_{\n}$ form a fault tolerant metric basis set. Using this result we are ensured that $\beta'(\G) \geq 3$ when $\G$ is a bicyclic graph of type-I and II.

\section{Results On Metric Basis of Bicyclic Graphs}
Before proceeding further, we look at some results for metric basis of bicyclic graphs. These results will be helpful later on.
\begin{Lemma}\label{lem1}
	For a base bicyclic graph  $\C_{\n,\m}$ of type I, 
the common vertex, $\vv_{\n}\simeq \vv_{\n+\m}$, does not belong to any
metric basis set.
\end{Lemma}

\begin{proof}
We suppose on the contrary that $\vv_{\n}$ belongs to a metric basis set $\W$ of $%
\C_{\n,\m}.$ We discuss the following three cases.

\begin{mycases}
\case When $\m,\n$ are both odd.
	
By Theorem \ref{odd}, $\beta(\C_{\n,\m})=2$, i.e., $\W$ contains two vertices. Since $\vv_{\n} \in \W$, we need one other vertex to complete it. Without loss of generality, let $\vv_{\ii}\in C_{\n}$ be the other vertex of $\W$, i.e., $\W=\{\vv_{\ii},\vv_{\n}\}$. Let us consider two
vertices $\vv_{\aaa},\vv_{\bb}\in \C_{\m}$ which are equidistant from vertex $\vv_{\n}$, i. e., $\dd(\vv_{\aaa},\vv_{\n})=\dd(\vv_{\bb},\vv_{\n})$, see figure \ref{proof1figure}(A). Now
\begin{eqnarray*}
d(\vv_{\ii},\vv_{\aaa}) &=&d(\vv_{\ii},\vv_{\n})+\dd(\vv_{\n},\vv_{\aaa})\text{ and }d(\vv_{\ii},\vv_{\bb})=\dd(\vv_{\ii},\vv_{\n})+\dd(\vv_{\n},\vv_{\bb}) \\
&\Longrightarrow &\text{ }d(\vv_{\ii},\vv_{\aaa})=\dd(\vv_{\ii},\vv_{\bb} ).
\end{eqnarray*}

Therefore $r\left(\vv_{\aaa}|\W\right) =r\left( \vv_{\bb}|\W\right) $, a contradiction to the fact that $\W$ is a resolving set. Hence, $\vv_{\n} \notin \W$ for metric basis set $\W$ of $\C_{\n,\m}$.

\case When one of $\m,\n$ is odd and the other is even.

The proof is similar to the case above.

\case When $\m,\n$ are both even.

By Theorem \ref{odd}, $\beta(\C_{\n,\m})=3$. Since $\vv_{\n} \in \W$, we need only two more vertices for $\W$. Let $\vv_{\ii},\vv_{\jj}$ be the other two vertices of $\W$, then both $\vv_{\ii}, \vv_{\jj}$ can not be in $\C_{\n}$ or $\C_{\m}$, since we can easily find vertices in the other cycle which are not resolved in this case.

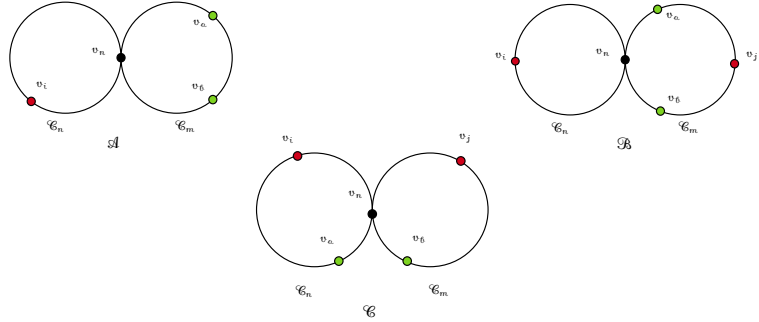
\begin{figure}[H]
	\begin{center}
			\tikzset{every picture/.style={line width=0.75pt}} 
				\resizebox{10cm}{4.2cm}{
			\begin{tikzpicture}[x=0.75pt,y=0.75pt,yscale=-1,xscale=1]
				
				\draw  [fill={rgb, 255:red, 0; green, 0; blue, 0 }  ,fill opacity=1 ] (99.41,60.11) .. controls (99.41,58.1) and (100.96,56.48) .. (102.88,56.48) .. controls (104.8,56.48) and (106.35,58.1) .. (106.35,60.11) .. controls (106.35,62.12) and (104.8,63.75) .. (102.88,63.75) .. controls (100.96,63.75) and (99.41,62.12) .. (99.41,60.11) -- cycle ;
				\draw   (102.88,60.11) .. controls (102.88,33.48) and (124.27,11.89) .. (150.66,11.89) .. controls (177.04,11.89) and (198.43,33.48) .. (198.43,60.11) .. controls (198.43,86.75) and (177.04,108.34) .. (150.66,108.34) .. controls (124.27,108.34) and (102.88,86.75) .. (102.88,60.11) -- cycle ;
				\draw   (7.33,60.11) .. controls (7.33,33.48) and (28.72,11.89) .. (55.11,11.89) .. controls (81.49,11.89) and (102.88,33.48) .. (102.88,60.11) .. controls (102.88,86.75) and (81.49,108.34) .. (55.11,108.34) .. controls (28.72,108.34) and (7.33,86.75) .. (7.33,60.11) -- cycle ;
				\draw  [fill={rgb, 255:red, 208; green, 2; blue, 27 }  ,fill opacity=1 ] (22.23,98.05) .. controls (22.23,96.04) and (23.79,94.41) .. (25.71,94.41) .. controls (27.63,94.41) and (29.18,96.04) .. (29.18,98.05) .. controls (29.18,100.06) and (27.63,101.68) .. (25.71,101.68) .. controls (23.79,101.68) and (22.23,100.06) .. (22.23,98.05) -- cycle ;
				\draw  [fill={rgb, 255:red, 126; green, 211; blue, 33 }  ,fill opacity=1 ] (178.36,23.58) .. controls (178.36,21.57) and (179.91,19.94) .. (181.83,19.94) .. controls (183.75,19.94) and (185.3,21.57) .. (185.3,23.58) .. controls (185.3,25.58) and (183.75,27.21) .. (181.83,27.21) .. controls (179.91,27.21) and (178.36,25.58) .. (178.36,23.58) -- cycle ;
				\draw  [fill={rgb, 255:red, 126; green, 211; blue, 33 }  ,fill opacity=1 ] (178.32,96.5) .. controls (178.32,94.49) and (179.87,92.86) .. (181.79,92.86) .. controls (183.71,92.86) and (185.26,94.49) .. (185.26,96.5) .. controls (185.26,98.51) and (183.71,100.14) .. (181.79,100.14) .. controls (179.87,100.14) and (178.32,98.51) .. (178.32,96.5) -- cycle ;
				\draw  [fill={rgb, 255:red, 0; green, 0; blue, 0 }  ,fill opacity=1 ] (532.72,61.86) .. controls (532.72,59.86) and (534.26,58.24) .. (536.16,58.24) .. controls (538.06,58.24) and (539.6,59.86) .. (539.6,61.86) .. controls (539.6,63.86) and (538.06,65.48) .. (536.16,65.48) .. controls (534.26,65.48) and (532.72,63.86) .. (532.72,61.86) -- cycle ;
				\draw   (536.16,61.86) .. controls (536.16,35.37) and (557.35,13.89) .. (583.49,13.89) .. controls (609.62,13.89) and (630.81,35.37) .. (630.81,61.86) .. controls (630.81,88.35) and (609.62,109.83) .. (583.49,109.83) .. controls (557.35,109.83) and (536.16,88.35) .. (536.16,61.86) -- cycle ;
				\draw   (441.51,61.86) .. controls (441.51,35.37) and (462.7,13.89) .. (488.84,13.89) .. controls (514.97,13.89) and (536.16,35.37) .. (536.16,61.86) .. controls (536.16,88.35) and (514.97,109.83) .. (488.84,109.83) .. controls (462.7,109.83) and (441.51,88.35) .. (441.51,61.86) -- cycle ;
				\draw  [fill={rgb, 255:red, 208; green, 2; blue, 27 }  ,fill opacity=1 ] (438.45,63.16) .. controls (438.45,61.34) and (439.93,59.87) .. (441.75,59.87) .. controls (443.56,59.87) and (445.04,61.34) .. (445.04,63.16) .. controls (445.04,64.98) and (443.56,66.45) .. (441.75,66.45) .. controls (439.93,66.45) and (438.45,64.98) .. (438.45,63.16) -- cycle ;
				\draw  [fill={rgb, 255:red, 208; green, 2; blue, 27 }  ,fill opacity=1 ] (626.72,65.36) .. controls (626.72,63.37) and (628.26,61.75) .. (630.16,61.75) .. controls (632.06,61.75) and (633.6,63.37) .. (633.6,65.36) .. controls (633.6,67.36) and (632.06,68.98) .. (630.16,68.98) .. controls (628.26,68.98) and (626.72,67.36) .. (626.72,65.36) -- cycle ;
				\draw  [fill={rgb, 255:red, 126; green, 211; blue, 33 }  ,fill opacity=1 ] (560.57,18.1) .. controls (560.57,16.11) and (562.11,14.49) .. (564.01,14.49) .. controls (565.91,14.49) and (567.45,16.11) .. (567.45,18.1) .. controls (567.45,20.1) and (565.91,21.72) .. (564.01,21.72) .. controls (562.11,21.72) and (560.57,20.1) .. (560.57,18.1) -- cycle ;
				\draw  [fill={rgb, 255:red, 126; green, 211; blue, 33 }  ,fill opacity=1 ] (563.04,106.48) .. controls (563.04,104.48) and (564.58,102.86) .. (566.48,102.86) .. controls (568.38,102.86) and (569.92,104.48) .. (569.92,106.48) .. controls (569.92,108.47) and (568.38,110.09) .. (566.48,110.09) .. controls (564.58,110.09) and (563.04,108.47) .. (563.04,106.48) -- cycle ;
				\draw  [fill={rgb, 255:red, 0; green, 0; blue, 0 }  ,fill opacity=1 ] (315.22,195.87) .. controls (315.22,193.82) and (316.84,192.15) .. (318.83,192.15) .. controls (320.83,192.15) and (322.45,193.82) .. (322.45,195.87) .. controls (322.45,197.93) and (320.83,199.6) .. (318.83,199.6) .. controls (316.84,199.6) and (315.22,197.93) .. (315.22,195.87) -- cycle ;
				\draw   (318.83,192.26) .. controls (318.83,164.99) and (341.11,142.88) .. (368.58,142.88) .. controls (396.06,142.88) and (418.33,164.99) .. (418.33,192.26) .. controls (418.33,219.54) and (396.06,241.64) .. (368.58,241.64) .. controls (341.11,241.64) and (318.83,219.54) .. (318.83,192.26) -- cycle ;
				\draw   (219.33,192.26) .. controls (219.33,164.99) and (241.61,142.88) .. (269.08,142.88) .. controls (296.56,142.88) and (318.83,164.99) .. (318.83,192.26) .. controls (318.83,219.54) and (296.56,241.64) .. (269.08,241.64) .. controls (241.61,241.64) and (219.33,219.54) .. (219.33,192.26) -- cycle ;
				\draw  [fill={rgb, 255:red, 208; green, 2; blue, 27 }  ,fill opacity=1 ] (250.98,145.45) .. controls (250.98,143.39) and (252.6,141.73) .. (254.6,141.73) .. controls (256.59,141.73) and (258.21,143.39) .. (258.21,145.45) .. controls (258.21,147.51) and (256.59,149.17) .. (254.6,149.17) .. controls (252.6,149.17) and (250.98,147.51) .. (250.98,145.45) -- cycle ;
				\draw  [fill={rgb, 255:red, 126; green, 211; blue, 33 }  ,fill opacity=1 ] (286.45,236.74) .. controls (286.45,234.69) and (288.07,233.02) .. (290.07,233.02) .. controls (292.07,233.02) and (293.68,234.69) .. (293.68,236.74) .. controls (293.68,238.8) and (292.07,240.47) .. (290.07,240.47) .. controls (288.07,240.47) and (286.45,238.8) .. (286.45,236.74) -- cycle ;
				\draw  [fill={rgb, 255:red, 126; green, 211; blue, 33 }  ,fill opacity=1 ] (345.23,236.74) .. controls (345.23,234.69) and (346.84,233.02) .. (348.84,233.02) .. controls (350.84,233.02) and (352.46,234.69) .. (352.46,236.74) .. controls (352.46,238.8) and (350.84,240.47) .. (348.84,240.47) .. controls (346.84,240.47) and (345.23,238.8) .. (345.23,236.74) -- cycle ;
				\draw  [fill={rgb, 255:red, 208; green, 2; blue, 27 }  ,fill opacity=1 ] (391.3,149.78) .. controls (391.3,147.72) and (392.92,146.06) .. (394.92,146.06) .. controls (396.92,146.06) and (398.54,147.72) .. (398.54,149.78) .. controls (398.54,151.84) and (396.92,153.51) .. (394.92,153.51) .. controls (392.92,153.51) and (391.3,151.84) .. (391.3,149.78) -- cycle ;
				
				\draw (76.05,49.73) node [anchor=north west][inner sep=0.75pt]  [font=\footnotesize]  {$\vv_{\n}$};
				\draw (37.71,113.45) node [anchor=north west][inner sep=0.75pt]  [font=\small]  {$\C_{\n}$};
				\draw (147.85,113.02) node [anchor=north west][inner sep=0.75pt]  [font=\small]  {$\C_{\m}$};
				\draw (162.97,25.62) node [anchor=north west][inner sep=0.75pt]  [font=\footnotesize]  {$\vv_{\aaa}$};
				\draw (162.07,81.12) node [anchor=north west][inner sep=0.75pt]  [font=\footnotesize]  {$\vv_{\bb}$};
				\draw (28.05,79.14) node [anchor=north west][inner sep=0.75pt]  [font=\footnotesize]  {$\vv_{\ii}$};
				\draw (89.79,127.06) node [anchor=north west][inner sep=0.75pt]  [font=\large]  {$\mathcal{A}$};
				\draw (509.21,50.95) node [anchor=north west][inner sep=0.75pt]  [font=\footnotesize]  {$\vv_{\n}$};
				\draw (471.32,114.66) node [anchor=north west][inner sep=0.75pt]  [font=\small]  {$\C_{\n}$};
				\draw (580.42,114.22) node [anchor=north west][inner sep=0.75pt]  [font=\small]  {$\C_{\m}$};
				\draw (638.36,53.63) node [anchor=north west][inner sep=0.75pt]  [font=\footnotesize]  {$\vv_{\jj}$};
				\draw (571.11,91.82) node [anchor=north west][inner sep=0.75pt]  [font=\footnotesize]  {$\vv_{\bb}$};
				\draw (422.65,54.54) node [anchor=north west][inner sep=0.75pt]  [font=\footnotesize]  {$\vv_{\ii}$};
				\draw (527.65,128.2) node [anchor=north west][inner sep=0.75pt]  [font=\large]  {$\mathcal{B}$};
				\draw (569.61,19.14) node [anchor=north west][inner sep=0.75pt]  [font=\footnotesize]  {$\vv_{\aaa}$};
				\draw (296.64,174.32) node [anchor=north west][inner sep=0.75pt]  [font=\footnotesize]  {$\vv_{\n}$};
				\draw (251.19,256.43) node [anchor=north west][inner sep=0.75pt]  [font=\small]  {$\C_{\n}$};
				\draw (365.96,255.98) node [anchor=north west][inner sep=0.75pt]  [font=\small]  {$\C_{\m}$};
				\draw (351.54,213.29) node [anchor=north west][inner sep=0.75pt]  [font=\footnotesize]  {$\vv_{\bb}$};
				\draw (239.58,126) node [anchor=north west][inner sep=0.75pt]  [font=\footnotesize]  {$\vv_{\ii}$};
				\draw (309.23,273.71) node [anchor=north west][inner sep=0.75pt]  [font=\large]  {$\C$};
				\draw (272.25,213.95) node [anchor=north west][inner sep=0.75pt]  [font=\footnotesize]  {$\vv_{\aaa}$};
				\draw (391.79,124.42) node [anchor=north west][inner sep=0.75pt]  [font=\footnotesize]  {$\vv_{\jj}$};

			\end{tikzpicture}

}
\caption{Relationship of common vertex $\vv_{\n} =\vv_{\n+\m}$ to metric basis of base bicyclic graph $\C_{\n,\m}$ of type-I.}
		\label{proof1figure}
	\end{center}
\end{figure}

Without loss of generality, let $\vv_{\ii} \in \C_{\n}$ and $\vv_{\jj} \in \C_{\m}$. If any of $\vv_{\ii}, \vv_{\jj}$ are middle labeled vertices of their respective cycles, i.e., $\vv_{\ii} \simeq \vv_{\frac{\n}{2}}$ or $\vv_{\jj} \simeq \vv_{\m+\frac{\n}{2}}$, we can easily find two vertices equidistant from $\vv_{\n}$, say $\vv_{\aaa},\vv_{\bb}$ in the same cycle, which are not resolved by $\W$. This argument can be easily visualized from Figure \ref{proof1figure}(B).

The above argument ensures that $\vv_{\ii}, \vv_{\jj}$ are not the middle labeled vertices of their respective cycle. Now, $\vv_{\ii} \in \C_{\n}$ and $\vv_{\ii} \not\simeq  \vv_{\frac{\n}{2}}$ ensures that $\vv_{\ii}$ lies in one half of $\C_{\n}$ and a similar argument shows that $\vv_{\jj}$ lies in one half of $\C_{\m}$. Let us consider two vertices $\vv_{\aaa},\vv_{\bb}$ equidistant from $\vv_{\n}$, such that $\vv_{\aaa}$ is in the opposite half of $\C_{\n}$ as compared to $\vv_{\ii}$, and the shortest $\vv_{\ii}\vv_{\aaa}$ path passes through $\vv_{\n}$. Similarly, we take $\vv_{\bb}$ from the opposite half of $\C_{\m}$ as compared to $\vv_{\jj}$, and the shortest $\vv_{\jj}\vv_{\bb}$ path passes through $\vv_{\n}$. This concept is illustrated in Figure \ref{proof1figure}(C). It can be easily concluded now that $\vv_{\aaa}, \vv_{\bb}$ are not resolved by $\W$. All these cases ensure that $\vv_{\n}$ can not be in any metric basis.\qedhere
\end{mycases}
\end{proof}

\begin{Lemma}\label{newlem}
	For bicyclic graph $\C_{\n,\m}$ where $\n$ is even and $\m$ is odd (resp. $\n$ is odd and $\m$ is even), the vertex $\vv_{\frac{\n}{2}} \big(\text{resp. } \vv_{\n+\frac{\m}{2}}\big)$ does not belong to any metric basis set $\W$.
\end{Lemma}

\begin{proof}
	Let us suppose on the contrary that $\vv_{\frac{\n}{2}} \in \W$ for some resolving set $\W$ of $\C_{\n,\m}$, where $\n,\m$ are as given above. By Theorem \ref{odd}, $\beta(\C_{\n,\m})=2$. Let $\vv_{\jj} \in \C_{\m}$ be the second vertex in $\W$. Consider two vertices $\vv_{\aaa}$, $\vv_{\bb}$ equidistant from $\vv_{\n}$, then they are equidistant from $\vv_{\frac{\n}{2}}$, since $\n$ is even. It can now be easily shown that these vertices are not resolved by $\W$, giving us the required contradiction.
\end{proof}

In the event when both $\n,\m$ are even, the result is a bit different in the sense that both middle vertices can not be in the same metric basis set. This is proved in the following.

\begin{Lemma} \label{11}
	Let $\W$ be any metric basis for $\C_{\n,\m}$, where $\n,\m$ are even, then, $\left \{\vv_{\frac{\n}{2}}, \vv_{\n+\frac{\m}{2}}\right \} \not\subset \W $.
\end{Lemma}

\begin{proof}
	Let $\W$ be a metric basis for $\C_{\n,\m}$, where $\n,\m$ are even, and let us suppose on the contrary that $\left \{ \vv_{\frac{\n}{2}}, \vv_{\frac{\n+m}{2}}\right\} \subset \W$. By Theorem \ref{odd}, $\beta(\C_{\n,\m})=3$. Without loss of generality, let $\vv_{\jj}\in \C_{\n}$ be the the third element of $\W$, then $\W=\left\{\vv_{\jj}, \vv_{\frac{\n}{2}}, \vv_{\frac{\n+m}{2}}\right\}.$ Suppose $\dd(\vv_{\jj}, \vv_{\frac{\n}{2}})=k$, then $\dd(\vv_{\jj},\vv_{\n})=\frac{\n}{2}-k.$ Now consider two more vertices $\vv_{\aaa},\vv_{\bb}\in \C_{\m}$ which are equidistant from $\vv_{\n}$. Let $\dd(\vv_{\aaa},\vv_{\n})=\dd(\vv_{\bb},\vv_{\n})=l$. Then $\dd\left(\vv_{\aaa}, \vv_{\frac{\n+m}{2}}\right)=\dd\left(\vv_{\bb}, \vv_{\frac{\n+\m}{2}}\right)=\frac{\m}{2}-l$. These concepts are visually illustrated in Figure \ref{proof2figure}. Now, it is a simple argument to show that $r(\vv_{\aaa}|\W)=r(\vv_{\bb}|\W)$. This is a contradiction to the fact that $\W$ is a metric basis and hence, $\left \{ \vv_{\frac{\n}{2}}, \vv_{\frac{\n+m}{2}}\right \} \not\subset \W$ for any metric basis set $\W$ of $\C_{\n,\m}.$\qedhere

	\begin{figure}[H]
		\begin{center}
			\tikzset{every picture/.style={line width=0.75pt}} 
			\resizebox{6cm}{3.0cm}{%
\begin{tikzpicture}[x=0.75pt,y=0.75pt,yscale=-1,xscale=1]
	
	\draw  [fill={rgb, 255:red, 0; green, 0; blue, 0 }  ,fill opacity=1 ] (361.98,142.29) .. controls (361.98,139.44) and (364.18,137.13) .. (366.9,137.13) .. controls (369.62,137.13) and (371.82,139.44) .. (371.82,142.29) .. controls (371.82,145.14) and (369.62,147.45) .. (366.9,147.45) .. controls (364.18,147.45) and (361.98,145.14) .. (361.98,142.29) -- cycle ;
	\draw   (366.9,142.29) .. controls (366.9,104.51) and (397.22,73.89) .. (434.62,73.89) .. controls (472.02,73.89) and (502.33,104.51) .. (502.33,142.29) .. controls (502.33,180.06) and (472.02,210.69) .. (434.62,210.69) .. controls (397.22,210.69) and (366.9,180.06) .. (366.9,142.29) -- cycle ;
	\draw   (231.46,142.29) .. controls (231.46,104.51) and (261.78,73.89) .. (299.18,73.89) .. controls (336.58,73.89) and (366.9,104.51) .. (366.9,142.29) .. controls (366.9,180.06) and (336.58,210.69) .. (299.18,210.69) .. controls (261.78,210.69) and (231.46,180.06) .. (231.46,142.29) -- cycle ;
	\draw  [fill={rgb, 255:red, 208; green, 2; blue, 27 }  ,fill opacity=1 ] (290.59,211.09) .. controls (290.59,208.24) and (292.79,205.93) .. (295.51,205.93) .. controls (298.23,205.93) and (300.43,208.24) .. (300.43,211.09) .. controls (300.43,213.94) and (298.23,216.25) .. (295.51,216.25) .. controls (292.79,216.25) and (290.59,213.94) .. (290.59,211.09) -- cycle ;
	\draw  [fill={rgb, 255:red, 126; green, 211; blue, 33 }  ,fill opacity=1 ] (386.88,90.47) .. controls (386.88,87.62) and (389.09,85.31) .. (391.81,85.31) .. controls (394.53,85.31) and (396.73,87.62) .. (396.73,90.47) .. controls (396.73,93.31) and (394.53,95.62) .. (391.81,95.62) .. controls (389.09,95.62) and (386.88,93.31) .. (386.88,90.47) -- cycle ;
	\draw  [fill={rgb, 255:red, 126; green, 211; blue, 33 }  ,fill opacity=1 ] (385.82,193.9) .. controls (385.82,191.05) and (388.03,188.74) .. (390.74,188.74) .. controls (393.46,188.74) and (395.67,191.05) .. (395.67,193.9) .. controls (395.67,196.75) and (393.46,199.06) .. (390.74,199.06) .. controls (388.03,199.06) and (385.82,196.75) .. (385.82,193.9) -- cycle ;
	\draw  [fill={rgb, 255:red, 0; green, 0; blue, 0 }  ,fill opacity=1 ] (226.46,142.29) .. controls (226.46,139.44) and (228.67,137.13) .. (231.39,137.13) .. controls (234.1,137.13) and (236.31,139.44) .. (236.31,142.29) .. controls (236.31,145.14) and (234.1,147.45) .. (231.39,147.45) .. controls (228.67,147.45) and (226.46,145.14) .. (226.46,142.29) -- cycle ;
	\draw  [fill={rgb, 255:red, 0; green, 0; blue, 0 }  ,fill opacity=1 ] (497.46,140.29) .. controls (497.46,137.44) and (499.67,135.13) .. (502.39,135.13) .. controls (505.1,135.13) and (507.31,137.44) .. (507.31,140.29) .. controls (507.31,143.14) and (505.1,145.45) .. (502.39,145.45) .. controls (499.67,145.45) and (497.46,143.14) .. (497.46,140.29) -- cycle ;
	\draw  [draw opacity=0] (296.13,195.55) .. controls (267.76,193.98) and (245.25,170.73) .. (245.25,142.29) .. controls (245.25,142.19) and (245.25,142.1) .. (245.25,142) -- (299.18,142.29) -- cycle ; \draw    (296.13,195.55) .. controls (267.76,193.98) and (245.25,170.73) .. (245.25,142.29) .. controls (245.25,142.19) and (245.25,142.1) .. (245.25,142) ; \draw [shift={(245.25,142)}, rotate = 83.85] [color={rgb, 255:red, 0; green, 0; blue, 0 }  ][line width=0.75]    (0,5.59) -- (0,-5.59)   ; \draw [shift={(296.13,195.55)}, rotate = 9.78] [color={rgb, 255:red, 0; green, 0; blue, 0 }  ][line width=0.75]    (0,5.59) -- (0,-5.59)   ;
	\draw  [draw opacity=0] (380.76,145.14) .. controls (380.71,144.19) and (380.68,143.24) .. (380.68,142.29) .. controls (380.68,124.68) and (389.3,109.07) .. (402.6,99.35) -- (434.62,142.29) -- cycle ; \draw    (380.76,145.14) .. controls (380.71,144.19) and (380.68,143.24) .. (380.68,142.29) .. controls (380.68,124.68) and (389.3,109.07) .. (402.6,99.35) ; \draw [shift={(402.6,99.35)}, rotate = 137.84] [color={rgb, 255:red, 0; green, 0; blue, 0 }  ][line width=0.75]    (0,5.59) -- (0,-5.59)   ; \draw [shift={(380.76,145.14)}, rotate = 93.32] [color={rgb, 255:red, 0; green, 0; blue, 0 }  ][line width=0.75]    (0,5.59) -- (0,-5.59)   ;
	
	\draw (331.11,130.31) node [anchor=north west][inner sep=0.75pt]  [font=\footnotesize]  {$\vv_{\n}$};
	\draw (268.76,32.69) node [anchor=north west][inner sep=0.75pt]  [font=\large]  {$\C_{\n}$};
	\draw (425.51,32.08) node [anchor=north west][inner sep=0.75pt]  [font=\large]  {$\C_{\m}$};
	\draw (365.32,64.26) node [anchor=north west][inner sep=0.75pt]  [font=\footnotesize]  {$\vv_{\aaa}$};
	\draw (372.36,199.13) node [anchor=north west][inner sep=0.75pt]  [font=\footnotesize]  {$\vv_{\bb}$};
	\draw (302.87,210.66) node [anchor=north west][inner sep=0.75pt]  [font=\footnotesize]  {$\vv_{\ii}$};
	\draw (194.32,127.26) node [anchor=north west][inner sep=0.75pt]  [font=\footnotesize]  {$\vv_{\frac{\n}{2}}$};
	\draw (523.32,127.26) node [anchor=north west][inner sep=0.75pt]  [font=\footnotesize]  {$\vv_{\cfrac{\n+m}{2}}$};
	\draw (263,156.4) node [anchor=north west][inner sep=0.75pt]    {$\kk$};
	\draw (396,116.4) node [anchor=north west][inner sep=0.75pt]    {$\el$};

\end{tikzpicture}

			}
			\caption{Relationship of $\left \{ \vv_{\frac{\n}{2}}, \vv_{\frac{\n+m}{2}}\right\}$ and metric basis of base bicyclic graph of type-I for even $\n,\m$.}
			\label{proof2figure}
		\end{center}
	\end{figure}
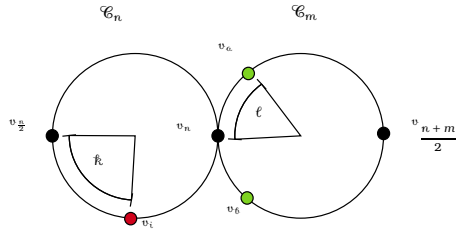
	
\end{proof}
Note that the vertex $\vv_{\frac{\n}{2}}$ or $\vv_{\n+\frac{\m}{2}}$ may individually appear in a metric basis set but they can not appear together.
\begin{Lemma} \label{oddoddmiddleall}
	Let $\C_{\n,\m}$ be a base bicyclic graph of type-I, where $\n,\m$ are odd, then $\W=\left \{ \vv_{\left \lfloor \frac{\n}{2}\right \rfloor },\vv_{\jj}\right \} $ $\left(resp. \ \W=\left \{ \vv_{\left \lfloor \frac{\n}{2}\right \rfloor+1 },\vv_{\jj}\right \} \right)$
	is always a metric basis of $\C_{\n,\m}$ for all $\vv_{\jj}\in \C_{\m}-\{\vv_{\n}\}$.
\end{Lemma}

\begin{proof}
	Let $\C_{\n,\m}$ be a base bicyclic graph of type-I, where $\n,\m$ are odd. By Lemma \ref{lem1}, $\vv_{\n} \notin \W$. To prove that $\W=\left \{ \vv_{\left \lfloor \frac{\n}{2}\right \rfloor},\vv_{\jj}\right \} $, $\vv_{\jj} \in \C_{\m}-\{\vv_{\n}\}$, is a metric basis of $\C_{\n,\m}$, we need to show that it resolves all pairs of vertices of $\C_{\n,\m}$. Let $\vv_{\aaa},\vv_{\bb} \in \V(\C_{\n,\m})(\vv_{\aaa} \neq \vv_{\bb}) $. If we take $\dd\left(\vv_{\aaa}, \vv_{\left \lfloor \frac{\n}{2}\right \rfloor}\right) \neq d\left(\vv_{\bb}, \vv_{\left \lfloor \frac{\n}{2}\right \rfloor}\right)$, there is nothing to prove, so we consider $\vv_{\aaa},\vv_{\bb}$ to be equidistant from $\vv_{\left \lfloor \frac{\n}{2}\right \rfloor}$. Since $\vv_{\aaa}, \vv_{\bb}$ are equidistant from $\vv_{\left \lfloor \frac{\n}{2}\right \rfloor}$ and $\n$ is odd, $\vv_{\aaa} \in \C_{\n}$ and $\vv_{\bb} \in \C_{\m}$ or vice versa, is not possible. We discuss the other two cases in the following.
	\begin{mycases}
		\case  When $\vv_{\aaa},\vv_{\bb}\in \C_{\n}$.
		
		Let $\dd\left(\vv_{\aaa}, \vv_{\left \lfloor \frac{\n}{2}\right \rfloor}\right)=\dd\left(\vv_{\bb}, \vv_{\left \lfloor \frac{\n}{2}\right \rfloor}\right)=k$. Let $\PP_{1}$ and $\PP_{2}$ be the two distinct $\vv_{\n} \vv_{\left \lfloor \frac{\n}{2}\right \rfloor}$-paths. Since $\n$ is odd, without loss of generality, we can consider $\left|\PP_{1}\right| < \left|\PP_{2}\right|$. Also, since $\vv_{\aaa},\vv_{\bb}$ are equidistant from $\vv_{\left \lfloor \frac{\n}{2}\right \rfloor}$, we conclude that if $\vv_{\aaa}$ lies on $\PP_1$ then $\vv_{\bb}$ lies on $\PP_2$ and vice versa. Let us suppose that $\vv_{\aaa}$ lies on $\PP_1$ and $\vv_{\bb}$ lies on $\PP_2$. It is obvious that $\vv_{\aaa} \vv_{\bb}$-path through $\vv_{\left \lfloor \frac{\n}{2}\right \rfloor}$ is of even length and $\vv_{\aaa} \vv_{\bb}$-path through $\vv_{\n}$ is of odd length. This argument ensures that $\dd(\vv_{\aaa},\vv_{\n}) \neq \dd(\vv_{\bb},\vv_{\n})$. These explanations can easily be seen in Figure \ref{proof3figure}. 
		
		\begin{figure}[H]
			\begin{center}
				\tikzset{every picture/.style={line width=0.75pt}} 
				\resizebox{6.4cm}{3.7cm}{
					\begin{tikzpicture}[x=0.75pt,y=0.75pt,yscale=-1,xscale=1]
	
	\draw  [fill={rgb, 255:red, 0; green, 0; blue, 0 }  ,fill opacity=1 ] (327.98,119.29) .. controls (327.98,116.44) and (330.18,114.13) .. (332.9,114.13) .. controls (335.62,114.13) and (337.82,116.44) .. (337.82,119.29) .. controls (337.82,122.14) and (335.62,124.45) .. (332.9,124.45) .. controls (330.18,124.45) and (327.98,122.14) .. (327.98,119.29) -- cycle ;
	\draw   (332.9,119.29) .. controls (332.9,81.51) and (363.22,50.89) .. (400.62,50.89) .. controls (438.02,50.89) and (468.33,81.51) .. (468.33,119.29) .. controls (468.33,157.06) and (438.02,187.69) .. (400.62,187.69) .. controls (363.22,187.69) and (332.9,157.06) .. (332.9,119.29) -- cycle ;
	\draw   (197.46,119.29) .. controls (197.46,81.51) and (227.78,50.89) .. (265.18,50.89) .. controls (302.58,50.89) and (332.9,81.51) .. (332.9,119.29) .. controls (332.9,157.06) and (302.58,187.69) .. (265.18,187.69) .. controls (227.78,187.69) and (197.46,157.06) .. (197.46,119.29) -- cycle ;
	\draw  [fill={rgb, 255:red, 126; green, 211; blue, 33 }  ,fill opacity=1 ] (229.59,181.09) .. controls (229.59,178.24) and (231.79,175.93) .. (234.51,175.93) .. controls (237.23,175.93) and (239.43,178.24) .. (239.43,181.09) .. controls (239.43,183.94) and (237.23,186.25) .. (234.51,186.25) .. controls (231.79,186.25) and (229.59,183.94) .. (229.59,181.09) -- cycle ;
	\draw  [fill={rgb, 255:red, 126; green, 211; blue, 33 }  ,fill opacity=1 ] (284.34,54.89) .. controls (284.34,52.04) and (286.54,49.73) .. (289.26,49.73) .. controls (291.98,49.73) and (294.18,52.04) .. (294.18,54.89) .. controls (294.18,57.74) and (291.98,60.05) .. (289.26,60.05) .. controls (286.54,60.05) and (284.34,57.74) .. (284.34,54.89) -- cycle ;
	\draw  [fill={rgb, 255:red, 0; green, 0; blue, 0 }  ,fill opacity=1 ] (199.46,88.29) .. controls (199.46,85.44) and (201.67,83.13) .. (204.39,83.13) .. controls (207.1,83.13) and (209.31,85.44) .. (209.31,88.29) .. controls (209.31,91.14) and (207.1,93.45) .. (204.39,93.45) .. controls (201.67,93.45) and (199.46,91.14) .. (199.46,88.29) -- cycle ;
	\draw  [fill={rgb, 255:red, 0; green, 0; blue, 0 }  ,fill opacity=1 ] (456.46,92.29) .. controls (456.46,89.44) and (458.67,87.13) .. (461.39,87.13) .. controls (464.1,87.13) and (466.31,89.44) .. (466.31,92.29) .. controls (466.31,95.14) and (464.1,97.45) .. (461.39,97.45) .. controls (458.67,97.45) and (456.46,95.14) .. (456.46,92.29) -- cycle ;
	\draw  [fill={rgb, 255:red, 0; green, 0; blue, 0 }  ,fill opacity=1 ] (200.46,151.29) .. controls (200.46,148.44) and (202.67,146.13) .. (205.39,146.13) .. controls (208.1,146.13) and (210.31,148.44) .. (210.31,151.29) .. controls (210.31,154.14) and (208.1,156.45) .. (205.39,156.45) .. controls (202.67,156.45) and (200.46,154.14) .. (200.46,151.29) -- cycle ;
	\draw  [fill={rgb, 255:red, 208; green, 2; blue, 27 }  ,fill opacity=1 ] (395.69,50.89) .. controls (395.69,48.04) and (397.9,45.73) .. (400.62,45.73) .. controls (403.33,45.73) and (405.54,48.04) .. (405.54,50.89) .. controls (405.54,53.74) and (403.33,56.05) .. (400.62,56.05) .. controls (397.9,56.05) and (395.69,53.74) .. (395.69,50.89) -- cycle ;
	\draw  [fill={rgb, 255:red, 0; green, 0; blue, 0 }  ,fill opacity=1 ] (457.46,145.29) .. controls (457.46,142.44) and (459.67,140.13) .. (462.39,140.13) .. controls (465.1,140.13) and (467.31,142.44) .. (467.31,145.29) .. controls (467.31,148.14) and (465.1,150.45) .. (462.39,150.45) .. controls (459.67,150.45) and (457.46,148.14) .. (457.46,145.29) -- cycle ;
	
	\draw (345.11,114.31) node [anchor=north west][inner sep=0.75pt]  [font=\footnotesize]  {$\vv_{\n}$};
	\draw (245.76,231.69) node [anchor=north west][inner sep=0.75pt]  [font=\large]  {$\C_{\n}$};
	\draw (402.51,232.08) node [anchor=north west][inner sep=0.75pt]  [font=\large]  {$\C_{\m}$};
	\draw (227.18,192.09) node [anchor=north west][inner sep=0.75pt]  [font=\footnotesize]  {$\vv_{\bb}$};
	\draw (157.32,77.26) node [anchor=north west][inner sep=0.75pt]  [font=\footnotesize]  {$\vv_{\lfloor \frac{\n}{2} \rfloor }$};
	\draw (476.32,67.26) node [anchor=north west][inner sep=0.75pt]  [font=\footnotesize]  {$\vv_{\n+\lfloor \frac{\m}{2} \rfloor }$};
	\draw (238,62.4) node [anchor=north west][inner sep=0.75pt]  [font=\footnotesize]  {$\PP_{1}$};
	\draw (277,163.4) node [anchor=north west][inner sep=0.75pt]  [font=\footnotesize]  {$\PP_{2}$};
	\draw (157.32,141.26) node [anchor=north west][inner sep=0.75pt]  [font=\footnotesize]  {$\vv_{\lfloor \frac{\n}{2} \rfloor +1}$};
	\draw (287,29.4) node [anchor=north west][inner sep=0.75pt]  [font=\footnotesize]  {$\vv_{\aaa}$};
	\draw (476.32,127.26) node [anchor=north west][inner sep=0.75pt]  [font=\footnotesize]  {$\vv_{\n+\lfloor \frac{\m}{2} \rfloor +1}$};
	\draw (407,20.4) node [anchor=north west][inner sep=0.75pt]  [font=\footnotesize]  {$\vv_{\jj}$};

\end{tikzpicture}

				}
				\setlength{\belowcaptionskip}{-20pt}
				\caption{$\vv_{\aaa},\vv_{\bb} \in \C_{\n}$ for base bicyclic graph $\C_{\n,\m}$ of type-I for odd $\n,\m$}
				\label{proof3figure}
			\end{center}
		\end{figure}
		\vspace{0.5cm}
		
		Now, $\dd(\vv_{\aaa},\vv_{\jj}) =\dd(\vv_{\aaa},\vv_{\n})+\dd(\vv_{\n},\vv_{\jj})\text{ and } 	d(\vv_{\bb},\vv_{\jj})=\dd(\vv_{\bb},\vv_{\n})+\dd(\vv_{\n},\vv_{\jj})$ implies that $\dd(\vv_{\aaa},\vv_{\jj})\neq \dd(\vv_{\bb},\vv_{\jj})$. Hence $\vv_{\jj}$ resolves the vertices $\vv_{\aaa}$ and $\vv_{\bb}$.
		\case When $\vv_{\aaa},\vv_{\bb} \in \C_{\m}$.	
		
		Since $\vv_{\aaa}, \vv_{\bb} \in \C_{\m}$ and are equidistant from $\vv_{\left\lfloor \frac{\n}{2} \right\rfloor }$, we can easily conclude that $\dd(\vv_{\aaa},\vv_{\n})=\dd(\vv_{\bb},\vv_{\n})$. Also, $\vv_{\aaa}$ and $\vv_{\bb}$ can not both lie on the same $\left(v_{\n} \vv_{\n+\left\lfloor \frac{\m}{2} \right\rfloor }\right)$-path. Let $\PP$ and $\PP'$ be these two distinct $\left(v_{\n} \vv_{\n+\left\lfloor \frac{\m}{2} \right\rfloor }\right)$-paths, such that $\vv_{\aaa}$ lies on $\PP$ and $\vv_{\bb}$ lies on $\PP'$, as shown in the figure. 
		
		\begin{figure}[H]
			\begin{center}
				\tikzset{every picture/.style={line width=0.75pt}} 
				\resizebox{6.4cm}{3.6cm}{					
					\begin{tikzpicture}[x=0.75pt,y=0.75pt,yscale=-1,xscale=1]
						
						\draw  [fill={rgb, 255:red, 0; green, 0; blue, 0 }  ,fill opacity=1 ] (315.98,135.29) .. controls (315.98,132.44) and (318.18,130.13) .. (320.9,130.13) .. controls (323.62,130.13) and (325.82,132.44) .. (325.82,135.29) .. controls (325.82,138.14) and (323.62,140.45) .. (320.9,140.45) .. controls (318.18,140.45) and (315.98,138.14) .. (315.98,135.29) -- cycle ;
						\draw   (320.9,135.29) .. controls (320.9,97.51) and (351.22,66.89) .. (388.62,66.89) .. controls (426.02,66.89) and (456.33,97.51) .. (456.33,135.29) .. controls (456.33,173.06) and (426.02,203.69) .. (388.62,203.69) .. controls (351.22,203.69) and (320.9,173.06) .. (320.9,135.29) -- cycle ;
						\draw   (185.46,135.29) .. controls (185.46,97.51) and (215.78,66.89) .. (253.18,66.89) .. controls (290.58,66.89) and (320.9,97.51) .. (320.9,135.29) .. controls (320.9,173.06) and (290.58,203.69) .. (253.18,203.69) .. controls (215.78,203.69) and (185.46,173.06) .. (185.46,135.29) -- cycle ;
						\draw  [fill={rgb, 255:red, 126; green, 211; blue, 33 }  ,fill opacity=1 ] (368.59,202.09) .. controls (368.59,199.24) and (370.79,196.93) .. (373.51,196.93) .. controls (376.23,196.93) and (378.43,199.24) .. (378.43,202.09) .. controls (378.43,204.94) and (376.23,207.25) .. (373.51,207.25) .. controls (370.79,207.25) and (368.59,204.94) .. (368.59,202.09) -- cycle ;
						\draw  [fill={rgb, 255:red, 126; green, 211; blue, 33 }  ,fill opacity=1 ] (368.34,69.89) .. controls (368.34,67.04) and (370.54,64.73) .. (373.26,64.73) .. controls (375.98,64.73) and (378.18,67.04) .. (378.18,69.89) .. controls (378.18,72.74) and (375.98,75.05) .. (373.26,75.05) .. controls (370.54,75.05) and (368.34,72.74) .. (368.34,69.89) -- cycle ;
						\draw  [fill={rgb, 255:red, 0; green, 0; blue, 0 }  ,fill opacity=1 ] (180.54,130.13) .. controls (180.54,127.28) and (182.75,124.97) .. (185.46,124.97) .. controls (188.18,124.97) and (190.39,127.28) .. (190.39,130.13) .. controls (190.39,132.98) and (188.18,135.29) .. (185.46,135.29) .. controls (182.75,135.29) and (180.54,132.98) .. (180.54,130.13) -- cycle ;
						\draw  [fill={rgb, 255:red, 0; green, 0; blue, 0 }  ,fill opacity=1 ] (449.46,121.29) .. controls (449.46,118.44) and (451.67,116.13) .. (454.39,116.13) .. controls (457.1,116.13) and (459.31,118.44) .. (459.31,121.29) .. controls (459.31,124.14) and (457.1,126.45) .. (454.39,126.45) .. controls (451.67,126.45) and (449.46,124.14) .. (449.46,121.29) -- cycle ;
						\draw  [fill={rgb, 255:red, 208; green, 2; blue, 27 }  ,fill opacity=1 ] (333.69,89.89) .. controls (333.69,87.04) and (335.9,84.73) .. (338.62,84.73) .. controls (341.33,84.73) and (343.54,87.04) .. (343.54,89.89) .. controls (343.54,92.74) and (341.33,95.05) .. (338.62,95.05) .. controls (335.9,95.05) and (333.69,92.74) .. (333.69,89.89) -- cycle ;
						
						\draw (288.11,130.31) node [anchor=north west][inner sep=0.75pt]  [font=\footnotesize]  {$\vv_{\n}$};
						\draw (233.76,247.69) node [anchor=north west][inner sep=0.75pt]  [font=\large]  {$\C_{\n}$};
						\draw (390.51,248.08) node [anchor=north west][inner sep=0.75pt]  [font=\large]  {$\C_{\m}$};
						\draw (367.18,216.09) node [anchor=north west][inner sep=0.75pt]  [font=\footnotesize]  {$\vv_{\bb}$};
						\draw (141.32,119.26) node [anchor=north west][inner sep=0.75pt]  [font=\footnotesize]  {$\vv_{\lfloor \frac{\n}{2} \rfloor }$};
						\draw (473.32,110.26) node [anchor=north west][inner sep=0.75pt]  [font=\footnotesize]  {$\vv_{\n+\lfloor \frac{\m}{2} \rfloor }$};
						\draw (368,37.4) node [anchor=north west][inner sep=0.75pt]  [font=\footnotesize]  {$\vv_{\aaa}$};
						\draw (333,57.4) node [anchor=north west][inner sep=0.75pt]  [font=\footnotesize]  {$\vv_{\jj}$};
						\draw (412,83.4) node [anchor=north west][inner sep=0.75pt]  [font=\large]  {$\PP$};
						\draw (349,166.4) node [anchor=north west][inner sep=0.75pt]  [font=\large]  {$\PP'$};

					\end{tikzpicture}

				}
				\setlength{\belowcaptionskip}{-20pt}
				\caption{$\vv_{\aaa},\vv_{\bb} \in \C_{\m}$ for base bicyclic graph $\C_{\n,\m}$ of type-I for odd $\n,\m$ and $\vv_{\jj}$ on $\PP_{\uu\vv}$ through $\vv_{\n}$}
				\label{proof4figure}
			\end{center}
		\end{figure}	
		Since $\C_{\m}$ is an odd cycle, we can easily conclude that $\left|\PP\right|\neq \left|\PP'\right|$. Since $\vv_{\jj}$ can be any vertex of $\C_{\m}$, without loss of generality, we suppose that $\vv_{\jj}$ lies on $\PP$. Following two possibilities arise for the placement of $\vv_{\jj}$.
		\begin{subcases}
			\subcase When $\vv_{\jj}$ lies on the path $\vv_{\aaa} \vv_{\bb}$-path through $\vv_{\n}$. This possibility is evident in Figure \ref{proof4figure}.
			
			\begin{subsubcases}
				\subsubcase If the shortest $\vv_{\jj} \vv_{\bb}$-path passes through $\vv_{\n}$, we get \[d(\vv_{\jj},\vv_{\aaa})=\dd(\vv_{\n},\vv_{\aaa})-\dd(\vv_{\n},\vv_{\jj}) \text{ and } \dd(\vv_{\jj},\vv_{\bb})=\dd(\vv_{\n},\vv_{\bb})+\dd(\vv_{\n},\vv_{\jj}).\] Since $\dd(\vv_{\n},\vv_{\aaa})=\dd(\vv_{\n},\vv_{\bb})$, we can easily conclude that $\dd(\vv_{\jj},\vv_{\aaa}) \neq \dd(\vv_{\jj},\vv_{\bb})$. Hence, $\vv_{\jj}$ resolves $\vv_{\aaa}$ and $\vv_{\bb}$.
				\subsubcase If the shortest $\vv_{\jj} \vv_{\bb}$-path passes through $\vv_{\n+\left\lfloor \frac{\m}{2}\right\rfloor}$, we get $\dd(\vv_{\jj},\vv_{\bb})=\dd(\vv_{\jj},\vv_{\aaa})+\dd(\vv_{\aaa},\vv_{\bb})$, again ensuring that $\vv_{\jj}$ resolves $\vv_{\aaa}$ and $\vv_{\bb}$.
			\end{subsubcases}
			
			\subcase When $\vv_{\jj}$ lies on the $\vv_{\aaa} \vv_{\bb}$-path through $\vv_{\n+\left\lfloor \frac{\m}{2}\right\rfloor}$. This is shown in Figure \ref{proof5figure}.

			\begin{figure}[H]
				\begin{center}
					\tikzset{every picture/.style={line width=0.75pt}} 
					\resizebox{6.4cm}{3.6cm}{
						\begin{tikzpicture}[x=0.75pt,y=0.75pt,yscale=-1,xscale=1]
							
							\draw  [fill={rgb, 255:red, 0; green, 0; blue, 0 }  ,fill opacity=1 ] (315.98,135.29) .. controls (315.98,132.44) and (318.18,130.13) .. (320.9,130.13) .. controls (323.62,130.13) and (325.82,132.44) .. (325.82,135.29) .. controls (325.82,138.14) and (323.62,140.45) .. (320.9,140.45) .. controls (318.18,140.45) and (315.98,138.14) .. (315.98,135.29) -- cycle ;
							\draw   (320.9,135.29) .. controls (320.9,97.51) and (351.22,66.89) .. (388.62,66.89) .. controls (426.02,66.89) and (456.33,97.51) .. (456.33,135.29) .. controls (456.33,173.06) and (426.02,203.69) .. (388.62,203.69) .. controls (351.22,203.69) and (320.9,173.06) .. (320.9,135.29) -- cycle ;
							\draw   (185.46,135.29) .. controls (185.46,97.51) and (215.78,66.89) .. (253.18,66.89) .. controls (290.58,66.89) and (320.9,97.51) .. (320.9,135.29) .. controls (320.9,173.06) and (290.58,203.69) .. (253.18,203.69) .. controls (215.78,203.69) and (185.46,173.06) .. (185.46,135.29) -- cycle ;
							\draw  [fill={rgb, 255:red, 126; green, 211; blue, 33 }  ,fill opacity=1 ] (368.59,202.09) .. controls (368.59,199.24) and (370.79,196.93) .. (373.51,196.93) .. controls (376.23,196.93) and (378.43,199.24) .. (378.43,202.09) .. controls (378.43,204.94) and (376.23,207.25) .. (373.51,207.25) .. controls (370.79,207.25) and (368.59,204.94) .. (368.59,202.09) -- cycle ;
							\draw  [fill={rgb, 255:red, 126; green, 211; blue, 33 }  ,fill opacity=1 ] (368.34,69.89) .. controls (368.34,67.04) and (370.54,64.73) .. (373.26,64.73) .. controls (375.98,64.73) and (378.18,67.04) .. (378.18,69.89) .. controls (378.18,72.74) and (375.98,75.05) .. (373.26,75.05) .. controls (370.54,75.05) and (368.34,72.74) .. (368.34,69.89) -- cycle ;
							\draw  [fill={rgb, 255:red, 0; green, 0; blue, 0 }  ,fill opacity=1 ] (180.54,130.13) .. controls (180.54,127.28) and (182.75,124.97) .. (185.46,124.97) .. controls (188.18,124.97) and (190.39,127.28) .. (190.39,130.13) .. controls (190.39,132.98) and (188.18,135.29) .. (185.46,135.29) .. controls (182.75,135.29) and (180.54,132.98) .. (180.54,130.13) -- cycle ;
							\draw  [fill={rgb, 255:red, 0; green, 0; blue, 0 }  ,fill opacity=1 ] (449.46,121.29) .. controls (449.46,118.44) and (451.67,116.13) .. (454.39,116.13) .. controls (457.1,116.13) and (459.31,118.44) .. (459.31,121.29) .. controls (459.31,124.14) and (457.1,126.45) .. (454.39,126.45) .. controls (451.67,126.45) and (449.46,124.14) .. (449.46,121.29) -- cycle ;
							\draw  [fill={rgb, 255:red, 208; green, 2; blue, 27 }  ,fill opacity=1 ] (432.69,87.89) .. controls (432.69,85.04) and (434.9,82.73) .. (437.62,82.73) .. controls (440.33,82.73) and (442.54,85.04) .. (442.54,87.89) .. controls (442.54,90.74) and (440.33,93.05) .. (437.62,93.05) .. controls (434.9,93.05) and (432.69,90.74) .. (432.69,87.89) -- cycle ;
							
							\draw (288.11,130.31) node [anchor=north west][inner sep=0.75pt]  [font=\footnotesize]  {$\vv_{\n}$};
							\draw (233.76,247.69) node [anchor=north west][inner sep=0.75pt]  [font=\large]  {$\C_{\n}$};
							\draw (390.51,248.08) node [anchor=north west][inner sep=0.75pt]  [font=\large]  {$\C_{\m}$};
							\draw (367.18,216.09) node [anchor=north west][inner sep=0.75pt]  [font=\footnotesize]  {$\vv_{\bb}$};
							\draw (141.32,119.26) node [anchor=north west][inner sep=0.75pt]  [font=\footnotesize]  {$\vv_{\lfloor \frac{\n}{2} \rfloor }$};
							\draw (473.32,110.26) node [anchor=north west][inner sep=0.75pt]  [font=\footnotesize]  {$\vv_{\n+\lfloor \frac{\m}{2} \rfloor }$};
							\draw (368,37.4) node [anchor=north west][inner sep=0.75pt]  [font=\footnotesize]  {$\vv_{\aaa}$};
							\draw (433,55.4) node [anchor=north west][inner sep=0.75pt]  [font=\footnotesize]  {$\vv_{\jj}$};
							\draw (347,100.4) node [anchor=north west][inner sep=0.75pt]  [font=\large]  {$\PP$};
							\draw (405,158.4) node [anchor=north west][inner sep=0.75pt]  [font=\large]  {$\PP'$};

						\end{tikzpicture}

					}
					\setlength{\belowcaptionskip}{-20pt}
					\caption{$\vv_{\aaa},\vv_{\bb} \in \C_{\m}$ for base bicyclic graph $\C_{\n,\m}$ of type-I for odd $\n,\m$ and $\vv_{\jj}$ on $\vv_{\aaa} \vv_{\bb}$-path through $\vv_{\n+\lfloor \frac{\m}{2} \rfloor }$.}
					\label{proof5figure}
				\end{center}
			\end{figure}
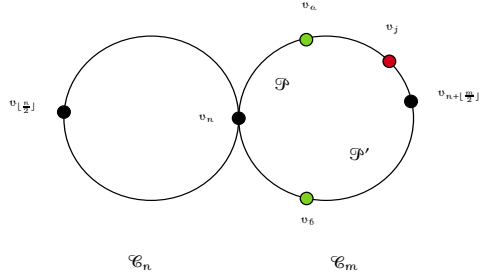
			If $\vv_{\jj}=\vv_{\n+\lfloor \frac{\m}{2} \rfloor }$, we can easily see that $\vv_{\jj}$ resolves $\vv_{\aaa}$ and $\vv_{\bb}$ by using the facts $\left|\PP\right|\neq \left|\PP'\right|$ and $\dd(\vv_{\n},\vv_{\aaa})=\dd(\vv_{\n},\vv_{\bb})$. Next, we consider the cases when $\vv_{\jj}\neq \vv_{\n+\lfloor \frac{\m}{2} \rfloor }$.
			\begin{subsubcases}
				\subsubcase If the shortest $\vv_{\jj} \vv_{\bb}$-paths passes through $\vv_{\n+\lfloor \frac{\m}{2} \rfloor }$, we get \[\dd(\vv_{\jj},\vv_{\bb})=\dd\left(\vv_{\jj},\vv_{\n+\lfloor \frac{\m}{2} \rfloor }\right)+\dd\left(\vv_{\bb}, \vv_{\n+\lfloor \frac{\m}{2} \rfloor }\right).\] On the other hand, \[\dd(\vv_{\jj},\vv_{\aaa})=\dd\left(\vv_{\aaa},\vv_{\n+\lfloor \frac{\m}{2} \rfloor }\right)-d\left(\vv_{\jj},\vv_{\n+\lfloor \frac{\m}{2} \rfloor }\right).\] 
				
				Let us suppose that $\vv_{\jj}$ does not resolve $\vv_{\aaa}$ and $\vv_{\bb}$, then $\dd(\vv_{\jj},\vv_{\aaa})=\dd(\vv_{\jj},\vv_{\bb})$ gives us,
				\begin{equation*}
					2\dd\left(\vv_{\jj},\vv_{\n+\lfloor \frac{\m}{2} \rfloor }\right)=\dd\left(\vv_{\aaa},\vv_{\n+\lfloor \frac{\m}{2} \rfloor }\right)-d\left(\vv_{\bb}, \vv_{\n+\lfloor \frac{\m}{2} \rfloor }\right).
				\end{equation*}
				Since $\dd\left(\vv_{\aaa},\vv_{\n+\lfloor \frac{\m}{2} \rfloor }\right) \neq d\left(\vv_{\bb}, \vv_{\n+\lfloor \frac{\m}{2} \rfloor }\right)$, and distance can not be a fraction, the above equation is a contradiction. Hence, $\vv_{\jj}$ resolves $\vv_{\aaa}$ and $\vv_{\bb}$.
				
				\subsubcase If the shortest $\vv_{\jj} \vv_{\bb}$-path passes through $\vv_{\n}$, we get $\dd(\vv_{\jj},\vv_{\bb})=\dd(\vv_{\jj},\vv_{\n})+\dd(\vv_{\n}, \vv_{\bb})$ and $\dd(\vv_{\jj},\vv_{\aaa})=\dd(\vv_{\jj},\vv_{\n})-\dd(\vv_{\n},\vv_{\aaa})$. Since $\dd(\vv_{\n},\vv_{\aaa})=\dd(\vv_{\n}, \vv_{\bb})$, we easily conclude that $\vv_{\jj}$ resolves $\vv_{\aaa}$ and $\vv_{\bb}$.
			\end{subsubcases}
			
		\end{subcases}
		
	\end{mycases}
	The above argument ensures that whenever $\vv_{\aaa},\vv_{\bb}$ are equidistant from $\vv_{\left\lfloor\frac{\n}{2}\right\rfloor}$, they are always resolved by $\vv_{\jj}$ for all $\vv_{\jj} \in \C_{\m}$, $\vv_{\jj} \neq \vv_{\n}$. Similarly, it can be shown that $\vv_{\left\lfloor\frac{\n}{2}\right\rfloor}$ resolves all $\vv_{\aaa},\vv_{\bb}$ equidistant from $\vv_{\jj}$. Hence, $\left \{ \vv_{\left \lfloor \frac{\n}{2}\right \rfloor },\vv_{\jj}\right \} $ is a metric basis for $\C_{\n,\m}$, $\n,\m$ odd, $\vv_{\jj} \neq \vv_{\n}$.
\end{proof}


\begin{Lemma}\label{lem2}
	Let $\C_{\n,\m}$ be a base bicyclic graph of type-I, where $\n,\m$ are odd, then $\ \W=\left \{ \vv_{\n+\left \lfloor \frac{\m}{2}\right \rfloor },\vv_{\jj}\right \} $ $\left(resp. \W=\left \{ \vv_{\n+\left \lfloor \frac{\m}{2}\right \rfloor+1 },\vv_{\jj}\right \} \right)$
	is always a metric basis of $\C_{\n,\m}$, for all $\vv_{\jj}\in \C_{\n}-\{\vv_{\n}\}$.
\end{Lemma}
\begin{proof}
	This proof is a straightforward task along the same lines as Lemma \ref{oddoddmiddleall}.
\end{proof}

%
%

As for bicyclic graphs of type-I where exactly one of $\n,\m$ is even, we have the following results.

\begin{Lemma} \label{lemm1}
	Let $\C_{\n,\m}$ be a base bicyclic graph of type-I where $\n$ is odd and $\m$ is even, then $\W=\left\{\vv_{\left \lfloor \frac{\n}{2}\right \rfloor },\vv_{\jj}\right\}$ $\left(resp. \W=\left \{ \vv_{\left \lfloor \frac{\n}{2}\right \rfloor+1 },\vv_{\jj}\right \} \right)$ , $\vv_{\jj} \in \C_{\m}-\left\{\vv_{\n} ,\vv_{\n+\frac{\m}{2}}\right\}$, is always a metric basis for $\C_{\n,\m}$.
\end{Lemma}
\begin{proof}
	By Lemma \ref{lem1}, $\vv_{\n} \notin \W$, while by Lemma \ref{newlem}, $\vv_{\n+\frac{\m}{2}} \notin \W$. The rest of the proof follows the same argument as given in Lemma \ref{oddoddmiddleall}.
\end{proof}

We now turn our attention to bicyclic graphs of type-I where $\n,\m$ are even.

Before proceeding further, let us define some partitions of $\V(\C_{\n,\m})$ which we will reference in the upcoming results. Let

\begin{equation}\label{eq1}
	V_{\ii}=\begin{cases}
		\left\{\vv_1 ,\vv_2, \cdots, \vv_{\frac{\n}{2}-1}\right\} \text{ for } \ii =1 \\
		\left\{\vv_{\frac{\n}{2}}\right\} \text{ for } \ii =2 \\ 
		\left\{\vv_{\frac{\n}{2}+1}, \vv_{\frac{\n}{2}+2} \cdots \vv_{\n-1}\right\} \text{ for } \ii =3 \\
		\left\{\vv_{\n}\right\} \text{ for } \ii =4 \\
		\left\{\vv_{\n+1}, \vv_{\n+2}, \cdots, \vv_{\n+\frac{\m}{2}-1}\right\} \text{ for } \ii =5\\
		\left\{\vv_{\n+\frac{\m}{2}}\right\} \text{ for } \ii =6 \\
		\left\{\vv_{\n+\frac{\m}{2}+1}, \vv_{\n+\frac{\m}{2}+2} \cdots \vv_{\n+m-1}\right\} \text{ for } \ii =7.
	\end{cases}
\end{equation}
For any metric basis $\W$ of $\C_{\n,\m}$, from Lemma \ref{lem1}, we know that $\V_4 \not \subset  \W$, and, from Lemma \ref{11}, we have $\V_2 \cup \V_6 \not \subset \W$. In fact, we can state and prove a much stronger result as follows.

\begin{Lemma} \label{evenevenallbasis}
	Any set of the form $\W=\{\vv_{\ii},\vv_{\jj}, \vv_{\kk}\}$, where, $\vv_{\ii} \in \V_1$ (resp. $\vv_{\ii} \in \V_5$), $\vv_{\jj} \in \V_3$ (resp. $\vv_{\jj} \in \V_7$), and $\vv_{\kk} \in \V_5 \cup \V_7$ (resp. $\vv_{\kk} \in \V_1 \cup \V_3$), is a metric basis for $\C_{\n,\m}$ for even $\n,\m$.
\end{Lemma}
\begin{proof}
	Let $\W=\{\vv_{\ii},\vv_{\jj}, \vv_{\kk}\}$ be as given above. Without loss of generality, let $\vv_{\ii} \in \V_1$, $\vv_{\jj} \in \V_3$ and $\vv_{\kk} \in \V_5$. We have to show that all vertices of $\C_{\n,\m}$ have distinct representation with respect to $\W$, that is, any two vertices which are not resolved by a particular element of $\W$, are resolved by at least one of the remaining two elements. 
	
	To this end, let $\vv_{\aaa}, \vv_{\bb} \in \V(\G)$ be equidistant from $\vv_{\ii}$. Different possibilities of such $\vv_{\aaa},\vv_{\bb}$ are given in the following figures.
	
		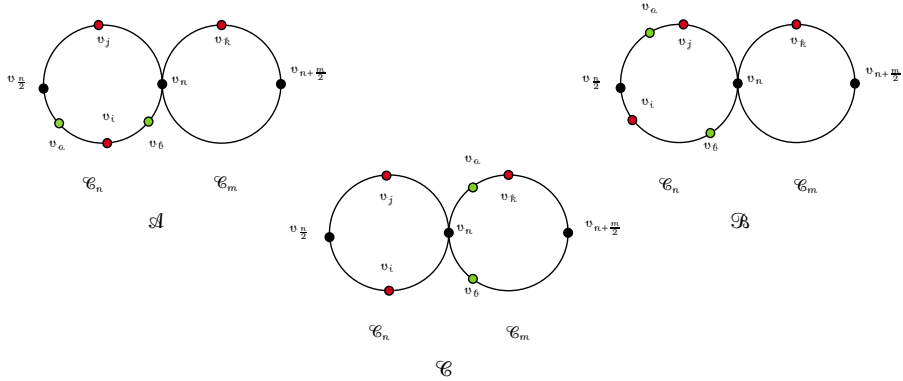
\begin{figure}[H]
		\begin{center}

			\tikzset{every picture/.style={line width=0.75pt}} 
			\resizebox{12cm}{5cm}
			{
			\begin{tikzpicture}[x=0.75pt,y=0.75pt,yscale=-1,xscale=1]
				
				\draw  [fill={rgb, 255:red, 0; green, 0; blue, 0 }  ,fill opacity=1 ] (116.23,59.88) .. controls (116.23,58.11) and (117.63,56.67) .. (119.35,56.67) .. controls (121.07,56.67) and (122.46,58.11) .. (122.46,59.88) .. controls (122.46,61.66) and (121.07,63.1) .. (119.35,63.1) .. controls (117.63,63.1) and (116.23,61.66) .. (116.23,59.88) -- cycle ;
				\draw   (119.35,59.88) .. controls (119.35,36.35) and (138.52,17.27) .. (162.17,17.27) .. controls (185.82,17.27) and (204.99,36.35) .. (204.99,59.88) .. controls (204.99,83.42) and (185.82,102.5) .. (162.17,102.5) .. controls (138.52,102.5) and (119.35,83.42) .. (119.35,59.88) -- cycle ;
				\draw   (33.71,59.88) .. controls (33.71,36.35) and (52.88,17.27) .. (76.53,17.27) .. controls (100.18,17.27) and (119.35,36.35) .. (119.35,59.88) .. controls (119.35,83.42) and (100.18,102.5) .. (76.53,102.5) .. controls (52.88,102.5) and (33.71,83.42) .. (33.71,59.88) -- cycle ;
				\draw  [fill={rgb, 255:red, 208; green, 2; blue, 27 }  ,fill opacity=1 ] (76.53,102.5) .. controls (76.53,100.72) and (77.92,99.28) .. (79.64,99.28) .. controls (81.36,99.28) and (82.75,100.72) .. (82.75,102.5) .. controls (82.75,104.27) and (81.36,105.71) .. (79.64,105.71) .. controls (77.92,105.71) and (76.53,104.27) .. (76.53,102.5) -- cycle ;
				\draw  [fill={rgb, 255:red, 0; green, 0; blue, 0 }  ,fill opacity=1 ] (30.59,63.1) .. controls (30.59,61.32) and (31.99,59.88) .. (33.71,59.88) .. controls (35.42,59.88) and (36.82,61.32) .. (36.82,63.1) .. controls (36.82,64.87) and (35.42,66.31) .. (33.71,66.31) .. controls (31.99,66.31) and (30.59,64.87) .. (30.59,63.1) -- cycle ;
				\draw  [fill={rgb, 255:red, 0; green, 0; blue, 0 }  ,fill opacity=1 ] (201.88,59.88) .. controls (201.88,58.11) and (203.27,56.67) .. (204.99,56.67) .. controls (206.71,56.67) and (208.1,58.11) .. (208.1,59.88) .. controls (208.1,61.66) and (206.71,63.1) .. (204.99,63.1) .. controls (203.27,63.1) and (201.88,61.66) .. (201.88,59.88) -- cycle ;
				\draw  [fill={rgb, 255:red, 208; green, 2; blue, 27 }  ,fill opacity=1 ] (159.05,17.27) .. controls (159.05,15.49) and (160.45,14.05) .. (162.17,14.05) .. controls (163.89,14.05) and (165.28,15.49) .. (165.28,17.27) .. controls (165.28,19.04) and (163.89,20.48) .. (162.17,20.48) .. controls (160.45,20.48) and (159.05,19.04) .. (159.05,17.27) -- cycle ;
				\draw  [fill={rgb, 255:red, 143; green, 211; blue, 64 }  ,fill opacity=1 ] (41.85,88.17) .. controls (41.85,86.39) and (43.24,84.95) .. (44.96,84.95) .. controls (46.68,84.95) and (48.07,86.39) .. (48.07,88.17) .. controls (48.07,89.94) and (46.68,91.38) .. (44.96,91.38) .. controls (43.24,91.38) and (41.85,89.94) .. (41.85,88.17) -- cycle ;
				\draw  [fill={rgb, 255:red, 126; green, 211; blue, 33 }  ,fill opacity=1 ] (106.34,86.92) .. controls (106.34,85.15) and (107.74,83.71) .. (109.46,83.71) .. controls (111.18,83.71) and (112.57,85.15) .. (112.57,86.92) .. controls (112.57,88.7) and (111.18,90.14) .. (109.46,90.14) .. controls (107.74,90.14) and (106.34,88.7) .. (106.34,86.92) -- cycle ;
				\draw  [fill={rgb, 255:red, 208; green, 2; blue, 27 }  ,fill opacity=1 ] (70.3,17.27) .. controls (70.3,15.49) and (71.69,14.05) .. (73.41,14.05) .. controls (75.13,14.05) and (76.53,15.49) .. (76.53,17.27) .. controls (76.53,19.04) and (75.13,20.48) .. (73.41,20.48) .. controls (71.69,20.48) and (70.3,19.04) .. (70.3,17.27) -- cycle ;
				\draw  [fill={rgb, 255:red, 0; green, 0; blue, 0 }  ,fill opacity=1 ] (531.54,59.43) .. controls (531.54,57.65) and (532.92,56.21) .. (534.61,56.21) .. controls (536.31,56.21) and (537.69,57.65) .. (537.69,59.43) .. controls (537.69,61.2) and (536.31,62.64) .. (534.61,62.64) .. controls (532.92,62.64) and (531.54,61.2) .. (531.54,59.43) -- cycle ;
				\draw   (534.61,59.43) .. controls (534.61,35.87) and (553.54,16.78) .. (576.9,16.78) .. controls (600.25,16.78) and (619.18,35.87) .. (619.18,59.43) .. controls (619.18,82.98) and (600.25,102.07) .. (576.9,102.07) .. controls (553.54,102.07) and (534.61,82.98) .. (534.61,59.43) -- cycle ;
				\draw   (450.05,59.43) .. controls (450.05,35.87) and (468.98,16.78) .. (492.33,16.78) .. controls (515.68,16.78) and (534.61,35.87) .. (534.61,59.43) .. controls (534.61,82.98) and (515.68,102.07) .. (492.33,102.07) .. controls (468.98,102.07) and (450.05,82.98) .. (450.05,59.43) -- cycle ;
				\draw  [fill={rgb, 255:red, 208; green, 2; blue, 27 }  ,fill opacity=1 ] (455.49,85.86) .. controls (455.49,84.08) and (456.86,82.64) .. (458.56,82.64) .. controls (460.26,82.64) and (461.64,84.08) .. (461.64,85.86) .. controls (461.64,87.64) and (460.26,89.08) .. (458.56,89.08) .. controls (456.86,89.08) and (455.49,87.64) .. (455.49,85.86) -- cycle ;
				\draw  [fill={rgb, 255:red, 0; green, 0; blue, 0 }  ,fill opacity=1 ] (446.97,62.64) .. controls (446.97,60.87) and (448.35,59.43) .. (450.05,59.43) .. controls (451.74,59.43) and (453.12,60.87) .. (453.12,62.64) .. controls (453.12,64.42) and (451.74,65.86) .. (450.05,65.86) .. controls (448.35,65.86) and (446.97,64.42) .. (446.97,62.64) -- cycle ;
				\draw  [fill={rgb, 255:red, 0; green, 0; blue, 0 }  ,fill opacity=1 ] (616.11,59.43) .. controls (616.11,57.65) and (617.48,56.21) .. (619.18,56.21) .. controls (620.88,56.21) and (622.26,57.65) .. (622.26,59.43) .. controls (622.26,61.2) and (620.88,62.64) .. (619.18,62.64) .. controls (617.48,62.64) and (616.11,61.2) .. (616.11,59.43) -- cycle ;
				\draw  [fill={rgb, 255:red, 208; green, 2; blue, 27 }  ,fill opacity=1 ] (573.82,16.78) .. controls (573.82,15) and (575.2,13.56) .. (576.9,13.56) .. controls (578.6,13.56) and (579.97,15) .. (579.97,16.78) .. controls (579.97,18.56) and (578.6,20) .. (576.9,20) .. controls (575.2,20) and (573.82,18.56) .. (573.82,16.78) -- cycle ;
				\draw  [fill={rgb, 255:red, 126; green, 211; blue, 33 }  ,fill opacity=1 ] (468.07,22.89) .. controls (468.07,21.11) and (469.45,19.67) .. (471.15,19.67) .. controls (472.85,19.67) and (474.22,21.11) .. (474.22,22.89) .. controls (474.22,24.66) and (472.85,26.1) .. (471.15,26.1) .. controls (469.45,26.1) and (468.07,24.66) .. (468.07,22.89) -- cycle ;
				\draw  [fill={rgb, 255:red, 126; green, 211; blue, 33 }  ,fill opacity=1 ] (511.78,95.21) .. controls (511.78,93.44) and (513.16,92) .. (514.86,92) .. controls (516.55,92) and (517.93,93.44) .. (517.93,95.21) .. controls (517.93,96.99) and (516.55,98.43) .. (514.86,98.43) .. controls (513.16,98.43) and (511.78,96.99) .. (511.78,95.21) -- cycle ;
				\draw  [fill={rgb, 255:red, 208; green, 2; blue, 27 }  ,fill opacity=1 ] (492.33,16.78) .. controls (492.33,15) and (493.71,13.56) .. (495.4,13.56) .. controls (497.1,13.56) and (498.48,15) .. (498.48,16.78) .. controls (498.48,18.56) and (497.1,20) .. (495.4,20) .. controls (493.71,20) and (492.33,18.56) .. (492.33,16.78) -- cycle ;
				\draw  [fill={rgb, 255:red, 0; green, 0; blue, 0 }  ,fill opacity=1 ] (322.99,167.23) .. controls (322.99,165.49) and (324.39,164.08) .. (326.12,164.08) .. controls (327.85,164.08) and (329.25,165.49) .. (329.25,167.23) .. controls (329.25,168.97) and (327.85,170.38) .. (326.12,170.38) .. controls (324.39,170.38) and (322.99,168.97) .. (322.99,167.23) -- cycle ;
				\draw   (326.12,167.23) .. controls (326.12,144.16) and (345.39,125.46) .. (369.16,125.46) .. controls (392.93,125.46) and (412.2,144.16) .. (412.2,167.23) .. controls (412.2,190.31) and (392.93,209.01) .. (369.16,209.01) .. controls (345.39,209.01) and (326.12,190.31) .. (326.12,167.23) -- cycle ;
				\draw   (240.03,167.23) .. controls (240.03,144.16) and (259.31,125.46) .. (283.08,125.46) .. controls (306.85,125.46) and (326.12,144.16) .. (326.12,167.23) .. controls (326.12,190.31) and (306.85,209.01) .. (283.08,209.01) .. controls (259.31,209.01) and (240.03,190.31) .. (240.03,167.23) -- cycle ;
				\draw  [fill={rgb, 255:red, 208; green, 2; blue, 27 }  ,fill opacity=1 ] (279.9,209.01) .. controls (279.9,207.27) and (281.3,205.86) .. (283.03,205.86) .. controls (284.76,205.86) and (286.16,207.27) .. (286.16,209.01) .. controls (286.16,210.75) and (284.76,212.16) .. (283.03,212.16) .. controls (281.3,212.16) and (279.9,210.75) .. (279.9,209.01) -- cycle ;
				\draw  [fill={rgb, 255:red, 0; green, 0; blue, 0 }  ,fill opacity=1 ] (236.91,170.38) .. controls (236.91,168.64) and (238.31,167.23) .. (240.03,167.23) .. controls (241.76,167.23) and (243.16,168.64) .. (243.16,170.38) .. controls (243.16,172.12) and (241.76,173.53) .. (240.03,173.53) .. controls (238.31,173.53) and (236.91,172.12) .. (236.91,170.38) -- cycle ;
				\draw  [fill={rgb, 255:red, 0; green, 0; blue, 0 }  ,fill opacity=1 ] (409.08,167.23) .. controls (409.08,165.49) and (410.48,164.08) .. (412.2,164.08) .. controls (413.93,164.08) and (415.33,165.49) .. (415.33,167.23) .. controls (415.33,168.97) and (413.93,170.38) .. (412.2,170.38) .. controls (410.48,170.38) and (409.08,168.97) .. (409.08,167.23) -- cycle ;
				\draw  [fill={rgb, 255:red, 208; green, 2; blue, 27 }  ,fill opacity=1 ] (366.03,125.46) .. controls (366.03,123.72) and (367.43,122.31) .. (369.16,122.31) .. controls (370.89,122.31) and (372.29,123.72) .. (372.29,125.46) .. controls (372.29,127.2) and (370.89,128.61) .. (369.16,128.61) .. controls (367.43,128.61) and (366.03,127.2) .. (366.03,125.46) -- cycle ;
				\draw  [fill={rgb, 255:red, 126; green, 211; blue, 33 }  ,fill opacity=1 ] (340.38,134.49) .. controls (340.38,132.75) and (341.78,131.34) .. (343.51,131.34) .. controls (345.24,131.34) and (346.64,132.75) .. (346.64,134.49) .. controls (346.64,136.23) and (345.24,137.64) .. (343.51,137.64) .. controls (341.78,137.64) and (340.38,136.23) .. (340.38,134.49) -- cycle ;
				\draw  [fill={rgb, 255:red, 126; green, 211; blue, 33 }  ,fill opacity=1 ] (340.38,200.46) .. controls (340.38,198.72) and (341.78,197.31) .. (343.51,197.31) .. controls (345.24,197.31) and (346.64,198.72) .. (346.64,200.46) .. controls (346.64,202.2) and (345.24,203.61) .. (343.51,203.61) .. controls (341.78,203.61) and (340.38,202.2) .. (340.38,200.46) -- cycle ;
				\draw  [fill={rgb, 255:red, 208; green, 2; blue, 27 }  ,fill opacity=1 ] (277.99,125.94) .. controls (277.99,124.2) and (279.39,122.79) .. (281.12,122.79) .. controls (282.85,122.79) and (284.25,124.2) .. (284.25,125.94) .. controls (284.25,127.68) and (282.85,129.09) .. (281.12,129.09) .. controls (279.39,129.09) and (277.99,127.68) .. (277.99,125.94) -- cycle ;
				
				\draw (124.31,52.61) node [anchor=north west][inner sep=0.75pt]  [font=\footnotesize]  {$\vv_{\n}$};
				\draw (60.28,125.3) node [anchor=north west][inner sep=0.75pt]  [font=\normalsize]  {$\C_{\n}$};
				\draw (155.06,124.92) node [anchor=north west][inner sep=0.75pt]  [font=\normalsize]  {$\C_{\m}$};
				\draw (5.83,52.38) node [anchor=north west][inner sep=0.75pt]  [font=\footnotesize]  {$\vv_{\frac{\n}{2}}$};
				\draw (210.22,45.52) node [anchor=north west][inner sep=0.75pt]  [font=\footnotesize]  {$\vv_{\n+\frac{\m}{2}}$};
				\draw (74.21,79.46) node [anchor=north west][inner sep=0.75pt]  [font=\footnotesize]  {$\vv_{\ii}$};
				\draw (155.41,24.01) node [anchor=north west][inner sep=0.75pt]  [font=\footnotesize]  {$\vv_{\kk}$};
				\draw (35.9,99.4) node [anchor=north west][inner sep=0.75pt]  [font=\footnotesize]  {$\vv_{\aaa}$};
				\draw (107.98,98.78) node [anchor=north west][inner sep=0.75pt]  [font=\footnotesize]  {$\vv_{\bb}$};
				\draw (70.41,24.01) node [anchor=north west][inner sep=0.75pt]  [font=\footnotesize]  {$\vv_{\jj}$};
				\draw (539.42,52.16) node [anchor=north west][inner sep=0.75pt]  [font=\footnotesize]  {$\vv_{\n}$};
				\draw (476.26,126.27) node [anchor=north west][inner sep=0.75pt]  [font=\normalsize]  {$\C_{\n}$};
				\draw (573.58,126.51) node [anchor=north west][inner sep=0.75pt]  [font=\normalsize]  {$\C_{\m}$};
				\draw (420.54,51.92) node [anchor=north west][inner sep=0.75pt]  [font=\footnotesize]  {$\vv_{\frac{\n}{2}}$};
				\draw (625.07,46.31) node [anchor=north west][inner sep=0.75pt]  [font=\footnotesize]  {$\vv_{\n+\frac{\m}{2}}$};
				\draw (463.11,69.05) node [anchor=north west][inner sep=0.75pt]  [font=\footnotesize]  {$\vv_{\ii}$};
				\draw (570.76,25.04) node [anchor=north west][inner sep=0.75pt]  [font=\footnotesize]  {$\vv_{\kk}$};
				\draw (463.36,0.46) node [anchor=north west][inner sep=0.75pt]  [font=\footnotesize]  {$\vv_{\aaa}$};
				\draw (508.15,100.24) node [anchor=north west][inner sep=0.75pt]  [font=\footnotesize]  {$\vv_{\bb}$};
				\draw (489.96,25.41) node [anchor=north west][inner sep=0.75pt]  [font=\footnotesize]  {$\vv_{\jj}$};
				\draw (329.88,159.99) node [anchor=north west][inner sep=0.75pt]  [font=\footnotesize]  {$\vv_{\n}$};
				\draw (266.91,232.54) node [anchor=north west][inner sep=0.75pt]  [font=\normalsize]  {$\C_{\n}$};
				\draw (366,232.78) node [anchor=north west][inner sep=0.75pt]  [font=\normalsize]  {$\C_{\m}$};
				\draw (209.34,160.65) node [anchor=north west][inner sep=0.75pt]  [font=\footnotesize]  {$\vv_{\frac{\n}{2}}$};
				\draw (422.68,157.15) node [anchor=north west][inner sep=0.75pt]  [font=\footnotesize]  {$\vv_{\n+\frac{\m}{2}}$};
				\draw (275.05,187.53) node [anchor=north west][inner sep=0.75pt]  [font=\footnotesize]  {$\vv_{\ii}$};
				\draw (361.77,137.11) node [anchor=north west][inner sep=0.75pt]  [font=\footnotesize]  {$\vv_{\kk}$};
				\draw (335.71,109.96) node [anchor=north west][inner sep=0.75pt]  [font=\footnotesize]  {$\vv_{\aaa}$};
				\draw (335.53,205.26) node [anchor=north west][inner sep=0.75pt]  [font=\footnotesize]  {$\vv_{\bb}$};
				\draw (274.42,137.67) node [anchor=north west][inner sep=0.75pt]  [font=\footnotesize]  {$\vv_{\jj}$};
				\draw (108,151.4) node [anchor=north west][inner sep=0.75pt]  [font=\Large]  {$\A$};
				\draw (528,151.4) node [anchor=north west][inner sep=0.75pt]  [font=\Large]  {$\B$};
				\draw (315,258.4) node [anchor=north west][inner sep=0.75pt]  [font=\Large]  {$\C$};

			\end{tikzpicture}
		}
				
			\caption{Different occurrences of $\vv_{\aaa}$ and $\vv_{\bb}$ equidistant from $\vv_{\ii}$}
			\label{evenevenresolvingfigure1}
			\end{center}
	\end{figure}
	
	\begin{itemize}
		\item In case ${\A}$, we see that $\vv_{\aaa},\vv_{\bb}$ lie on the same shortest $\vv_{\kk} \vv_{\frac{\n}{2}}$ path and hence are resolved by $\vv_{\kk}$.
		\item In case $\B$, since $\vv_{\aaa},\vv_{\bb}$ are equidistant from $\vv_{\ii}$ and $\vv_{\ii} \neq \vv_{\frac{\n}{2}}$, it is evident that $\dd\left(\vv_{\aaa}, \vv_{\frac{\n}{2}}\right) \neq d\left(\vv_{\bb}, \vv_{\frac{\n}{2}}\right)$. Since both $\vv_{\n} \vv_{\frac{\n}{2}}$ paths have same lengths, we conclude that $\dd(\vv_{\aaa}, \vv_{\n}) \neq \dd(\vv_{\bb}, \vv_{\n})$. Consequently, $\vv_{\kk}$ resolves both $\vv_{\aaa}$ and $\vv_{\bb}$.
		\item In case $\C$, we see that $\vv_{\aaa}, \vv_{\bb}$ equidistant from $\vv_{\ii}$ ensure that they both belong to different halves of $\C_{\m}$ and $\dd(\vv_{\aaa}, \vv_{\n})=\dd(\vv_{\bb},\vv_{\n})$. Since $\m$ is even, we can also conclude that $\dd\left(\vv_{\aaa}, \vv_{\n+\frac{\m}{2}}\right) = d\left(\vv_{\bb}, \vv_{\n+\frac{\m}{2}}\right)$. Now $\vv_{\kk} \in \V_5=\left\{\vv_{\n+1}, \vv_{\n+2}, \cdots, \vv_{\n+\frac{\m}{2}-1}\right\}$ implies that $\vv_{\kk} \neq \vv_{\n}$ and $\vv_{\kk} \neq  \vv_{\n+\frac{\m}{2}}$. Consequently, $\dd(\vv_{\aaa}, \vv_{\kk}) \neq \dd(\vv_{\bb}, \vv_{\kk})$, and $\vv_{\kk}$ resolves both $\vv_{\aaa}$ and $\vv_{\bb}$.
	\end{itemize}
	
	When $\vv_{\aaa}, \vv_{\bb}$ are equidistant from $\vv_{\jj}$ or $\vv_{\kk}$, an argument along the same lines ensures that they are always resolved by at least one element of $\W$. This concludes our proof for $\W=\{\vv_{\ii},\vv_{\jj}, \vv_{\kk}\}$ where $\vv_{\ii} \in \V_1, \vv_{\jj} \in \V_3$ and $\vv_{\kk} \in \V_5$.

	All other possible constructions of $\W$ can be shown to be a metric basis following the same proof technique.
	
\end{proof}

%

Now we move on to bicyclic graphs of type-II denoted by $\C_{\n,\rr,\m}$. It should be noted that we will use vertex labeling as given in figure \ref{bicyclicgraphtypeii}.

\begin{Lemma}\label{bicycliciipathvertices}
The vertices $\V(\PP_{\rr})=\{\vv_{\n}, \vv_{\n+1}, \cdots, \vv_{\n+\rr}\}$ do not belong to any metric basis of $\C_{\n,\rr,\m}$.
\end{Lemma}

\begin{proof}
	By Theorem \ref{oldthm2}, $\beta(\C_{\n,\rr,\m})=2$. Let us suppose on the contrary that the vertices from $\V(\P_{\rr})$ belong to a metric basis set $\W$. Two possibilities arise here, namely, both elements of $\W$ are from $\V(\P_{\rr})$, or only one element of $\W$ is from $\V(\P_{\rr})$. These cases are shown in figure \ref{bicycliciipathfigure}, $\A$ and $\B$.
	
	\begin{figure}[H]
		\begin{center}
			\tikzset{every picture/.style={line width=0.75pt}} 
			\resizebox{12cm}{3cm}{	
				\begin{tikzpicture}[x=0.75pt,y=0.75pt,yscale=-1,xscale=1]
					
					\draw  [fill={rgb, 255:red, 0; green, 0; blue, 0 }  ,fill opacity=1 ] (553.36,110.09) .. controls (553.36,108.31) and (554.73,106.86) .. (556.41,106.86) .. controls (558.1,106.86) and (559.46,108.31) .. (559.46,110.09) .. controls (559.46,111.87) and (558.1,113.31) .. (556.41,113.31) .. controls (554.73,113.31) and (553.36,111.87) .. (553.36,110.09) -- cycle ;
					\draw   (556.41,110.09) .. controls (556.41,86.49) and (575.2,67.35) .. (598.37,67.35) .. controls (621.55,67.35) and (640.33,86.49) .. (640.33,110.09) .. controls (640.33,133.69) and (621.55,152.82) .. (598.37,152.82) .. controls (575.2,152.82) and (556.41,133.69) .. (556.41,110.09) -- cycle ;
					\draw   (360.33,110.09) .. controls (360.33,86.49) and (379.12,67.35) .. (402.29,67.35) .. controls (425.47,67.35) and (444.26,86.49) .. (444.26,110.09) .. controls (444.26,133.69) and (425.47,152.82) .. (402.29,152.82) .. controls (379.12,152.82) and (360.33,133.69) .. (360.33,110.09) -- cycle ;
					\draw  [fill={rgb, 255:red, 126; green, 211; blue, 33 }  ,fill opacity=1 ] (570.93,75.97) .. controls (570.93,74.19) and (572.3,72.75) .. (573.98,72.75) .. controls (575.67,72.75) and (577.03,74.19) .. (577.03,75.97) .. controls (577.03,77.75) and (575.67,79.2) .. (573.98,79.2) .. controls (572.3,79.2) and (570.93,77.75) .. (570.93,75.97) -- cycle ;
					\draw  [fill={rgb, 255:red, 126; green, 211; blue, 33 }  ,fill opacity=1 ] (570.32,144.07) .. controls (570.32,142.29) and (571.68,140.85) .. (573.37,140.85) .. controls (575.05,140.85) and (576.42,142.29) .. (576.42,144.07) .. controls (576.42,145.85) and (575.05,147.3) .. (573.37,147.3) .. controls (571.68,147.3) and (570.32,145.85) .. (570.32,144.07) -- cycle ;
					\draw  [fill={rgb, 255:red, 0; green, 0; blue, 0 }  ,fill opacity=1 ] (441.21,110.09) .. controls (441.21,108.31) and (442.57,106.86) .. (444.26,106.86) .. controls (445.94,106.86) and (447.31,108.31) .. (447.31,110.09) .. controls (447.31,111.87) and (445.94,113.31) .. (444.26,113.31) .. controls (442.57,113.31) and (441.21,111.87) .. (441.21,110.09) -- cycle ;
					\draw    (444.26,110.09) -- (469.66,110.09) ;
					\draw  [fill={rgb, 255:red, 0; green, 0; blue, 0 }  ,fill opacity=1 ] (466.61,110.09) .. controls (466.61,108.31) and (467.98,106.86) .. (469.66,106.86) .. controls (471.35,106.86) and (472.71,108.31) .. (472.71,110.09) .. controls (472.71,111.87) and (471.35,113.31) .. (469.66,113.31) .. controls (467.98,113.31) and (466.61,111.87) .. (466.61,110.09) -- cycle ;
					\draw  [fill={rgb, 255:red, 0; green, 0; blue, 0 }  ,fill opacity=1 ] (525.96,110.09) .. controls (525.96,108.31) and (527.32,106.86) .. (529.01,106.86) .. controls (530.69,106.86) and (532.06,108.31) .. (532.06,110.09) .. controls (532.06,111.87) and (530.69,113.31) .. (529.01,113.31) .. controls (527.32,113.31) and (525.96,111.87) .. (525.96,110.09) -- cycle ;
					\draw    (531.01,110.09) -- (556.41,110.09) ;
					\draw  [dash pattern={on 0.84pt off 2.51pt}]  (469.66,110.09) -- (531.01,110.09) ;
					\draw  [fill={rgb, 255:red, 208; green, 2; blue, 27 }  ,fill opacity=1 ] (495.11,110.09) .. controls (495.11,108.31) and (496.48,106.86) .. (498.16,106.86) .. controls (499.85,106.86) and (501.21,108.31) .. (501.21,110.09) .. controls (501.21,111.87) and (499.85,113.31) .. (498.16,113.31) .. controls (496.48,113.31) and (495.11,111.87) .. (495.11,110.09) -- cycle ;
					\draw  [fill={rgb, 255:red, 0; green, 0; blue, 0 }  ,fill opacity=1 ] (203.05,110.76) .. controls (203.05,108.97) and (204.42,107.53) .. (206.11,107.53) .. controls (207.8,107.53) and (209.17,108.97) .. (209.17,110.76) .. controls (209.17,112.55) and (207.8,113.99) .. (206.11,113.99) .. controls (204.42,113.99) and (203.05,112.55) .. (203.05,110.76) -- cycle ;
					\draw   (206.11,110.76) .. controls (206.11,87.09) and (224.97,67.89) .. (248.22,67.89) .. controls (271.48,67.89) and (290.33,87.09) .. (290.33,110.76) .. controls (290.33,134.43) and (271.48,153.63) .. (248.22,153.63) .. controls (224.97,153.63) and (206.11,134.43) .. (206.11,110.76) -- cycle ;
					\draw   (9.33,110.76) .. controls (9.33,87.09) and (28.19,67.89) .. (51.44,67.89) .. controls (74.7,67.89) and (93.55,87.09) .. (93.55,110.76) .. controls (93.55,134.43) and (74.7,153.63) .. (51.44,153.63) .. controls (28.19,153.63) and (9.33,134.43) .. (9.33,110.76) -- cycle ;
					\draw  [fill={rgb, 255:red, 126; green, 211; blue, 33 }  ,fill opacity=1 ] (73.93,76.54) .. controls (73.93,74.76) and (75.3,73.31) .. (76.99,73.31) .. controls (78.68,73.31) and (80.05,74.76) .. (80.05,76.54) .. controls (80.05,78.33) and (78.68,79.77) .. (76.99,79.77) .. controls (75.3,79.77) and (73.93,78.33) .. (73.93,76.54) -- cycle ;
					\draw  [fill={rgb, 255:red, 126; green, 211; blue, 33 }  ,fill opacity=1 ] (73.31,144.85) .. controls (73.31,143.07) and (74.68,141.62) .. (76.37,141.62) .. controls (78.06,141.62) and (79.43,143.07) .. (79.43,144.85) .. controls (79.43,146.64) and (78.06,148.08) .. (76.37,148.08) .. controls (74.68,148.08) and (73.31,146.64) .. (73.31,144.85) -- cycle ;
					\draw  [fill={rgb, 255:red, 0; green, 0; blue, 0 }  ,fill opacity=1 ] (90.49,110.76) .. controls (90.49,108.97) and (91.86,107.53) .. (93.55,107.53) .. controls (95.25,107.53) and (96.62,108.97) .. (96.62,110.76) .. controls (96.62,112.55) and (95.25,113.99) .. (93.55,113.99) .. controls (91.86,113.99) and (90.49,112.55) .. (90.49,110.76) -- cycle ;
					\draw    (93.55,110.76) -- (119.05,110.76) ;
					\draw  [fill={rgb, 255:red, 0; green, 0; blue, 0 }  ,fill opacity=1 ] (115.99,110.76) .. controls (115.99,108.97) and (117.36,107.53) .. (119.05,107.53) .. controls (120.74,107.53) and (122.11,108.97) .. (122.11,110.76) .. controls (122.11,112.55) and (120.74,113.99) .. (119.05,113.99) .. controls (117.36,113.99) and (115.99,112.55) .. (115.99,110.76) -- cycle ;
					\draw  [fill={rgb, 255:red, 0; green, 0; blue, 0 }  ,fill opacity=1 ] (177.55,110.76) .. controls (177.55,108.97) and (178.92,107.53) .. (180.62,107.53) .. controls (182.31,107.53) and (183.68,108.97) .. (183.68,110.76) .. controls (183.68,112.55) and (182.31,113.99) .. (180.62,113.99) .. controls (178.92,113.99) and (177.55,112.55) .. (177.55,110.76) -- cycle ;
					\draw    (180.62,110.76) -- (206.11,110.76) ;
					\draw  [dash pattern={on 0.84pt off 2.51pt}]  (119.05,110.76) -- (180.62,110.76) ;
					\draw  [fill={rgb, 255:red, 208; green, 2; blue, 27 }  ,fill opacity=1 ] (137.76,110.76) .. controls (137.76,108.97) and (139.13,107.53) .. (140.82,107.53) .. controls (142.51,107.53) and (143.88,108.97) .. (143.88,110.76) .. controls (143.88,112.55) and (142.51,113.99) .. (140.82,113.99) .. controls (139.13,113.99) and (137.76,112.55) .. (137.76,110.76) -- cycle ;
					\draw  [fill={rgb, 255:red, 208; green, 2; blue, 27 }  ,fill opacity=1 ] (155.17,110.76) .. controls (155.17,108.97) and (156.54,107.53) .. (158.23,107.53) .. controls (159.92,107.53) and (161.29,108.97) .. (161.29,110.76) .. controls (161.29,112.55) and (159.92,113.99) .. (158.23,113.99) .. controls (156.54,113.99) and (155.17,112.55) .. (155.17,110.76) -- cycle ;
					
					\draw (421.52,98.69) node [anchor=north west][inner sep=0.75pt]  [font=\footnotesize]  {$\vv_{\n}$};
					\draw (396.51,169.34) node [anchor=north west][inner sep=0.75pt]  [font=\normalsize]  {$\C_{\n}$};
					\draw (581.22,172.58) node [anchor=north west][inner sep=0.75pt]  [font=\normalsize]  {$\C_{\m}$};
					\draw (568.01,44.64) node [anchor=north west][inner sep=0.75pt]  [font=\footnotesize]  {$\vv_{\aaa}$};
					\draw (551.08,137.12) node [anchor=north west][inner sep=0.75pt]  [font=\footnotesize]  {$\vv_{\bb}$};
					\draw (564.65,103.69) node [anchor=north west][inner sep=0.75pt]  [font=\footnotesize]  {$\vv_{\n+\rr}$};
					\draw (459.08,120.7) node [anchor=north west][inner sep=0.75pt]  [font=\footnotesize]  {$\vv_{\n+1}$};
					\draw (511.95,85.7) node [anchor=north west][inner sep=0.75pt]  [font=\footnotesize]  {$\vv_{\n+\rr-1}$};
					\draw (493,196.4) node [anchor=north west][inner sep=0.75pt]  [font=\Large]  {$\B$};
					\draw (69.49,102.49) node [anchor=north west][inner sep=0.75pt]  [font=\footnotesize]  {$\vv_{\n}$};
					\draw (34.64,177.25) node [anchor=north west][inner sep=0.75pt]  [font=\large]  {$\C_{\n}$};
					\draw (244.11,177.49) node [anchor=north west][inner sep=0.75pt]  [font=\large]  {$\C_{\m}$};
					\draw (69.02,51.15) node [anchor=north west][inner sep=0.75pt]  [font=\footnotesize]  {$\vv_{\aaa}$};
					\draw (68.09,153.32) node [anchor=north west][inner sep=0.75pt]  [font=\footnotesize]  {$\vv_{\bb}$};
					\draw (217.41,104.34) node [anchor=north west][inner sep=0.75pt]  [font=\footnotesize]  {$\vv_{\n+\rr}$};
					\draw (106.48,121.32) node [anchor=north west][inner sep=0.75pt]  [font=\footnotesize]  {$\vv_{\n+1}$};
					\draw (162.56,89.32) node [anchor=north west][inner sep=0.75pt]  [font=\footnotesize]  {$\vv_{\n+\rr-1}$};
					\draw (133,196.4) node [anchor=north west][inner sep=0.75pt]  [font=\Large]  {$\A$};

				\end{tikzpicture}
			}
				
			\caption{Possibilities of elements of $\W$ from $\V(\PP_{\rr})$.}
			\label{bicycliciipathfigure}
		\end{center}
	\end{figure}
	Elements of $\W$ are represented as red colored vertices in the above figure. One can see that in case $\A$, the two red vertices don't resolve $\vv_{\aaa}, \vv_{\bb}$ of $\C_{\n}$ which are equidistant from $\vv_{\n}$. In case $\B$, without loss of generality, we assume that $\vv_{\jj} \in \W$ is from $\C_{\n}$ along with a vertex from $\PP_{\rr}$. One can again see that the vertices $\vv_{\aaa},\vv_{\bb}$ equidistant from $\vv_{\n+\rr}$ can not be resolved by this $\W$. This concludes our result.
\end{proof}
It should also be noted that if $\W$ is a metric basis for $\C_{\n, \rr, \m}$, the vertices of $\W$ can not be from the same cycle, since the vertices of the other cycle are not resolved by such a $\W$. This together with the Lemma \ref{bicycliciipathvertices} ensures that $\W$ contains a vertex of $\C_{\n}$ and a vertex of $\C_{\m}$ except $\vv_{\n}$ and/or $\vv_{\n+\rr}$. Next, we try to find all such vertices.
\begin{Lemma}
	For even $\C_{\n}$ (resp. even $\C_{\m}$) in $\C_{\n, \rr, \m}$, vertex $\vv_{\frac{\n}{2}}$ (resp. $\vv_{\n+\rr+\frac{\m}{2}}$) does not belong to any metric basis set $\W$.
\end{Lemma}

\begin{proof}
	Suppose on the contrary that $\vv_{\frac{\n}{2}} \in \W$, for some metric basis set $\W$. Since $\beta(\C_{\n, \rr, \m})=2$, let $\vv_{\jj} \in \W$ be the second vertex, then $\vv_{\jj} \in \C_{\m}$. Let $\vv_{\aaa}, \vv_{\bb}$ be two vertices from $\C_{\n}$ equidistant from $\vv_{\frac{\n}{2}}$. Since $\n$ is even, we can easily conclude that $\vv_{\aaa}, \vv_{\bb}$ are equidistant from $\vv_{\n}$. Again, since the shortest $\vv_{\jj} \vv_{\aaa}$ and $\vv_{\jj} \vv_{\bb}$ paths pass through $\vv_{\n}$, we conclude that the vertex $\vv_{\jj}$ does not resolve $\vv_{\aaa}$ and $\vv_{\bb}$. This concludes that the set $\W$ is not a resolving set and hence can not be a metric basis.
\end{proof}

\begin{Lemma}\label{type2metricbasis}
	Any set of the form $\W=\{\vv_{\ii},\vv_{\jj}\}$, $\vv_{\ii} \in \C_{\n} \big(\vv_{\ii} \neq \vv_{\n}$ for all $\n$ and $\vv_{\ii} \neq \vv_{\frac{\n}{2}}$ for even $\n\big)$ and $\vv_{\jj} \in \C_{\m} \big( \vv_{\jj} \neq \vv_{\n+\rr}$ for all $\n$ and $\vv_{\jj} \neq \vv_{\n+\rr+\frac{\m}{2}}$ for even $\n \big)$ forms a metric basis for $\C_{\n, \rr, \m}$.
\end{Lemma}

\begin{proof}
	Let $\W=\{\vv_{\ii},\vv_{\jj}\}$ be as given above. We have to show that such a set resolves all vertices of $\C_{\n, \rr, \m}$. Since there are possibilities of both $\n,\m$ to be odd or even or one to be odd and another to be even, let us consider $\n$ to be even and $\m$ to be odd.
	
	Since $\beta(\C_{\n, \rr, \m})=2$, we only need to show that any vertex not resolved by $\vv_{\ii}$ is resolved by $\vv_{\jj}$ and vice versa.
	
	Let $\vv_{\aaa}, \vv_{\bb}$ be vertices not resolved by $\vv_{\ii}$, i.e., $\dd(\vv_{\aaa},\vv_{\ii})=\dd(\vv_{\bb},\vv_{\ii})$. Once again, different possibilities arise for the positions of $\vv_{\aaa}$ and $\vv_{\bb}$. These are shown in figure \ref{Cnrmallbasis}.
	
	\begin{figure}[H]
		\begin{center}
			\tikzset{every picture/.style={line width=0.75pt}} 
			\resizebox{12cm}{6.5cm}{
				\begin{tikzpicture}[x=0.75pt,y=0.75pt,yscale=-1,xscale=1]
					
					\draw  [fill={rgb, 255:red, 0; green, 0; blue, 0 }  ,fill opacity=1 ] (252.43,82.49) .. controls (252.43,80.55) and (253.93,78.97) .. (255.78,78.97) .. controls (257.63,78.97) and (259.13,80.55) .. (259.13,82.49) .. controls (259.13,84.43) and (257.63,86) .. (255.78,86) .. controls (253.93,86) and (252.43,84.43) .. (252.43,82.49) -- cycle ;
					\draw   (255.78,82.49) .. controls (255.78,56.74) and (276.43,35.88) .. (301.9,35.88) .. controls (327.37,35.88) and (348.02,56.74) .. (348.02,82.49) .. controls (348.02,108.23) and (327.37,129.1) .. (301.9,129.1) .. controls (276.43,129.1) and (255.78,108.23) .. (255.78,82.49) -- cycle ;
					\draw   (40.27,82.49) .. controls (40.27,56.74) and (60.91,35.88) .. (86.39,35.88) .. controls (111.86,35.88) and (132.51,56.74) .. (132.51,82.49) .. controls (132.51,108.23) and (111.86,129.1) .. (86.39,129.1) .. controls (60.91,129.1) and (40.27,108.23) .. (40.27,82.49) -- cycle ;
					\draw  [fill={rgb, 255:red, 0; green, 0; blue, 0 }  ,fill opacity=1 ] (129.15,82.49) .. controls (129.15,80.55) and (130.65,78.97) .. (132.51,78.97) .. controls (134.36,78.97) and (135.86,80.55) .. (135.86,82.49) .. controls (135.86,84.43) and (134.36,86) .. (132.51,86) .. controls (130.65,86) and (129.15,84.43) .. (129.15,82.49) -- cycle ;
					\draw    (132.51,82.49) -- (160.43,82.49) ;
					\draw  [fill={rgb, 255:red, 0; green, 0; blue, 0 }  ,fill opacity=1 ] (157.08,82.49) .. controls (157.08,80.55) and (158.58,78.97) .. (160.43,78.97) .. controls (162.28,78.97) and (163.78,80.55) .. (163.78,82.49) .. controls (163.78,84.43) and (162.28,86) .. (160.43,86) .. controls (158.58,86) and (157.08,84.43) .. (157.08,82.49) -- cycle ;
					\draw  [fill={rgb, 255:red, 0; green, 0; blue, 0 }  ,fill opacity=1 ] (224.5,82.49) .. controls (224.5,80.55) and (226,78.97) .. (227.86,78.97) .. controls (229.71,78.97) and (231.21,80.55) .. (231.21,82.49) .. controls (231.21,84.43) and (229.71,86) .. (227.86,86) .. controls (226,86) and (224.5,84.43) .. (224.5,82.49) -- cycle ;
					\draw    (227.86,82.49) -- (255.78,82.49) ;
					\draw  [dash pattern={on 0.84pt off 2.51pt}]  (160.43,82.49) -- (227.86,82.49) ;
					\draw  [fill={rgb, 255:red, 208; green, 2; blue, 27 }  ,fill opacity=1 ] (53.56,47.05) .. controls (53.56,45.11) and (55.06,43.54) .. (56.91,43.54) .. controls (58.76,43.54) and (60.26,45.11) .. (60.26,47.05) .. controls (60.26,48.99) and (58.76,50.57) .. (56.91,50.57) .. controls (55.06,50.57) and (53.56,48.99) .. (53.56,47.05) -- cycle ;
					\draw  [fill={rgb, 255:red, 126; green, 211; blue, 33 }  ,fill opacity=1 ] (108.04,42.96) .. controls (108.04,41.02) and (109.54,39.45) .. (111.39,39.45) .. controls (113.24,39.45) and (114.75,41.02) .. (114.75,42.96) .. controls (114.75,44.9) and (113.24,46.48) .. (111.39,46.48) .. controls (109.54,46.48) and (108.04,44.9) .. (108.04,42.96) -- cycle ;
					\draw  [fill={rgb, 255:red, 126; green, 211; blue, 33 }  ,fill opacity=1 ] (39.25,97.48) .. controls (39.25,95.54) and (40.75,93.96) .. (42.61,93.96) .. controls (44.46,93.96) and (45.96,95.54) .. (45.96,97.48) .. controls (45.96,99.42) and (44.46,100.99) .. (42.61,100.99) .. controls (40.75,100.99) and (39.25,99.42) .. (39.25,97.48) -- cycle ;
					\draw  [fill={rgb, 255:red, 208; green, 2; blue, 27 }  ,fill opacity=1 ] (338.92,60.68) .. controls (338.92,58.74) and (340.42,57.16) .. (342.27,57.16) .. controls (344.13,57.16) and (345.63,58.74) .. (345.63,60.68) .. controls (345.63,62.62) and (344.13,64.2) .. (342.27,64.2) .. controls (340.42,64.2) and (338.92,62.62) .. (338.92,60.68) -- cycle ;
					\draw  [fill={rgb, 255:red, 0; green, 0; blue, 0 }  ,fill opacity=1 ] (642.26,82.37) .. controls (642.26,80.43) and (643.76,78.85) .. (645.61,78.85) .. controls (647.47,78.85) and (648.97,80.43) .. (648.97,82.37) .. controls (648.97,84.31) and (647.47,85.89) .. (645.61,85.89) .. controls (643.76,85.89) and (642.26,84.31) .. (642.26,82.37) -- cycle ;
					\draw   (645.61,82.37) .. controls (645.61,56.6) and (666.29,35.71) .. (691.8,35.71) .. controls (717.32,35.71) and (738,56.6) .. (738,82.37) .. controls (738,108.14) and (717.32,129.03) .. (691.8,129.03) .. controls (666.29,129.03) and (645.61,108.14) .. (645.61,82.37) -- cycle ;
					\draw   (429.77,82.37) .. controls (429.77,56.6) and (450.45,35.71) .. (475.96,35.71) .. controls (501.47,35.71) and (522.15,56.6) .. (522.15,82.37) .. controls (522.15,108.14) and (501.47,129.03) .. (475.96,129.03) .. controls (450.45,129.03) and (429.77,108.14) .. (429.77,82.37) -- cycle ;
					\draw  [fill={rgb, 255:red, 0; green, 0; blue, 0 }  ,fill opacity=1 ] (518.8,82.37) .. controls (518.8,80.43) and (520.3,78.85) .. (522.15,78.85) .. controls (524.01,78.85) and (525.51,80.43) .. (525.51,82.37) .. controls (525.51,84.31) and (524.01,85.89) .. (522.15,85.89) .. controls (520.3,85.89) and (518.8,84.31) .. (518.8,82.37) -- cycle ;
					\draw    (522.15,82.37) -- (550.12,82.37) ;
					\draw  [fill={rgb, 255:red, 0; green, 0; blue, 0 }  ,fill opacity=1 ] (546.76,82.37) .. controls (546.76,80.43) and (548.27,78.85) .. (550.12,78.85) .. controls (551.97,78.85) and (553.48,80.43) .. (553.48,82.37) .. controls (553.48,84.31) and (551.97,85.89) .. (550.12,85.89) .. controls (548.27,85.89) and (546.76,84.31) .. (546.76,82.37) -- cycle ;
					\draw  [fill={rgb, 255:red, 0; green, 0; blue, 0 }  ,fill opacity=1 ] (614.29,82.37) .. controls (614.29,80.43) and (615.79,78.85) .. (617.65,78.85) .. controls (619.5,78.85) and (621.01,80.43) .. (621.01,82.37) .. controls (621.01,84.31) and (619.5,85.89) .. (617.65,85.89) .. controls (615.79,85.89) and (614.29,84.31) .. (614.29,82.37) -- cycle ;
					\draw    (617.65,82.37) -- (645.61,82.37) ;
					\draw  [dash pattern={on 0.84pt off 2.51pt}]  (550.12,82.37) -- (617.65,82.37) ;
					\draw  [fill={rgb, 255:red, 208; green, 2; blue, 27 }  ,fill opacity=1 ] (443.08,46.89) .. controls (443.08,44.95) and (444.59,43.37) .. (446.44,43.37) .. controls (448.29,43.37) and (449.8,44.95) .. (449.8,46.89) .. controls (449.8,48.84) and (448.29,50.41) .. (446.44,50.41) .. controls (444.59,50.41) and (443.08,48.84) .. (443.08,46.89) -- cycle ;
					\draw  [fill={rgb, 255:red, 126; green, 211; blue, 33 }  ,fill opacity=1 ] (497.65,121.22) .. controls (497.65,119.27) and (499.15,117.7) .. (501.01,117.7) .. controls (502.86,117.7) and (504.37,119.27) .. (504.37,121.22) .. controls (504.37,123.16) and (502.86,124.74) .. (501.01,124.74) .. controls (499.15,124.74) and (497.65,123.16) .. (497.65,121.22) -- cycle ;
					\draw  [fill={rgb, 255:red, 126; green, 211; blue, 33 }  ,fill opacity=1 ] (579.97,81.93) .. controls (579.97,79.99) and (581.48,78.41) .. (583.33,78.41) .. controls (585.19,78.41) and (586.69,79.99) .. (586.69,81.93) .. controls (586.69,83.88) and (585.19,85.45) .. (583.33,85.45) .. controls (581.48,85.45) and (579.97,83.88) .. (579.97,81.93) -- cycle ;
					\draw  [fill={rgb, 255:red, 208; green, 2; blue, 27 }  ,fill opacity=1 ] (728.88,60.54) .. controls (728.88,58.59) and (730.39,57.02) .. (732.24,57.02) .. controls (734.1,57.02) and (735.6,58.59) .. (735.6,60.54) .. controls (735.6,62.48) and (734.1,64.06) .. (732.24,64.06) .. controls (730.39,64.06) and (728.88,62.48) .. (728.88,60.54) -- cycle ;
					\draw  [fill={rgb, 255:red, 0; green, 0; blue, 0 }  ,fill opacity=1 ] (451.07,340.92) .. controls (451.07,338.96) and (452.57,337.37) .. (454.43,337.37) .. controls (456.28,337.37) and (457.78,338.96) .. (457.78,340.92) .. controls (457.78,342.89) and (456.28,344.48) .. (454.43,344.48) .. controls (452.57,344.48) and (451.07,342.89) .. (451.07,340.92) -- cycle ;
					\draw   (454.43,340.92) .. controls (454.43,314.89) and (475.1,293.8) .. (500.6,293.8) .. controls (526.1,293.8) and (546.77,314.89) .. (546.77,340.92) .. controls (546.77,366.95) and (526.1,388.05) .. (500.6,388.05) .. controls (475.1,388.05) and (454.43,366.95) .. (454.43,340.92) -- cycle ;
					\draw   (238.69,340.92) .. controls (238.69,314.89) and (259.36,293.8) .. (284.86,293.8) .. controls (310.35,293.8) and (331.02,314.89) .. (331.02,340.92) .. controls (331.02,366.95) and (310.35,388.05) .. (284.86,388.05) .. controls (259.36,388.05) and (238.69,366.95) .. (238.69,340.92) -- cycle ;
					\draw  [fill={rgb, 255:red, 0; green, 0; blue, 0 }  ,fill opacity=1 ] (327.67,340.92) .. controls (327.67,338.96) and (329.17,337.37) .. (331.02,337.37) .. controls (332.88,337.37) and (334.38,338.96) .. (334.38,340.92) .. controls (334.38,342.89) and (332.88,344.48) .. (331.02,344.48) .. controls (329.17,344.48) and (327.67,342.89) .. (327.67,340.92) -- cycle ;
					\draw    (331.02,340.92) -- (358.98,340.92) ;
					\draw  [fill={rgb, 255:red, 0; green, 0; blue, 0 }  ,fill opacity=1 ] (355.62,340.92) .. controls (355.62,338.96) and (357.12,337.37) .. (358.98,337.37) .. controls (360.83,337.37) and (362.33,338.96) .. (362.33,340.92) .. controls (362.33,342.89) and (360.83,344.48) .. (358.98,344.48) .. controls (357.12,344.48) and (355.62,342.89) .. (355.62,340.92) -- cycle ;
					\draw  [fill={rgb, 255:red, 0; green, 0; blue, 0 }  ,fill opacity=1 ] (423.12,340.92) .. controls (423.12,338.96) and (424.62,337.37) .. (426.47,337.37) .. controls (428.33,337.37) and (429.83,338.96) .. (429.83,340.92) .. controls (429.83,342.89) and (428.33,344.48) .. (426.47,344.48) .. controls (424.62,344.48) and (423.12,342.89) .. (423.12,340.92) -- cycle ;
					\draw    (426.47,340.92) -- (454.43,340.92) ;
					\draw  [dash pattern={on 0.84pt off 2.51pt}]  (358.98,340.92) -- (426.47,340.92) ;
					\draw  [fill={rgb, 255:red, 208; green, 2; blue, 27 }  ,fill opacity=1 ] (251.99,305.09) .. controls (251.99,303.13) and (253.49,301.54) .. (255.35,301.54) .. controls (257.2,301.54) and (258.7,303.13) .. (258.7,305.09) .. controls (258.7,307.06) and (257.2,308.65) .. (255.35,308.65) .. controls (253.49,308.65) and (251.99,307.06) .. (251.99,305.09) -- cycle ;
					\draw  [fill={rgb, 255:red, 126; green, 211; blue, 33 }  ,fill opacity=1 ] (483.8,385.02) .. controls (483.8,383.06) and (485.3,381.47) .. (487.15,381.47) .. controls (489.01,381.47) and (490.51,383.06) .. (490.51,385.02) .. controls (490.51,386.98) and (489.01,388.57) .. (487.15,388.57) .. controls (485.3,388.57) and (483.8,386.98) .. (483.8,385.02) -- cycle ;
					\draw  [fill={rgb, 255:red, 126; green, 211; blue, 33 }  ,fill opacity=1 ] (482.43,296.14) .. controls (482.43,294.17) and (483.94,292.58) .. (485.79,292.58) .. controls (487.64,292.58) and (489.15,294.17) .. (489.15,296.14) .. controls (489.15,298.1) and (487.64,299.69) .. (485.79,299.69) .. controls (483.94,299.69) and (482.43,298.1) .. (482.43,296.14) -- cycle ;
					\draw  [fill={rgb, 255:red, 208; green, 2; blue, 27 }  ,fill opacity=1 ] (537.66,318.87) .. controls (537.66,316.91) and (539.16,315.32) .. (541.02,315.32) .. controls (542.87,315.32) and (544.37,316.91) .. (544.37,318.87) .. controls (544.37,320.84) and (542.87,322.43) .. (541.02,322.43) .. controls (539.16,322.43) and (537.66,320.84) .. (537.66,318.87) -- cycle ;
					
					\draw (108.25,74.6) node [anchor=north west][inner sep=0.75pt]  [font=\footnotesize]  {$\vv_{\n}$};
					\draw (69.17,155.71) node [anchor=north west][inner sep=0.75pt]  [font=\large]  {$\C_{\n}$};
					\draw (298.73,155.97) node [anchor=north west][inner sep=0.75pt]  [font=\large]  {$\C_{\m}$};
					\draw (121.28,33.78) node [anchor=north west][inner sep=0.75pt]  [font=\footnotesize]  {$\vv_{\aaa}$};
					\draw (20.29,88.99) node [anchor=north west][inner sep=0.75pt]  [font=\footnotesize]  {$\vv_{\bb}$};
					\draw (266.09,73.24) node [anchor=north west][inner sep=0.75pt]  [font=\footnotesize]  {$\vv_{\n+\rr}$};
					\draw (148.99,97.59) node [anchor=north west][inner sep=0.75pt]  [font=\footnotesize]  {$\vv_{\n+1}$};
					\draw (210.84,56.05) node [anchor=north west][inner sep=0.75pt]  [font=\footnotesize]  {$\vv_{\n+\rr-1}$};
					\draw (48.12,24.18) node [anchor=north west][inner sep=0.75pt]  [font=\footnotesize]  {$\vv_{\ii}$};
					\draw (349.62,52.8) node [anchor=north west][inner sep=0.75pt]  [font=\footnotesize]  {$\vv_{\jj}$};
					\draw (497.88,74.49) node [anchor=north west][inner sep=0.75pt]  [font=\footnotesize]  {$\vv_{\n}$};
					\draw (458.74,155.69) node [anchor=north west][inner sep=0.75pt]  [font=\large]  {$\C_{\n}$};
					\draw (688.65,155.95) node [anchor=north west][inner sep=0.75pt]  [font=\large]  {$\C_{\m}$};
					\draw (510.92,33.61) node [anchor=north west][inner sep=0.75pt]  [font=\footnotesize]  {$\vv_{\aaa}$};
					\draw (577.55,55.62) node [anchor=north west][inner sep=0.75pt]  [font=\footnotesize]  {$\vv_{\bb}$};
					\draw (655.96,73.12) node [anchor=north west][inner sep=0.75pt]  [font=\footnotesize]  {$\vv_{\n+\rr}$};
					\draw (538.68,97.5) node [anchor=north west][inner sep=0.75pt]  [font=\footnotesize]  {$\vv_{\n+1}$};
					\draw (600.63,97.5) node [anchor=north west][inner sep=0.75pt]  [font=\footnotesize]  {$\vv_{\n+\rr-1}$};
					\draw (437.65,24) node [anchor=north west][inner sep=0.75pt]  [font=\footnotesize]  {$\vv_{\ii}$};
					\draw (739.61,52.65) node [anchor=north west][inner sep=0.75pt]  [font=\footnotesize]  {$\vv_{\jj}$};
					\draw (306.81,336.22) node [anchor=north west][inner sep=0.75pt]  [font=\footnotesize]  {$\vv_{\n}$};
					\draw (267.64,415.08) node [anchor=north west][inner sep=0.75pt]  [font=\large]  {$\C_{\n}$};
					\draw (497.44,415.34) node [anchor=north west][inner sep=0.75pt]  [font=\large]  {$\C_{\m}$};
					\draw (483.47,391.85) node [anchor=north west][inner sep=0.75pt]  [font=\footnotesize]  {$\vv_{\aaa}$};
					\draw (478.51,274.04) node [anchor=north west][inner sep=0.75pt]  [font=\footnotesize]  {$\vv_{\bb}$};
					\draw (464.3,335.84) node [anchor=north west][inner sep=0.75pt]  [font=\footnotesize]  {$\vv_{\n+\rr}$};
					\draw (345.65,314.62) node [anchor=north west][inner sep=0.75pt]  [font=\footnotesize]  {$\vv_{\n+1}$};
					\draw (408.14,315.99) node [anchor=north west][inner sep=0.75pt]  [font=\footnotesize]  {$\vv_{\n+\rr-1}$};
					\draw (246.6,282.23) node [anchor=north west][inner sep=0.75pt]  [font=\footnotesize]  {$\vv_{\ii}$};
					\draw (548.48,311.17) node [anchor=north west][inner sep=0.75pt]  [font=\footnotesize]  {$\vv_{\jj}$};
					\draw (179,163.4) node [anchor=north west][inner sep=0.75pt]  [font=\LARGE]  {$\A$};
					\draw (579,162.4) node [anchor=north west][inner sep=0.75pt]  [font=\LARGE]  {$\B$};
					\draw (379,432.4) node [anchor=north west][inner sep=0.75pt]  [font=\LARGE]  {$\C$};

				\end{tikzpicture}

			}
			\caption{Possibilities of Positions of $\vv_{\aaa},\vv_{\bb}$ in $\C_{\n, \rr, \m}$.}
			\label{Cnrmallbasis}
		\end{center}
	\end{figure}
	In case $B$, since $\vv_{\bb}$ lies on the smallest $\vv_{\jj} \vv_{\aaa}$ path, we can easily conclude that $\dd(\vv_{\jj},\vv_{\aaa}) \neq \dd(\vv_{\jj},\vv_{\bb})$. In case $\C$, since $\vv_{\aaa},\vv_{\bb}$ are equidistant from $\vv_{\ii}$, they are equidistant from $\vv_{\n+\rr}$. Since $\m$ is odd, no matter which $\vv_{\jj} \in \C_{\m}$ we consider, we get $\dd(\vv_{\jj},\vv_{\aaa}) \neq \dd(\vv_{\jj},\vv_{\bb})$.
	
	Case $\A$ can be further divided into two cases. $\A-1$ where $\vv_{\aaa},\vv_{\bb}$ both lie in the same half of $\C_{\n}$ and case $\A-2$, where they lie in different halves of $\C_{\n}$.
	
	When both $\vv_{\aaa}, \vv_{\bb}$ lie in the same half of $\C_{\n}$, since they are equidistant from $\vv_{\ii}$, $\vv_{\ii}$ must be in the same half. From here we can easily conclude that $\dd(\vv_{\jj},\vv_{\aaa}) \neq \dd(\vv_{\jj},\vv_{\bb})$.
	
	In the case of $\vv_{\aaa},\vv_{\bb}$ from different halves of $\C_{\n}$, since $\n$ is even and, $\vv_{\ii} \neq \vv_{\n}$ and $\vv_{\ii} \neq \vv_{\frac{\n}{2}}$, we easily conclude that $\dd(\vv_{\aaa}, \vv_{\frac{\n}{2}}) \neq \dd(\vv_{\bb}, \vv_{\frac{\n}{2}})$. This once again ensures that $\dd(\vv_{\jj},\vv_{\aaa}) \neq \dd(\vv_{\jj},\vv_{\bb})$.
	
	When we consider $\vv_{\aaa},\vv_{\bb}$ to be equidistant from $\vv_{\jj}$, an argument along the same lines ensures that $\vv_{\ii}$ resolves these vertices.
	
	Both other cases where $\n,\m$ are both odd or even can be solved in a similar way. This concludes our result.
\end{proof}
	
\section{Fault Tolerant Metric Dimension of Bicyclic Generated Graphs}

Armed with the knowledge from above Lemmas, we are now ready to state and prove the results for fault tolerant metric basis and fault tolerant metric dimension of bicyclic graphs of type-I and type-II. 

\begin{Lemma}\label{lemm3}
	For any bicyclic graph of type-I and II, $\beta'(\G)>3$.
\end{Lemma}

\begin{proof}
	Since $\beta'(\C_{\n})=3$, we can easily infer that $\beta'(\G)\geq 3$ when $\G$ is a bicyclic graph of type-I or II. Since in both types of bicyclic graphs, we only have two cycles, let us denote them by $\C_{\n}$ and $\C_{\m}$. If we take $\beta'(\G)= 3$, then by pigeonhole principal, at least one of the cycles contributes two vertices to the fault tolerant metric basis. Let $\W=\{\vv_{\ii},\vv_{\jj}, \vv_{\kk}\}$ be the fault tolerant metric basis. Without loss of generality, we assume that $\vv_{\ii},\vv_{\jj} \in \C_{\n}$ and $\vv_{\kk} \in \C_{\m}$. Now, $\W$ can not be a fault tolerant resolving set since $\W-\{\vv_{\kk}\}$ does not resolve any vertex of $\C_{\m}$. This concludes our argument.
\end{proof}

\begin{Theorem}
	Let $\C_{\n,\m}$ be a bicyclic graph of type-I, then $\beta'(\C_{\n,\m})=4$.
\end{Theorem}

\begin{proof}
	We prove this assertion using two cases. 
	
	First, when at least one of $\n,\m$ is odd. Without loss of generality, let $\n$ be odd and $\m$ be even, then by Lemma \ref{lemm1}, $\W_1=\left\{\vv_{\left \lfloor \frac{\n}{2}\right \rfloor },\vv_{\jj}\right\}$, $\vv_{\jj} \in \C_{\m}-\left\{\vv_{\n}, \vv_{\n+\frac{\m}{2}}\right\}$, is a metric basis for $\C_{\n,\m}$. Since $\m > 3$, we can choose $\vv_{\jj}$ to be $\vv_{\n+\left \lfloor \frac{\m}{2}\right \rfloor-1 }$. Again, by Lemma \ref{lemm1}, $\W_2=\left\{\vv_{\left \lfloor \frac{\n}{2}\right \rfloor +1}, \vv_{\kk}\right\}$, $\vv_{\kk} \in \C_{\n}-\left\{\vv_{\n}, \vv_{\n+\frac{\m}{2}}\right\}$ is a metric basis for $\C_{\n,\m}$. Let $\vv_{\kk} =v_{\n+\left \lfloor \frac{\m}{2}\right \rfloor+1 }$. It is evident from the construction of $\W_1$ and $\W_2$ that $\W_1 \cap \W_2= \Phi$. We can also conclude that the set $\W=\W_1 \cup \W_2=\left\{\vv_{\left \lfloor \frac{\n}{2}\right \rfloor }, \vv_{\n+\left \lfloor \frac{\m}{2}\right \rfloor-1 }, \vv_{\left \lfloor \frac{\n}{2}\right \rfloor +1}, \vv_{\n+\left \lfloor \frac{\m}{2}\right \rfloor+1 }\right\}$ is a fault tolerant resolving set. Now, by Lemma \ref{lemm3}, $\beta'(\C_{\n,\m}) \geq 4$. This ensures that $\W$ is the smallest fault tolerant resolving set and hence $\beta'(\C_{\n,\m})=4$.
	
	Next, we consider the case when $\n,\m$ are both even. Let us consider the partitioning of $\V(\C_{\n,\m})$ as given in Equation (\ref{eq1}). Using Lemma \ref{evenevenallbasis}, let us construct a metric basis $\W_1=\{\vv_{\ii},\vv_{\jj},\vv_{\kk}\}$ where $\vv_{\ii} \in \V_1, \vv_{\jj} \in \V_3$ and $\vv_{\kk} \in \V_5$. Using the same lemma, let us construct another metric basis $\W_2=\{\vv_l, \vv_{\kk}, \vv_{\jj}\}$ where $\vv_{\el}\in \V_7, \vv_{\kk} \in \V_5$ and $\vv_{\jj} \in \V_3$. Now $\W=\W_1 \cup \W_2=\{\vv_{\ii},\vv_{\jj},\vv_{\kk},\vv_{\el}\}$ is a fault tolerant resolving set. Again, by Lemma \ref{lemm3}, $\beta'(\C_{\n,\m}) \geq 4$, and hence $\beta'(\C_{\n,\m})=4$.
\end{proof}

\begin{Theorem}\label{FTMDBi2}
	For bicyclic graph of type II, $\beta'(\C_{\n, \rr, \m})=4$.
\end{Theorem}
\begin{proof}
	Let $\vv_{\ii}, \vv_{\kk} \in \C_{\n}$ such that $\vv_{\ii}\neq \vv_{\kk} \neq \vv_{\n}$ for all $\n$ and $\vv_{\ii}\neq \vv_{\kk} \neq \vv_{\frac{\n}{2}}$ for even $\n$. Similarly, let $\vv_{\jj}, \vv_l \in \C_{\m}$ such that $\vv_{\jj}\neq \vv_l  \neq \vv_{\n+\rr}$ for all $\n, \rr$, and $\vv_{\jj}\neq \vv_l  \neq \vv_{\n+\rr+\frac{\m}{2}}$ for all $\n, \rr $ and even $\m$. By Lemma \ref{type2metricbasis}, $\W_1=\{\vv_{\ii},\vv_{\jj}\}$ and $\W_2=\{\vv_{\kk} , \vv_l \}$ are metric basis for $\C_{\n, \rr, \m}$. Now $\W=\W_1 \cup \W_2$ is a fault tolerant resolving set and since by Lemma \ref{lemm3}, $\beta'(\C_{\n, \rr, \m}) \geq 4$, we obtain $\beta'(\C_{\n, \rr, \m})=4$.
\end{proof}

\section{Fault Tolerant Metric Dimension of Cacti Graph without Leaves}

From the above results for bicyclic graphs, we see that $\left|\SSS_{\C_{\ii}}\right|=2$ for any cycle of bicyclic graphs of type-I or II. Before extending this concept to cacti graphs without leaves, we define some new terminology for such graphs.

Let $\G$ be a cactus graph without leaves and let $\A(\G) \subset \V(\G)$ be the set of all vertices with degree $3$ or higher. Let $\A(\C_{\ii})=\V(\C_{\ii}) \cap \A(\G)$, then we say that a cycle $\C_{\ii}$ of $\G$ is an outer cycle, if $\left|\A(\C_{\ii})\right|=1$. Similarly, a cycle is called an inner cycle if $\left|\A(\C_{\ii})\right|\geq 2$. Vertices in $\A(\C_{\ii})$ are termed as common vertices.

It should be noted that if $\vv$ is a common vertex of an inner cycle $\C_{\m}$, then $\vv$ is also a vertex of another cycle or path. If $\vv$ is a vertex of another cycle $\C_{\n}$, where $\left|\A(\C_{\n})\right|=1$, then we say that the cycle $\C_{\n}$ is a corresponding outer cycle of $\vv$. If $\left|\A(\C_{\n})\right|\geq 2$, then we continue along any of the other common vertices of $\C_{\n}$ until we come across an outer cycle, and call it, the outer cycle corresponding to the vertex $\vv$. On the other hand, if $\vv$ is vertex of a path, and since we are working with cacti graphs without leaves, that path must be attached to another cycle. Now, applying the same logic as above, we can easily find an outer cycle corresponding to the vertex $\vv$.

It is worthwhile to mention that a common vertex $\vv$ of an inner cycle may have more than one corresponding outer cycles. These concepts are presented in the following example.

\begin{Example}
In the figure \ref{example}, we see that the cycles $\C_1$, $\C_3$, $\C_5$ and $\C_6$ are external cycles, since they all have only one common vertex. On the other hand, cycles $\C_2$ and $\C_4$ are internal cycles, both having three common vertices. Considering the common vertices of $\C_2$, we see that there are $4$ corresponding outer cycles, $1$ each for vertices on the left and right and $2$ for the remaining common vertex.

\begin{figure}[H]
	\begin{center}
		\tikzset{every picture/.style={line width=0.75pt}} 
		\resizebox{6.5cm}{4cm}{%

\begin{tikzpicture}[x=0.75pt,y=0.75pt,yscale=-1,xscale=1]
	
	\draw  [fill={rgb, 255:red, 0; green, 0; blue, 0 }  ,fill opacity=1 ] (145.41,30) .. controls (145.41,27.15) and (147.61,24.84) .. (150.33,24.84) .. controls (153.05,24.84) and (155.26,27.15) .. (155.26,30) .. controls (155.26,32.85) and (153.05,35.16) .. (150.33,35.16) .. controls (147.61,35.16) and (145.41,32.85) .. (145.41,30) -- cycle ;
	\draw  [fill={rgb, 255:red, 0; green, 0; blue, 0 }  ,fill opacity=1 ] (185.41,140) .. controls (185.41,137.15) and (187.61,134.84) .. (190.33,134.84) .. controls (193.05,134.84) and (195.26,137.15) .. (195.26,140) .. controls (195.26,142.85) and (193.05,145.16) .. (190.33,145.16) .. controls (187.61,145.16) and (185.41,142.85) .. (185.41,140) -- cycle ;
	\draw  [fill={rgb, 255:red, 0; green, 0; blue, 0 }  ,fill opacity=1 ] (245.41,79) .. controls (245.41,76.15) and (247.61,73.84) .. (250.33,73.84) .. controls (253.05,73.84) and (255.26,76.15) .. (255.26,79) .. controls (255.26,81.85) and (253.05,84.16) .. (250.33,84.16) .. controls (247.61,84.16) and (245.41,81.85) .. (245.41,79) -- cycle ;
	\draw  [fill={rgb, 255:red, 0; green, 0; blue, 0 }  ,fill opacity=1 ] (345.41,30) .. controls (345.41,27.15) and (347.61,24.84) .. (350.33,24.84) .. controls (353.05,24.84) and (355.26,27.15) .. (355.26,30) .. controls (355.26,32.85) and (353.05,35.16) .. (350.33,35.16) .. controls (347.61,35.16) and (345.41,32.85) .. (345.41,30) -- cycle ;
	\draw    (150.33,30) -- (200.33,79) ;
	\draw    (350.33,190) -- (400.33,240) ;
	\draw    (350.33,30) -- (400.33,79) ;
	\draw    (300.33,79) -- (350.33,130) ;
	\draw    (400.33,79) -- (450.33,130) ;
	\draw    (450.33,29) -- (400.33,79) ;
	\draw    (350.33,30) -- (300.33,79) ;
	\draw    (400.33,79) -- (350.33,130) ;
	\draw    (150.33,30) -- (100.33,79) ;
	\draw    (350.33,190) -- (300.33,240) ;
	\draw    (350.33,130) -- (350.33,190) ;
	\draw    (200.33,79) -- (250.33,79) ;
	\draw    (250.33,79) -- (300.33,79) ;
	\draw    (300.33,240) -- (400.33,240) ;
	\draw    (110.33,140) -- (190.33,140) ;
	\draw    (190.33,140) -- (200.33,79) ;
	\draw    (110.33,140) -- (100.33,79) ;
	\draw  [fill={rgb, 255:red, 0; green, 0; blue, 0 }  ,fill opacity=1 ] (495.41,30) .. controls (495.41,27.15) and (497.61,24.84) .. (500.33,24.84) .. controls (503.05,24.84) and (505.26,27.15) .. (505.26,30) .. controls (505.26,32.85) and (503.05,35.16) .. (500.33,35.16) .. controls (497.61,35.16) and (495.41,32.85) .. (495.41,30) -- cycle ;
	\draw  [fill={rgb, 255:red, 0; green, 0; blue, 0 }  ,fill opacity=1 ] (496.41,130) .. controls (496.41,127.15) and (498.61,124.84) .. (501.33,124.84) .. controls (504.05,124.84) and (506.26,127.15) .. (506.26,130) .. controls (506.26,132.85) and (504.05,135.16) .. (501.33,135.16) .. controls (498.61,135.16) and (496.41,132.85) .. (496.41,130) -- cycle ;
	\draw  [fill={rgb, 255:red, 0; green, 0; blue, 0 }  ,fill opacity=1 ] (545.41,80) .. controls (545.41,77.15) and (547.61,74.84) .. (550.33,74.84) .. controls (553.05,74.84) and (555.26,77.15) .. (555.26,80) .. controls (555.26,82.85) and (553.05,85.16) .. (550.33,85.16) .. controls (547.61,85.16) and (545.41,82.85) .. (545.41,80) -- cycle ;
	\draw    (500.33,30) -- (550.33,80) ;
	\draw    (501.33,130) -- (550.33,80) ;
	\draw    (450.33,29) -- (500.33,30) ;
	\draw    (451.33,129) -- (501.33,130) ;
	\draw    (450.33,190) -- (500.33,240) ;
	\draw    (450.33,190) -- (400.33,240) ;
	\draw    (400.33,240) -- (500.33,240) ;
	\draw  [fill={rgb, 255:red, 0; green, 0; blue, 0 }  ,fill opacity=1 ] (197.41,239) .. controls (197.41,236.15) and (199.61,233.84) .. (202.33,233.84) .. controls (205.05,233.84) and (207.26,236.15) .. (207.26,239) .. controls (207.26,241.85) and (205.05,244.16) .. (202.33,244.16) .. controls (199.61,244.16) and (197.41,241.85) .. (197.41,239) -- cycle ;
	\draw    (252.33,190) -- (302.33,239) ;
	\draw    (202.33,239) -- (252.33,290) ;
	\draw    (252.33,190) -- (202.33,239) ;
	\draw    (302.33,239) -- (252.33,290) ;
	\draw  [fill={rgb, 255:red, 0; green, 0; blue, 0 }  ,fill opacity=1 ] (105.41,140) .. controls (105.41,137.15) and (107.61,134.84) .. (110.33,134.84) .. controls (113.05,134.84) and (115.26,137.15) .. (115.26,140) .. controls (115.26,142.85) and (113.05,145.16) .. (110.33,145.16) .. controls (107.61,145.16) and (105.41,142.85) .. (105.41,140) -- cycle ;
	\draw  [fill={rgb, 255:red, 0; green, 0; blue, 0 }  ,fill opacity=1 ] (95.41,79) .. controls (95.41,76.15) and (97.61,73.84) .. (100.33,73.84) .. controls (103.05,73.84) and (105.26,76.15) .. (105.26,79) .. controls (105.26,81.85) and (103.05,84.16) .. (100.33,84.16) .. controls (97.61,84.16) and (95.41,81.85) .. (95.41,79) -- cycle ;
	\draw  [fill={rgb, 255:red, 0; green, 0; blue, 0 }  ,fill opacity=1 ] (445.41,29) .. controls (445.41,26.15) and (447.61,23.84) .. (450.33,23.84) .. controls (453.05,23.84) and (455.26,26.15) .. (455.26,29) .. controls (455.26,31.85) and (453.05,34.16) .. (450.33,34.16) .. controls (447.61,34.16) and (445.41,31.85) .. (445.41,29) -- cycle ;
	\draw  [fill={rgb, 255:red, 0; green, 0; blue, 0 }  ,fill opacity=1 ] (445.41,129) .. controls (445.41,126.15) and (447.61,123.84) .. (450.33,123.84) .. controls (453.05,123.84) and (455.26,126.15) .. (455.26,129) .. controls (455.26,131.85) and (453.05,134.16) .. (450.33,134.16) .. controls (447.61,134.16) and (445.41,131.85) .. (445.41,129) -- cycle ;
	\draw  [fill={rgb, 255:red, 0; green, 0; blue, 0 }  ,fill opacity=1 ] (445.41,190) .. controls (445.41,187.15) and (447.61,184.84) .. (450.33,184.84) .. controls (453.05,184.84) and (455.26,187.15) .. (455.26,190) .. controls (455.26,192.85) and (453.05,195.16) .. (450.33,195.16) .. controls (447.61,195.16) and (445.41,192.85) .. (445.41,190) -- cycle ;
	\draw  [fill={rgb, 255:red, 0; green, 0; blue, 0 }  ,fill opacity=1 ] (495.41,240) .. controls (495.41,237.15) and (497.61,234.84) .. (500.33,234.84) .. controls (503.05,234.84) and (505.26,237.15) .. (505.26,240) .. controls (505.26,242.85) and (503.05,245.16) .. (500.33,245.16) .. controls (497.61,245.16) and (495.41,242.85) .. (495.41,240) -- cycle ;
	\draw  [fill={rgb, 255:red, 0; green, 0; blue, 0 }  ,fill opacity=1 ] (247.41,190) .. controls (247.41,187.15) and (249.61,184.84) .. (252.33,184.84) .. controls (255.05,184.84) and (257.26,187.15) .. (257.26,190) .. controls (257.26,192.85) and (255.05,195.16) .. (252.33,195.16) .. controls (249.61,195.16) and (247.41,192.85) .. (247.41,190) -- cycle ;
	\draw  [fill={rgb, 255:red, 0; green, 0; blue, 0 }  ,fill opacity=1 ] (247.41,290) .. controls (247.41,287.15) and (249.61,284.84) .. (252.33,284.84) .. controls (255.05,284.84) and (257.26,287.15) .. (257.26,290) .. controls (257.26,292.85) and (255.05,295.16) .. (252.33,295.16) .. controls (249.61,295.16) and (247.41,292.85) .. (247.41,290) -- cycle ;
	\draw  [fill={rgb, 255:red, 65; green, 117; blue, 5 }  ,fill opacity=1 ] (195.41,79) .. controls (195.41,76.15) and (197.61,73.84) .. (200.33,73.84) .. controls (203.05,73.84) and (205.26,76.15) .. (205.26,79) .. controls (205.26,81.85) and (203.05,84.16) .. (200.33,84.16) .. controls (197.61,84.16) and (195.41,81.85) .. (195.41,79) -- cycle ;
	\draw  [fill={rgb, 255:red, 65; green, 117; blue, 5 }  ,fill opacity=1 ] (295.41,79) .. controls (295.41,76.15) and (297.61,73.84) .. (300.33,73.84) .. controls (303.05,73.84) and (305.26,76.15) .. (305.26,79) .. controls (305.26,81.85) and (303.05,84.16) .. (300.33,84.16) .. controls (297.61,84.16) and (295.41,81.85) .. (295.41,79) -- cycle ;
	\draw  [fill={rgb, 255:red, 65; green, 117; blue, 5 }  ,fill opacity=1 ] (395.41,79) .. controls (395.41,76.15) and (397.61,73.84) .. (400.33,73.84) .. controls (403.05,73.84) and (405.26,76.15) .. (405.26,79) .. controls (405.26,81.85) and (403.05,84.16) .. (400.33,84.16) .. controls (397.61,84.16) and (395.41,81.85) .. (395.41,79) -- cycle ;
	\draw  [fill={rgb, 255:red, 65; green, 117; blue, 5 }  ,fill opacity=1 ] (345.41,190) .. controls (345.41,187.15) and (347.61,184.84) .. (350.33,184.84) .. controls (353.05,184.84) and (355.26,187.15) .. (355.26,190) .. controls (355.26,192.85) and (353.05,195.16) .. (350.33,195.16) .. controls (347.61,195.16) and (345.41,192.85) .. (345.41,190) -- cycle ;
	\draw  [fill={rgb, 255:red, 65; green, 117; blue, 5 }  ,fill opacity=1 ] (395.41,240) .. controls (395.41,237.15) and (397.61,234.84) .. (400.33,234.84) .. controls (403.05,234.84) and (405.26,237.15) .. (405.26,240) .. controls (405.26,242.85) and (403.05,245.16) .. (400.33,245.16) .. controls (397.61,245.16) and (395.41,242.85) .. (395.41,240) -- cycle ;
	\draw  [fill={rgb, 255:red, 65; green, 117; blue, 5 }  ,fill opacity=1 ] (345.41,130) .. controls (345.41,127.15) and (347.61,124.84) .. (350.33,124.84) .. controls (353.05,124.84) and (355.26,127.15) .. (355.26,130) .. controls (355.26,132.85) and (353.05,135.16) .. (350.33,135.16) .. controls (347.61,135.16) and (345.41,132.85) .. (345.41,130) -- cycle ;
	\draw  [fill={rgb, 255:red, 65; green, 117; blue, 5 }  ,fill opacity=1 ] (295.41,239) .. controls (295.41,236.15) and (297.61,233.84) .. (300.33,233.84) .. controls (303.05,233.84) and (305.26,236.15) .. (305.26,239) .. controls (305.26,241.85) and (303.05,244.16) .. (300.33,244.16) .. controls (297.61,244.16) and (295.41,241.85) .. (295.41,239) -- cycle ;
	
	\draw (141,76.4) node [anchor=north west][inner sep=0.75pt]    {$\C_{1}$};
	\draw (345,75.4) node [anchor=north west][inner sep=0.75pt]    {$\C_{2}$};
	\draw (469,75.4) node [anchor=north west][inner sep=0.75pt]    {$\C_{3}$};
	\draw (343,215.4) node [anchor=north west][inner sep=0.75pt]    {$\C_{4}$};
	\draw (244,92.4) node [anchor=north west][inner sep=0.75pt]    {$\PP_{1}$};
	\draw (366,151.4) node [anchor=north west][inner sep=0.75pt]    {$\PP_{2}$};
	\draw (443,215.4) node [anchor=north west][inner sep=0.75pt]    {$\C_{5}$};
	\draw (247,235.4) node [anchor=north west][inner sep=0.75pt]    {$\C_{6}$};

\end{tikzpicture}
}
\caption{Cacti Graph without Leaves}
\label{example}
\end{center}
\end{figure}
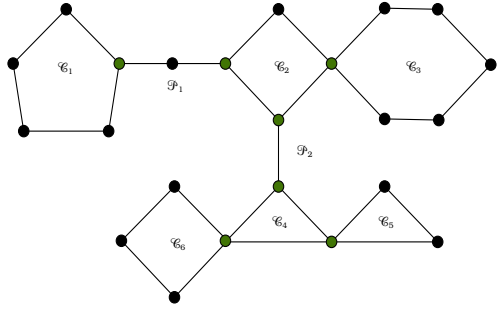

\end{Example}

Now, we are ready to extend the concept of fault tolerant metric generator to cacti graphs without leaves.

\begin{Lemma} \label{outercycle1}
	Let $\G$ be a cactus graph without leaves then $\left|\SSS_{\C_{\m}}\right|\geq 2$ for every outer cycle $\C_{\m}$.
\end{Lemma}
\begin{proof}
	Let $\G$ be a cactus graph and let us suppose on the contrary that there is an outer cycle $\C_{\m}$ contributing one or no vertices to fault tolerant metric basis $\W$.	
	If $\C_{\m}$ contributes only one vertex, say $\mathcal{c}$, then $\W-\{\mathcal{c}\}$ does not resolve vertices of $\C_{\m}$ and hence $\W$ is not a fault tolerant resolving set.	
	Similarly if $\C_{\m}$ does not contribute any vertex to $\W$, again, vertices of $\C_{\m}$ are not resolved by $\W$ in the first place. This ensures that every outer cycle must contribute two or more vertices to the fault tolerant metric basis.
\end{proof}

In fact,we can prove a much stronger result for $\left|\SSS_{\C_{\m}}\right|$ for outer cycles $\C_{\m}$.

\begin{Lemma}\label{outercycle2}
	Let $\G$ be a cactus graph without leaves then $\left|\SSS_{\C_{\m}}\right| = 2$ for every outer cycle $\C_{\m}$.
\end{Lemma}
\begin{proof}
	Let $\G$ be a cactus graph without leaves having $\W$ as fault tolerant metric basis, and let $\C_{\m}$ be an outer cycle of $\G$. Let us denote the common vertex of $\C_{\m}$ by $\vv_{\m}$, then there is at least one outer cycle, say $\C_{\n}$, corresponding to this common vertex. By Lemma \ref{outercycle1}, $\C_{\n}$ contributes at least two vertices to the fault tolerant metric basis set $\W$ of $\G$. Let $\SSS_{\C_{\n}}$ be the set of vertices contributed by $\C_{\n}$ to $\W$.
	
	Next, we move to the cycle $\C_{\m}$ and consider the vertices $\vv_{\ii}$ and $\vv_{\jj}$ adjacent to $\vv_{\m}$. Obviously these vertices lie to the either side of $\vv_{\m}$. These ideas are shown in figure \ref{outercycleadds2}.

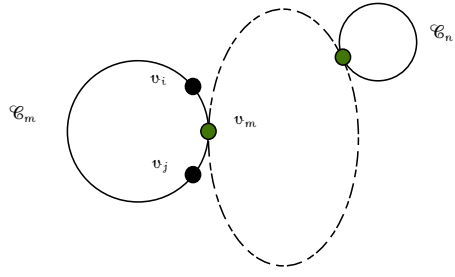
\begin{figure}[H]
	\begin{center}
		\tikzset{every picture/.style={line width=0.75pt}} 
		\resizebox{6cm}{3.5cm}{%
			
			\begin{tikzpicture}[x=0.75pt,y=0.75pt,yscale=-1,xscale=1]
		
		\draw   (160.08,135.67) .. controls (160.08,110.45) and (180.52,90) .. (205.74,90) .. controls (230.97,90) and (251.41,110.45) .. (251.41,135.67) .. controls (251.41,160.89) and (230.97,181.33) .. (205.74,181.33) .. controls (180.52,181.33) and (160.08,160.89) .. (160.08,135.67) -- cycle ;
		\draw  [dash pattern={on 3.75pt off 3pt on 7.5pt off 1.5pt}] (251.33,139.67) .. controls (251.33,93.83) and (273.05,56.67) .. (299.83,56.67) .. controls (326.62,56.67) and (348.33,93.83) .. (348.33,139.67) .. controls (348.33,185.51) and (326.62,222.67) .. (299.83,222.67) .. controls (273.05,222.67) and (251.33,185.51) .. (251.33,139.67) -- cycle ;
		\draw   (336,78) .. controls (336,64.19) and (347.19,53) .. (361,53) .. controls (374.81,53) and (386,64.19) .. (386,78) .. controls (386,91.81) and (374.81,103) .. (361,103) .. controls (347.19,103) and (336,91.81) .. (336,78) -- cycle ;
		\draw  [fill={rgb, 255:red, 65; green, 117; blue, 5 }  ,fill opacity=1 ] (246.41,135.67) .. controls (246.41,132.82) and (248.61,130.51) .. (251.33,130.51) .. controls (254.05,130.51) and (256.26,132.82) .. (256.26,135.67) .. controls (256.26,138.52) and (254.05,140.83) .. (251.33,140.83) .. controls (248.61,140.83) and (246.41,138.52) .. (246.41,135.67) -- cycle ;
		\draw  [fill={rgb, 255:red, 65; green, 117; blue, 5 }  ,fill opacity=1 ] (333.41,87.67) .. controls (333.41,84.82) and (335.61,82.51) .. (338.33,82.51) .. controls (341.05,82.51) and (343.26,84.82) .. (343.26,87.67) .. controls (343.26,90.52) and (341.05,92.83) .. (338.33,92.83) .. controls (335.61,92.83) and (333.41,90.52) .. (333.41,87.67) -- cycle ;
		\draw  [fill={rgb, 255:red, 0; green, 0; blue, 0 }  ,fill opacity=1 ] (236.41,106.67) .. controls (236.41,103.82) and (238.61,101.51) .. (241.33,101.51) .. controls (244.05,101.51) and (246.26,103.82) .. (246.26,106.67) .. controls (246.26,109.52) and (244.05,111.83) .. (241.33,111.83) .. controls (238.61,111.83) and (236.41,109.52) .. (236.41,106.67) -- cycle ;
		\draw  [fill={rgb, 255:red, 0; green, 0; blue, 0 }  ,fill opacity=1 ] (236.41,163.67) .. controls (236.41,160.82) and (238.61,158.51) .. (241.33,158.51) .. controls (244.05,158.51) and (246.26,160.82) .. (246.26,163.67) .. controls (246.26,166.52) and (244.05,168.83) .. (241.33,168.83) .. controls (238.61,168.83) and (236.41,166.52) .. (236.41,163.67) -- cycle ;
		
		\draw (121,117.4) node [anchor=north west][inner sep=0.75pt]    {$\C_{\m}$};
		\draw (393,65.4) node [anchor=north west][inner sep=0.75pt]    {$\C_{\n}$};
		\draw (266,124.4) node [anchor=north west][inner sep=0.75pt]  [font=\small]  {$\vv_{\m}$};
		\draw (212,97.4) node [anchor=north west][inner sep=0.75pt]  [font=\small]  {$\vv_{\ii}$};
		\draw (213,152.4) node [anchor=north west][inner sep=0.75pt]  [font=\small]  {$\vv_{\jj}$};

	\end{tikzpicture}
	
}
\caption{Fault Tolerant Metric Generator of an Outer Cycle $\C_{\m}$}
\label{outercycleadds2}
\end{center}
\end{figure}

We now claim that all vertices of the cycle $\C_{\m}$ are fault tolerantly resolved by the set $\W'=\SSS_{\C_{\n}} \cup \{\vv_{\ii},\vv_{\jj}\}$. Before proving this claim, we observe that the set $\W'-\{\vv_{\kk}\}, \vv_{\kk} \in \W'$, always contains at least one vertex of $\C_{\n}$, since $\left|\SSS_{\C_{\n}}\right| \geq 2$. This ensures that when we consider the distance vector of a vertex $\vv_{\aaa} \in \C_{\m}$ with respect to $\W'-\{\vv_{\kk}\}, \vv_{\kk} \in \W'$, there is at least one distance entry corresponding to a vertex of $\C_{\n}$. The shortest path for this distance always passes through the common vertex $\vv_{\m}$. 

We now consider the set $\W''=\{\vv_{\ii}, \vv_{\m},\vv_{\jj}\}$, and check it for fault tolerant distance producer of $\C_{\m}$. Note that, we won't delete $\vv_{\m}$, to check for fault tolerant metric generator, due to the argument provided above.

Now let us consider two vertices equidistant from $\vv_{\m}$, then obviously, they are at equal distances, both from $\vv_{\ii}$ and $\vv_{\jj}$. Since we can not delete $\vv_{\m}$, we proceed to deleting any one of $\vv_{\ii},\vv_{\jj}$ to check for fault tolerance. Deleting any one of these leaves us with two adjacent vertices of $\C_{\m}$, and since all vertices of a cycle are always resolved by a pair of adjacent vertices, $\W''$ is a fault tolerant distance producer in this case.

In the other two remaining cases, where vertices are equidistant from $\vv_{\ii}$ or $\vv_{\jj}$, a similar argument can easily be applied. Hence, $\W''$ containing vertex $\vv_{\m}$ is always a fault tolerant distance producer for $\C_{\m}$. 

Since all shortest paths from $\C_{\n}$ to the vertices of $\C_{\m}$ pass through $\vv_{\m}$, after replacing the vertex $\vv_{\m}$ with vertices of $\SSS_{\C_{\n}}$, the set $\W'=\SSS_{\C_{\n}} \cup \{\vv_{\ii},\vv_{\jj}\}$ is still a fault tolerant distance producer for $\C_{\m}$ and hence, $\SSS_{\C_{\m}}=\{\vv_{\ii},\vv_{\jj}\}$ and $\left|\SSS_{\C_{\m}}\right| \leq 2$.

Using the result from Lemma \ref{outercycle1}, we now conclude that $\left|\SSS_{\C_{\m}}\right| =2 $.

\end{proof}

Next, we turn our attention to inner cycles. Depending on the cycle being odd or even, we have different possibilities for $\left|\SSS_{\C_{\m}}\right|$.

\begin{Lemma} \label{oddcycleincactiresult}
	Let $\G$ be a cactus graph without leaves, then $\left|\SSS_{\C_{\m}}\right|=0$ for every inner cycle $\C_{\m}$ of odd length.
\end{Lemma}
\begin{proof}
	Let $\G$ and $\C_{\m}$ be as given above. Since $\C_{\m}$ is an inner cycle, we have, $\left|\V(\C_{\m}) \cap \A(\G)\right|\geq 2$. Without loss of generality, we choose any two of these common vertices. Let $\vv_{\ii}, \vv_{\jj} \in \C_{\m}$ be the common vertices as shown in figure \ref{oddcycleincacti}.

	\begin{figure}[H]
		\begin{center}
			\tikzset{every picture/.style={line width=0.75pt}} 
			\resizebox{10cm}{3.7cm}{%

\begin{tikzpicture}[x=0.75pt,y=0.75pt,yscale=-1,xscale=1]
	
	\draw  [fill={rgb, 255:red, 0; green, 0; blue, 0 }  ,fill opacity=1 ] (336.98,73.29) .. controls (336.98,70.44) and (339.18,68.13) .. (341.9,68.13) .. controls (344.62,68.13) and (346.82,70.44) .. (346.82,73.29) .. controls (346.82,76.14) and (344.62,78.45) .. (341.9,78.45) .. controls (339.18,78.45) and (336.98,76.14) .. (336.98,73.29) -- cycle ;
	\draw   (224.46,120.29) .. controls (224.46,82.51) and (254.78,51.89) .. (292.18,51.89) .. controls (329.58,51.89) and (359.9,82.51) .. (359.9,120.29) .. controls (359.9,158.06) and (329.58,188.69) .. (292.18,188.69) .. controls (254.78,188.69) and (224.46,158.06) .. (224.46,120.29) -- cycle ;
	\draw  [fill={rgb, 255:red, 0; green, 0; blue, 0 }  ,fill opacity=1 ] (220.59,107.09) .. controls (220.59,104.24) and (222.79,101.93) .. (225.51,101.93) .. controls (228.23,101.93) and (230.43,104.24) .. (230.43,107.09) .. controls (230.43,109.94) and (228.23,112.25) .. (225.51,112.25) .. controls (222.79,112.25) and (220.59,109.94) .. (220.59,107.09) -- cycle ;
	\draw  [dash pattern={on 4.5pt off 4.5pt}] (91.95,53.27) .. controls (100.61,34.72) and (137.53,31.74) .. (174.41,46.61) .. controls (211.29,61.47) and (234.17,88.55) .. (225.51,107.09) .. controls (216.84,125.63) and (179.92,128.61) .. (143.04,113.75) .. controls (106.16,98.88) and (83.28,71.81) .. (91.95,53.27) -- cycle ;
	\draw    (341.9,73.29) .. controls (381.9,43.29) and (401.9,103.29) .. (441.9,73.29) ;
	\draw  [dash pattern={on 4.5pt off 4.5pt}] (441.9,73.29) .. controls (438.4,55.96) and (467.45,35.48) .. (506.79,27.53) .. controls (546.13,19.59) and (580.86,27.19) .. (584.35,44.51) .. controls (587.85,61.83) and (558.8,82.32) .. (519.46,90.27) .. controls (480.12,98.21) and (445.4,90.61) .. (441.9,73.29) -- cycle ;
	\draw  [fill={rgb, 255:red, 0; green, 0; blue, 0 }  ,fill opacity=1 ] (436.98,73.29) .. controls (436.98,70.44) and (439.18,68.13) .. (441.9,68.13) .. controls (444.62,68.13) and (446.82,70.44) .. (446.82,73.29) .. controls (446.82,76.14) and (444.62,78.45) .. (441.9,78.45) .. controls (439.18,78.45) and (436.98,76.14) .. (436.98,73.29) -- cycle ;
	\draw  [fill={rgb, 255:red, 0; green, 0; blue, 0 }  ,fill opacity=1 ] (515.54,90.42) .. controls (515.54,87.58) and (517.74,85.27) .. (520.46,85.27) .. controls (523.18,85.27) and (525.38,87.58) .. (525.38,90.42) .. controls (525.38,93.27) and (523.18,95.58) .. (520.46,95.58) .. controls (517.74,95.58) and (515.54,93.27) .. (515.54,90.42) -- cycle ;
	\draw  [dash pattern={on 4.5pt off 4.5pt}] (551.07,199.49) .. controls (533.52,204.57) and (512.21,184.24) .. (503.49,154.07) .. controls (494.76,123.91) and (501.91,95.35) .. (519.46,90.27) .. controls (537.01,85.19) and (558.32,105.52) .. (567.05,135.68) .. controls (575.78,165.84) and (568.62,194.41) .. (551.07,199.49) -- cycle ;
	\draw  [fill={rgb, 255:red, 0; green, 0; blue, 0 }  ,fill opacity=1 ] (497.56,151.07) .. controls (497.56,148.23) and (499.77,145.92) .. (502.49,145.92) .. controls (505.2,145.92) and (507.41,148.23) .. (507.41,151.07) .. controls (507.41,153.92) and (505.2,156.23) .. (502.49,156.23) .. controls (499.77,156.23) and (497.56,153.92) .. (497.56,151.07) -- cycle ;
	\draw   (419.33,169.07) .. controls (419.33,145.28) and (438.62,126) .. (462.41,126) .. controls (486.2,126) and (505.49,145.28) .. (505.49,169.07) .. controls (505.49,192.86) and (486.2,212.15) .. (462.41,212.15) .. controls (438.62,212.15) and (419.33,192.86) .. (419.33,169.07) -- cycle ;
	\draw  [fill={rgb, 255:red, 0; green, 0; blue, 0 }  ,fill opacity=1 ] (170.59,122.09) .. controls (170.59,119.24) and (172.79,116.93) .. (175.51,116.93) .. controls (178.23,116.93) and (180.43,119.24) .. (180.43,122.09) .. controls (180.43,124.94) and (178.23,127.25) .. (175.51,127.25) .. controls (172.79,127.25) and (170.59,124.94) .. (170.59,122.09) -- cycle ;
	\draw   (141.59,156) .. controls (141.59,137.28) and (156.78,122.09) .. (175.51,122.09) .. controls (194.24,122.09) and (209.42,137.28) .. (209.42,156) .. controls (209.42,174.73) and (194.24,189.92) .. (175.51,189.92) .. controls (156.78,189.92) and (141.59,174.73) .. (141.59,156) -- cycle ;
	\draw  [fill={rgb, 255:red, 65; green, 117; blue, 5 }  ,fill opacity=1 ] (252.98,62.29) .. controls (252.98,59.44) and (255.18,57.13) .. (257.9,57.13) .. controls (260.62,57.13) and (262.82,59.44) .. (262.82,62.29) .. controls (262.82,65.14) and (260.62,67.45) .. (257.9,67.45) .. controls (255.18,67.45) and (252.98,65.14) .. (252.98,62.29) -- cycle ;
	\draw  [fill={rgb, 255:red, 65; green, 117; blue, 5 }  ,fill opacity=1 ] (232.98,160.29) .. controls (232.98,157.44) and (235.18,155.13) .. (237.9,155.13) .. controls (240.62,155.13) and (242.82,157.44) .. (242.82,160.29) .. controls (242.82,163.14) and (240.62,165.45) .. (237.9,165.45) .. controls (235.18,165.45) and (232.98,163.14) .. (232.98,160.29) -- cycle ;
	
	\draw (200,87.4) node [anchor=north west][inner sep=0.75pt]    {$\vv_{\ii}$};
	\draw (325,81.4) node [anchor=north west][inner sep=0.75pt]    {$\vv_{\jj}$};
	\draw (309,125.4) node [anchor=north west][inner sep=0.75pt]    {$\C_{\m}$};
	\draw (164,160.4) node [anchor=north west][inner sep=0.75pt]    {$\C_{\p}$};
	\draw (454,180.4) node [anchor=north west][inner sep=0.75pt]    {$\C_{\q}$};
	\draw (252,78.4) node [anchor=north west][inner sep=0.75pt]    {$\vv_{\aaa}$};
	\draw (241,137.4) node [anchor=north west][inner sep=0.75pt]    {$\vv_{\bb}$};
	\draw (281,28.4) node [anchor=north west][inner sep=0.75pt]    {$\PP_{1}$};
	\draw (288,199.4) node [anchor=north west][inner sep=0.75pt]    {$\PP_{2}$};

\end{tikzpicture}

}
\caption{Odd Inner Cycle in Cacti Graph Without Leaves.}
\label{oddcycleincacti}
\end{center}
\end{figure}
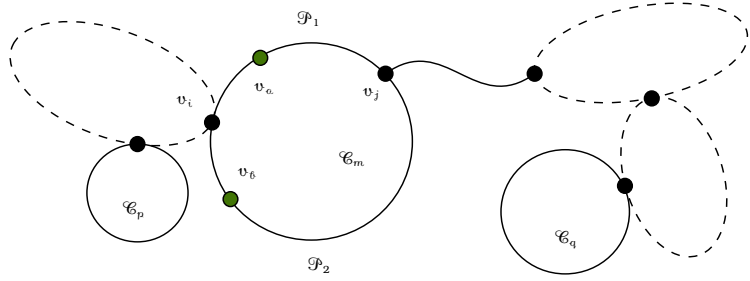
	
Let $\C_{\p}$ and $\C_{\q}$ be the corresponding outer cycles of $\vv_{\ii}$ and $\vv_{\jj}$, respectively. Since $\C_{\p}$ and $\C_{\q}$ are both outer cycles, we have $\left|\SSS_{\C_{\p}}\right|=\left|\SSS_{\C_{\q}}\right|=2$. Let $\vv_{\aaa},\vv_{\bb} \in \C_{\m}$ be such that they are not resolved by vertices of $\C_{\p}$, then they are equidistant from $\vv_{\ii}$. Since $\C_{\m}$ is odd, the two paths $\PP_1, \PP_2$ between $\vv_{\ii},\vv_{\jj}$ have distinct length, one being even and the other being odd. Let us arbitrarily assume that $\PP_1$ is odd and let $\vv_{\aaa}$ lies on this path. From here, we conclude that $\dd(\vv_{\jj},\vv_{\aaa})=\left|\PP_1\right|-\dd(\vv_{\ii},\vv_{\aaa})$.

Now, considering the distance of $\vv_{\jj}$ and $\vv_{\bb}$, we see that there are two possibilities. Either, $\vv_{\ii}$ lies on this shortest path or it does not. If $\vv_{\ii}$ lies on this path, then $\dd(\vv_{\jj},\vv_{\bb})=\left|\PP_1\right|+\dd(\vv_{\ii},\vv_{\bb})$, and since $\dd(\vv_{\ii},\vv_{\aaa})= \dd(\vv_{\ii},\vv_{\bb})$, we observe that $\vv_{\jj}$ resolves $\vv_{\aaa}$ and $\vv_{\bb}$. On the other hand, if the shortest $\vv_{\jj}\vv_{\bb}$ path does not pass through $\vv_{\ii}$, we see that $\dd(\vv_{\jj},\vv_{\bb})=\left|\PP_2\right|-\dd(\vv_{\ii},\vv_{\bb})$. Since $\dd(\vv_{\ii},\vv_{\aaa})=\dd(\vv_{\ii},\vv_{\bb})$, we get, $\dd(\vv_{\jj},\vv_{\bb})=\left|\PP_2\right|-\dd(\vv_{\ii},\vv_{\aaa})$. Once again, we conclude that $\vv_{\jj}$ resolves both $\vv_{\aaa}$ and $\vv_{\bb}$, owing to the fact that $\left|\PP_1\right|\neq \left|\PP_2\right|$.
	
Similarly, it is easy to prove that the vertices of $\C_{\m}$ not resolved by $\C_{\q}$, are resolved by $\C_{\p}$. 

Hence, all vertices of $\C_{\m}$ are resolved by $\C_{\p}$ and $\C_{\q}$. From here, we observe that all vertices of $\C_{\m}$ are FT-resolved by vertices outside $\C_{\m}$ and hence $\left|\SSS_{\C_{\m}}\right|=0$.
\end{proof}

\begin{Lemma} \label{evenacm}
	Let $\G$ be a leafless cactus graph. Let $\C_{\m}$ be an even inner cycle of $\G$ with $\left|\A({\C_{\m}})\right|= 2$, then $\left|\SSS_{\C_{\m}}\right|=0$ if the common vertices are not antipodal.
\end{Lemma}

\begin{proof}
	Let $\G$ and $\C_{\m}$ be as given above with $\left|\A({\C_{\m}})\right|= 2$, where both paths between common vertices of $\C_{\m}$ are of different lengths (Common vertices are not antipodal). Let $\vv_{\ii}$ and $\vv_{\jj}$ be the common vertices, and $\PP_1$ and $\PP_2$ be the two distinct paths between them. This is given in Figure \ref{evencycleincacti}.

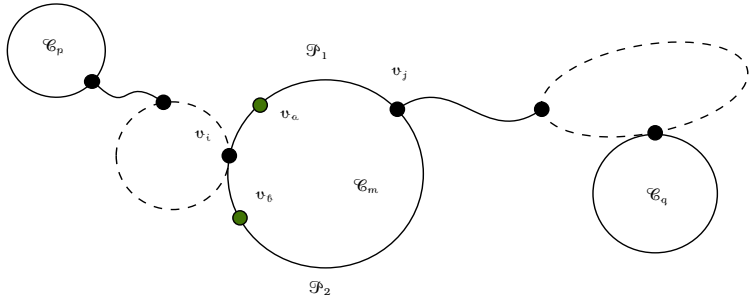
\begin{figure}[H]
	\begin{center}
		\tikzset{every picture/.style={line width=0.75pt}} 
		\resizebox{10cm}{3.9cm}{ 
			
			\begin{tikzpicture}[x=0.75pt,y=0.75pt,yscale=-1,xscale=1]
				
				\draw  [fill={rgb, 255:red, 0; green, 0; blue, 0 }  ,fill opacity=1 ] (356.98,93.29) .. controls (356.98,90.44) and (359.18,88.13) .. (361.9,88.13) .. controls (364.62,88.13) and (366.82,90.44) .. (366.82,93.29) .. controls (366.82,96.14) and (364.62,98.45) .. (361.9,98.45) .. controls (359.18,98.45) and (356.98,96.14) .. (356.98,93.29) -- cycle ;
				\draw   (244.46,140.29) .. controls (244.46,102.51) and (274.78,71.89) .. (312.18,71.89) .. controls (349.58,71.89) and (379.9,102.51) .. (379.9,140.29) .. controls (379.9,178.06) and (349.58,208.69) .. (312.18,208.69) .. controls (274.78,208.69) and (244.46,178.06) .. (244.46,140.29) -- cycle ;
				\draw  [fill={rgb, 255:red, 0; green, 0; blue, 0 }  ,fill opacity=1 ] (240.59,127.09) .. controls (240.59,124.24) and (242.79,121.93) .. (245.51,121.93) .. controls (248.23,121.93) and (250.43,124.24) .. (250.43,127.09) .. controls (250.43,129.94) and (248.23,132.25) .. (245.51,132.25) .. controls (242.79,132.25) and (240.59,129.94) .. (240.59,127.09) -- cycle ;
				\draw    (361.9,93.29) .. controls (401.9,63.29) and (421.9,123.29) .. (461.9,93.29) ;
				\draw  [dash pattern={on 4.5pt off 4.5pt}] (461.9,93.29) .. controls (458.4,75.96) and (487.45,55.48) .. (526.79,47.53) .. controls (566.13,39.59) and (600.86,47.19) .. (604.35,64.51) .. controls (607.85,81.83) and (578.8,102.32) .. (539.46,110.27) .. controls (500.12,118.21) and (465.4,110.61) .. (461.9,93.29) -- cycle ;
				\draw  [fill={rgb, 255:red, 0; green, 0; blue, 0 }  ,fill opacity=1 ] (456.98,93.29) .. controls (456.98,90.44) and (459.18,88.13) .. (461.9,88.13) .. controls (464.62,88.13) and (466.82,90.44) .. (466.82,93.29) .. controls (466.82,96.14) and (464.62,98.45) .. (461.9,98.45) .. controls (459.18,98.45) and (456.98,96.14) .. (456.98,93.29) -- cycle ;
				\draw  [fill={rgb, 255:red, 0; green, 0; blue, 0 }  ,fill opacity=1 ] (535.54,110.42) .. controls (535.54,107.58) and (537.74,105.27) .. (540.46,105.27) .. controls (543.18,105.27) and (545.38,107.58) .. (545.38,110.42) .. controls (545.38,113.27) and (543.18,115.58) .. (540.46,115.58) .. controls (537.74,115.58) and (535.54,113.27) .. (535.54,110.42) -- cycle ;
				\draw   (497.39,154.66) .. controls (497.39,130.87) and (516.67,111.58) .. (540.46,111.58) .. controls (564.25,111.58) and (583.54,130.87) .. (583.54,154.66) .. controls (583.54,178.45) and (564.25,197.73) .. (540.46,197.73) .. controls (516.67,197.73) and (497.39,178.45) .. (497.39,154.66) -- cycle ;
				\draw  [fill={rgb, 255:red, 0; green, 0; blue, 0 }  ,fill opacity=1 ] (194.97,88.23) .. controls (194.97,85.38) and (197.18,83.07) .. (199.89,83.07) .. controls (202.61,83.07) and (204.82,85.38) .. (204.82,88.23) .. controls (204.82,91.08) and (202.61,93.39) .. (199.89,93.39) .. controls (197.18,93.39) and (194.97,91.08) .. (194.97,88.23) -- cycle ;
				\draw   (124.4,16.75) .. controls (143.12,16.04) and (158.87,30.63) .. (159.58,49.35) .. controls (160.29,68.06) and (145.7,83.81) .. (126.98,84.53) .. controls (108.27,85.24) and (92.52,70.64) .. (91.8,51.93) .. controls (91.09,33.21) and (105.69,17.46) .. (124.4,16.75) -- cycle ;
				\draw    (151.59,73.49) .. controls (178.67,101.48) and (169.31,65.81) .. (199.89,88.23) ;
				\draw  [fill={rgb, 255:red, 0; green, 0; blue, 0 }  ,fill opacity=1 ] (150.25,68.18) .. controls (153.1,68.08) and (155.49,70.19) .. (155.6,72.91) .. controls (155.7,75.62) and (153.47,77.91) .. (150.63,78.02) .. controls (147.78,78.13) and (145.39,76.02) .. (145.29,73.3) .. controls (145.18,70.58) and (147.41,68.29) .. (150.25,68.18) -- cycle ;
				\draw  [dash pattern={on 4.5pt off 4.5pt}] (167.17,127.09) .. controls (167.17,105.46) and (184.71,87.93) .. (206.34,87.93) .. controls (227.97,87.93) and (245.51,105.46) .. (245.51,127.09) .. controls (245.51,148.72) and (227.97,166.26) .. (206.34,166.26) .. controls (184.71,166.26) and (167.17,148.72) .. (167.17,127.09) -- cycle ;
				\draw  [fill={rgb, 255:red, 65; green, 117; blue, 5 }  ,fill opacity=1 ] (261.98,90.29) .. controls (261.98,87.44) and (264.18,85.13) .. (266.9,85.13) .. controls (269.62,85.13) and (271.82,87.44) .. (271.82,90.29) .. controls (271.82,93.14) and (269.62,95.45) .. (266.9,95.45) .. controls (264.18,95.45) and (261.98,93.14) .. (261.98,90.29) -- cycle ;
				\draw  [fill={rgb, 255:red, 65; green, 117; blue, 5 }  ,fill opacity=1 ] (247.98,172.29) .. controls (247.98,169.44) and (250.18,167.13) .. (252.9,167.13) .. controls (255.62,167.13) and (257.82,169.44) .. (257.82,172.29) .. controls (257.82,175.14) and (255.62,177.45) .. (252.9,177.45) .. controls (250.18,177.45) and (247.98,175.14) .. (247.98,172.29) -- cycle ;
				
				\draw (220,107.4) node [anchor=north west][inner sep=0.75pt]    {$\vv_{\ii}$};
				\draw (356,61.4) node [anchor=north west][inner sep=0.75pt]    {$\vv_{\jj}$};
				\draw (330,142.4) node [anchor=north west][inner sep=0.75pt]    {$\C_{\m}$};
				\draw (114.64,40.57) node [anchor=north west][inner sep=0.75pt]    {$\C_{\p}$};
				\draw (533,149.4) node [anchor=north west][inner sep=0.75pt]    {$\C_{\q}$};
				\draw (297,44.4) node [anchor=north west][inner sep=0.75pt]    {$\PP_{1}$};
				\draw (299,216.4) node [anchor=north west][inner sep=0.75pt]    {$\PP_{2}$};
				\draw (279,92.4) node [anchor=north west][inner sep=0.75pt]    {$\vv_{\aaa}$};
				\draw (262,151.4) node [anchor=north west][inner sep=0.75pt]    {$\vv_{\bb}$};

			\end{tikzpicture}
			
}
\caption{Even Inner Cycle in Cacti Graph Without Leaves Having Unequal Paths Between Two Common Vertices.}
\label{evencycleincacti}
\end{center}
\end{figure}	

Proceeding in the same way as in Lemma \ref{oddcycleincactiresult}, we can easily conclude that the vertices of $\C_{\m}$ can be resolved using vertices of $\C_{\p}$ and $\C_{\q}$, and hence they can be fault tolerantly resolved by the same cycles. This gives us the result.
		
\end{proof}

\begin{Lemma}\label{eveninner2}
	Let $\G$ be a leafless cactus graph. Let $\C_{\m}$ be an even inner cycle of $\G$ with $\left|\A({\C_{\m}})\right|= 2$, then $\left|\SSS_{\C_{\m}}\right|=2$, if the common vertices are antipodal.
\end{Lemma}

\begin{proof}
	Let $\G$ and $\C_{\m}$ be as given above and let the common vertices of $\C_{\m}$ be antipodal, then $\left|\PP_1\right|= \left|\PP_2\right|$ in figure \ref{evencycleincactiantipodal}. It is also easily observable that the vertices $\vv_{\aaa},\vv_{\bb}$ equidistant from $\vv_{\ii}$, will have the same representation with respect to vertices $\SSS_{\C_{\p}} \cup \SSS_{\C_{\q}}$, or from the union of all other outer cycles' fault tolerant metric generator sets, for that matter.
\begin{figure}[H]
\begin{center}
	\tikzset{every picture/.style={line width=0.75pt}} 
	\resizebox{10cm}{3.9cm}{ 
		
		\begin{tikzpicture}[x=0.75pt,y=0.75pt,yscale=-1,xscale=1]
		
		\draw  [fill={rgb, 255:red, 0; green, 0; blue, 0 }  ,fill opacity=1 ] (374.98,135.29) .. controls (374.98,132.44) and (377.18,130.13) .. (379.9,130.13) .. controls (382.62,130.13) and (384.82,132.44) .. (384.82,135.29) .. controls (384.82,138.14) and (382.62,140.45) .. (379.9,140.45) .. controls (377.18,140.45) and (374.98,138.14) .. (374.98,135.29) -- cycle ;
		\draw   (244.46,140.29) .. controls (244.46,102.51) and (274.78,71.89) .. (312.18,71.89) .. controls (349.58,71.89) and (379.9,102.51) .. (379.9,140.29) .. controls (379.9,178.06) and (349.58,208.69) .. (312.18,208.69) .. controls (274.78,208.69) and (244.46,178.06) .. (244.46,140.29) -- cycle ;
		\draw  [fill={rgb, 255:red, 0; green, 0; blue, 0 }  ,fill opacity=1 ] (240.59,136.09) .. controls (240.59,133.24) and (242.79,130.93) .. (245.51,130.93) .. controls (248.23,130.93) and (250.43,133.24) .. (250.43,136.09) .. controls (250.43,138.94) and (248.23,141.25) .. (245.51,141.25) .. controls (242.79,141.25) and (240.59,138.94) .. (240.59,136.09) -- cycle ;
		\draw    (379.9,135.29) .. controls (419.9,105.29) and (439.9,165.29) .. (479.9,135.29) ;
		\draw  [dash pattern={on 4.5pt off 4.5pt}] (479.9,135.29) .. controls (476.4,117.96) and (505.45,97.48) .. (544.79,89.53) .. controls (584.13,81.59) and (618.86,89.19) .. (622.35,106.51) .. controls (625.85,123.83) and (596.8,144.32) .. (557.46,152.27) .. controls (518.12,160.21) and (483.4,152.61) .. (479.9,135.29) -- cycle ;
		\draw  [fill={rgb, 255:red, 0; green, 0; blue, 0 }  ,fill opacity=1 ] (474.98,135.29) .. controls (474.98,132.44) and (477.18,130.13) .. (479.9,130.13) .. controls (482.62,130.13) and (484.82,132.44) .. (484.82,135.29) .. controls (484.82,138.14) and (482.62,140.45) .. (479.9,140.45) .. controls (477.18,140.45) and (474.98,138.14) .. (474.98,135.29) -- cycle ;
		\draw  [fill={rgb, 255:red, 0; green, 0; blue, 0 }  ,fill opacity=1 ] (553.54,152.42) .. controls (553.54,149.58) and (555.74,147.27) .. (558.46,147.27) .. controls (561.18,147.27) and (563.38,149.58) .. (563.38,152.42) .. controls (563.38,155.27) and (561.18,157.58) .. (558.46,157.58) .. controls (555.74,157.58) and (553.54,155.27) .. (553.54,152.42) -- cycle ;
		\draw   (515.39,196.66) .. controls (515.39,172.87) and (534.67,153.58) .. (558.46,153.58) .. controls (582.25,153.58) and (601.54,172.87) .. (601.54,196.66) .. controls (601.54,220.45) and (582.25,239.73) .. (558.46,239.73) .. controls (534.67,239.73) and (515.39,220.45) .. (515.39,196.66) -- cycle ;
		\draw  [fill={rgb, 255:red, 0; green, 0; blue, 0 }  ,fill opacity=1 ] (192.97,96.23) .. controls (192.97,93.38) and (195.18,91.07) .. (197.89,91.07) .. controls (200.61,91.07) and (202.82,93.38) .. (202.82,96.23) .. controls (202.82,99.08) and (200.61,101.39) .. (197.89,101.39) .. controls (195.18,101.39) and (192.97,99.08) .. (192.97,96.23) -- cycle ;
		\draw   (122.4,24.75) .. controls (141.12,24.04) and (156.87,38.63) .. (157.58,57.35) .. controls (158.29,76.06) and (143.7,91.81) .. (124.98,92.53) .. controls (106.27,93.24) and (90.52,78.64) .. (89.8,59.93) .. controls (89.09,41.21) and (103.69,25.46) .. (122.4,24.75) -- cycle ;
		\draw    (149.59,81.49) .. controls (176.67,109.48) and (167.31,73.81) .. (197.89,96.23) ;
		\draw  [fill={rgb, 255:red, 0; green, 0; blue, 0 }  ,fill opacity=1 ] (148.25,76.18) .. controls (151.1,76.08) and (153.49,78.19) .. (153.6,80.91) .. controls (153.7,83.62) and (151.47,85.91) .. (148.63,86.02) .. controls (145.78,86.13) and (143.39,84.02) .. (143.29,81.3) .. controls (143.18,78.58) and (145.41,76.29) .. (148.25,76.18) -- cycle ;
		\draw  [dash pattern={on 4.5pt off 4.5pt}] (165.17,135.09) .. controls (165.17,113.46) and (182.71,95.93) .. (204.34,95.93) .. controls (225.97,95.93) and (243.51,113.46) .. (243.51,135.09) .. controls (243.51,156.72) and (225.97,174.26) .. (204.34,174.26) .. controls (182.71,174.26) and (165.17,156.72) .. (165.17,135.09) -- cycle ;
		\draw  [fill={rgb, 255:red, 65; green, 117; blue, 5 }  ,fill opacity=1 ] (261.98,90.29) .. controls (261.98,87.44) and (264.18,85.13) .. (266.9,85.13) .. controls (269.62,85.13) and (271.82,87.44) .. (271.82,90.29) .. controls (271.82,93.14) and (269.62,95.45) .. (266.9,95.45) .. controls (264.18,95.45) and (261.98,93.14) .. (261.98,90.29) -- cycle ;
		\draw  [fill={rgb, 255:red, 65; green, 117; blue, 5 }  ,fill opacity=1 ] (255.98,184.29) .. controls (255.98,181.44) and (258.18,179.13) .. (260.9,179.13) .. controls (263.62,179.13) and (265.82,181.44) .. (265.82,184.29) .. controls (265.82,187.14) and (263.62,189.45) .. (260.9,189.45) .. controls (258.18,189.45) and (255.98,187.14) .. (255.98,184.29) -- cycle ;
		
		\draw (218,115.4) node [anchor=north west][inner sep=0.75pt]    {$\vv_{\ii}$};
		\draw (383,102.4) node [anchor=north west][inner sep=0.75pt]    {$\vv_{\jj}$};
		\draw (330,142.4) node [anchor=north west][inner sep=0.75pt]    {$\C_{\m}$};
		\draw (112.64,48.57) node [anchor=north west][inner sep=0.75pt]    {$\C_{\p}$};
		\draw (551,191.4) node [anchor=north west][inner sep=0.75pt]    {$\C_{\q}$};
		\draw (297,44.4) node [anchor=north west][inner sep=0.75pt]    {$\PP_{1}$};
		\draw (299,216.4) node [anchor=north west][inner sep=0.75pt]    {$\PP_{2}$};
		\draw (269,100.4) node [anchor=north west][inner sep=0.75pt]    {$\vv_{\aaa}$};
		\draw (269,169.4) node [anchor=north west][inner sep=0.75pt]    {$\vv_{\bb}$};

	\end{tikzpicture}
}
\caption{Even Inner Cycle in Cacti Graph Without Leaves Having Equal Paths Between Two Common Vertices.}
\label{evencycleincactiantipodal}
\end{center}
\end{figure}
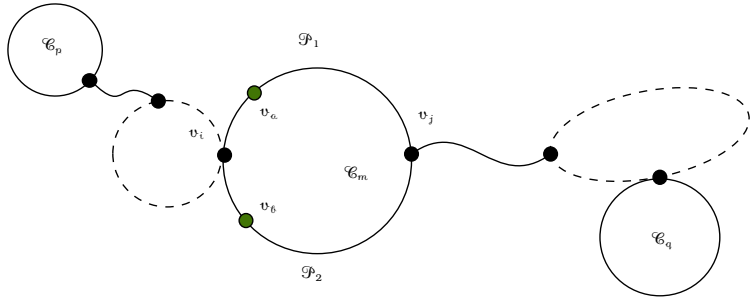

Let us now consider the vertices $\vv'_{\ii}$ and $\vv'_{\jj}$ such that both belong to the same $\vv_{\ii} \vv_{\jj}$ path, and $\vv'_{\ii}\sim \vv_{\ii}$ and $\vv'_{\jj} \sim \vv_{\jj}$, as shown in the following figure. Notice that the above construction is only possible if $\m>4$. The case where $\m=4$ is handled a bit differently.

\begin{figure}[H]
	\begin{center}
		\tikzset{every picture/.style={line width=0.75pt}} 
		\resizebox{10cm}{4.9cm}{
\begin{tikzpicture}[x=0.75pt,y=0.75pt,yscale=-1,xscale=1]
	
	\draw  [fill={rgb, 255:red, 0; green, 0; blue, 0 }  ,fill opacity=1 ] (356.98,155.29) .. controls (356.98,152.44) and (359.18,150.13) .. (361.9,150.13) .. controls (364.62,150.13) and (366.82,152.44) .. (366.82,155.29) .. controls (366.82,158.14) and (364.62,160.45) .. (361.9,160.45) .. controls (359.18,160.45) and (356.98,158.14) .. (356.98,155.29) -- cycle ;
	\draw  [fill={rgb, 255:red, 0; green, 0; blue, 0 }  ,fill opacity=1 ] (222.59,156.09) .. controls (222.59,153.24) and (224.79,150.93) .. (227.51,150.93) .. controls (230.23,150.93) and (232.43,153.24) .. (232.43,156.09) .. controls (232.43,158.94) and (230.23,161.25) .. (227.51,161.25) .. controls (224.79,161.25) and (222.59,158.94) .. (222.59,156.09) -- cycle ;
	\draw  [dash pattern={on 4.5pt off 4.5pt}]  (361.9,155.29) .. controls (441.33,135) and (359.33,107) .. (433.33,80) ;
	\draw  [fill={rgb, 255:red, 0; green, 0; blue, 0 }  ,fill opacity=1 ] (428.98,79.29) .. controls (428.98,76.44) and (431.18,74.13) .. (433.9,74.13) .. controls (436.62,74.13) and (438.82,76.44) .. (438.82,79.29) .. controls (438.82,82.14) and (436.62,84.45) .. (433.9,84.45) .. controls (431.18,84.45) and (428.98,82.14) .. (428.98,79.29) -- cycle ;
	\draw  [fill={rgb, 255:red, 0; green, 0; blue, 0 }  ,fill opacity=1 ] (461.54,122.42) .. controls (461.54,119.58) and (463.74,117.27) .. (466.46,117.27) .. controls (469.18,117.27) and (471.38,119.58) .. (471.38,122.42) .. controls (471.38,125.27) and (469.18,127.58) .. (466.46,127.58) .. controls (463.74,127.58) and (461.54,125.27) .. (461.54,122.42) -- cycle ;
	\draw   (414.39,165.66) .. controls (414.39,141.87) and (433.67,122.58) .. (457.46,122.58) .. controls (481.25,122.58) and (500.54,141.87) .. (500.54,165.66) .. controls (500.54,189.45) and (481.25,208.73) .. (457.46,208.73) .. controls (433.67,208.73) and (414.39,189.45) .. (414.39,165.66) -- cycle ;
	\draw  [dash pattern={on 4.5pt off 4.5pt}]  (127.51,156.09) .. controls (167.51,126.09) and (187.51,186.09) .. (227.51,156.09) ;
	\draw  [fill={rgb, 255:red, 0; green, 0; blue, 0 }  ,fill opacity=1 ] (122.98,155.29) .. controls (122.98,152.44) and (125.18,150.13) .. (127.9,150.13) .. controls (130.62,150.13) and (132.82,152.44) .. (132.82,155.29) .. controls (132.82,158.14) and (130.62,160.45) .. (127.9,160.45) .. controls (125.18,160.45) and (122.98,158.14) .. (122.98,155.29) -- cycle ;
	\draw  [dash pattern={on 4.5pt off 4.5pt}] (98.55,101.53) .. controls (114.84,101.58) and (127.98,125.71) .. (127.9,155.45) .. controls (127.82,185.18) and (114.54,209.24) .. (98.25,209.2) .. controls (81.96,209.15) and (68.82,185.01) .. (68.9,155.28) .. controls (68.98,125.55) and (82.26,101.48) .. (98.55,101.53) -- cycle ;
	\draw  [fill={rgb, 255:red, 0; green, 0; blue, 0 }  ,fill opacity=1 ] (97.98,101.29) .. controls (97.98,98.44) and (100.18,96.13) .. (102.9,96.13) .. controls (105.62,96.13) and (107.82,98.44) .. (107.82,101.29) .. controls (107.82,104.14) and (105.62,106.45) .. (102.9,106.45) .. controls (100.18,106.45) and (97.98,104.14) .. (97.98,101.29) -- cycle ;
	\draw   (51.68,70.93) .. controls (51.68,54.17) and (74.61,40.57) .. (102.9,40.57) .. controls (131.19,40.57) and (154.12,54.17) .. (154.12,70.93) .. controls (154.12,87.7) and (131.19,101.29) .. (102.9,101.29) .. controls (74.61,101.29) and (51.68,87.7) .. (51.68,70.93) -- cycle ;
	\draw   (294.7,56) .. controls (331.81,56) and (361.9,100.55) .. (361.9,155.5) .. controls (361.9,210.45) and (331.81,255) .. (294.7,255) .. controls (257.59,255) and (227.51,210.45) .. (227.51,155.5) .. controls (227.51,100.55) and (257.59,56) .. (294.7,56) -- cycle ;
	\draw  [fill={rgb, 255:red, 74; green, 144; blue, 226 }  ,fill opacity=1 ] (225.98,185.29) .. controls (225.98,182.44) and (228.18,180.13) .. (230.9,180.13) .. controls (233.62,180.13) and (235.82,182.44) .. (235.82,185.29) .. controls (235.82,188.14) and (233.62,190.45) .. (230.9,190.45) .. controls (228.18,190.45) and (225.98,188.14) .. (225.98,185.29) -- cycle ;
	\draw  [fill={rgb, 255:red, 74; green, 144; blue, 226 }  ,fill opacity=1 ] (353.98,185.29) .. controls (353.98,182.44) and (356.18,180.13) .. (358.9,180.13) .. controls (361.62,180.13) and (363.82,182.44) .. (363.82,185.29) .. controls (363.82,188.14) and (361.62,190.45) .. (358.9,190.45) .. controls (356.18,190.45) and (353.98,188.14) .. (353.98,185.29) -- cycle ;
	\draw  [dash pattern={on 4.5pt off 4.5pt}] (434.33,79) .. controls (434.33,53.75) and (454.8,33.28) .. (480.05,33.28) .. controls (505.3,33.28) and (525.77,53.75) .. (525.77,79) .. controls (525.77,104.25) and (505.3,124.72) .. (480.05,124.72) .. controls (454.8,124.72) and (434.33,104.25) .. (434.33,79) -- cycle ;
	\draw  [fill={rgb, 255:red, 0; green, 0; blue, 0 }  ,fill opacity=1 ] (520.54,76.42) .. controls (520.54,73.58) and (522.74,71.27) .. (525.46,71.27) .. controls (528.18,71.27) and (530.38,73.58) .. (530.38,76.42) .. controls (530.38,79.27) and (528.18,81.58) .. (525.46,81.58) .. controls (522.74,81.58) and (520.54,79.27) .. (520.54,76.42) -- cycle ;
	\draw   (551.54,21.26) .. controls (565.9,21.26) and (577.54,45.96) .. (577.54,76.42) .. controls (577.54,106.89) and (565.9,131.59) .. (551.54,131.59) .. controls (537.18,131.59) and (525.54,106.89) .. (525.54,76.42) .. controls (525.54,45.96) and (537.18,21.26) .. (551.54,21.26) -- cycle ;
	
	\draw (208,122.4) node [anchor=north west][inner sep=0.75pt]    {$\vv_{\ii}$};
	\draw (365,122.4) node [anchor=north west][inner sep=0.75pt]    {$\vv_{\jj}$};
	\draw (280,131.4) node [anchor=north west][inner sep=0.75pt]    {$\C_{\m}$};
	\draw (453,161.4) node [anchor=north west][inner sep=0.75pt]    {$\C_{\q}$};
	\draw (284,26.4) node [anchor=north west][inner sep=0.75pt]    {$\PP_{1}$};
	\draw (285,269.4) node [anchor=north west][inner sep=0.75pt]    {$\PP_{2}$};
	\draw (247,175.4) node [anchor=north west][inner sep=0.75pt]    {$\vv'_{i}$};
	\draw (328,175.4) node [anchor=north west][inner sep=0.75pt]    {$\vv'_{j}$};
	\draw (93,62.4) node [anchor=north west][inner sep=0.75pt]    {$\C_{\p}$};
	\draw (545,62.4) node [anchor=north west][inner sep=0.75pt]    {$\C_{r}$};

\end{tikzpicture}

}
\caption{Even Inner Cycle in Leafless Cactus Graph with $\left|\A({\C_{\m}})\right|= 2$ and Having Antipodal Vertices.}
\label{evencycleincactiantipodalbasis}
\end{center}
\end{figure}
Now, any two vertices which are not resolved by $\vv_{\ii}$, are resolved by $\vv_{\ii}'$ and vice versa. Similarly, the vertices which are not resolved by $\vv_{\jj}$, are resolved by $\vv_{\jj}'$ and vice versa. Moreover, $\{\vv_{\ii}, \vv'_{\ii},\vv_{\jj}, \vv'_{\jj}\}$ form a fault tolerant resolving set of $\C_{\m}$.

If $\m=4$, obviously $\C_{\m}$ contains only $4$ vertices. The common vertices are indexed $\vv_{\ii},\vv_{\jj}$ and the remaining two vertices are indexed $\vv_{\ii}',\vv_{\jj}'$. Again, $\{\vv_{\ii}, \vv'_{\ii},\vv_{\jj}, \vv'_{\jj}\}$ form a fault tolerant resolving set $\C_{\m}$.

Since all vertices of $\C_{\m}$ resolved by $\vv_{\ii}$ and $\vv_{\jj}$ are also resolved by their corresponding outer cycles $\C_{\p}$ and $\C_{\q}$, we can conclude that $\SSS_{\C_{\p}} \cup \SSS_{\C_{\q}} \cup \{\vv'_{\ii}, \vv'_{\jj}\}$ is a fault tolerant metric producer for the cycle $\C_{\m}$. Hence $\SSS_{\C_{\m}}=\{\vv'_{\ii}, \vv'_{\jj}\}$ and $\left|\SSS_{\C_{\m}}\right|=2$.

\end{proof}

\begin{Lemma} \label{evenmorethan2}
	Let $\G$ be a cactus graph without leaves. Let $\C_{\m}$ be an even inner cycle of $\G$ with $\left|\A({\C_{\m}})\right|\geq 2$, then $\left|\SSS_{\C_{\m}}\right|=0$.
\end{Lemma}

\begin{proof}
	Let $\G$ and $\C_{\m}$ be as given above. Since $\left|\A({\C_{\m}})\right|\geq 2$, there are at least two common vertices which are not antipodal. Let these vertices be $\vv_{\ii},\vv_{\jj}$. Arguing along the same lines as in Lemma \ref{evenacm}, we see that the vertices of $\C_{\m}$ are fault tolerantly resolved by the union of fault tolerant metric generators of corresponding outer cycles of $\vv_{\ii}$ and $\vv_{\jj}$. This gives us the result.
\end{proof}

\begin{Lemma}\label{pathspm}
	Let $\G$ be a cactus graph without leaves, and let $\PP_{\n}$ be a path joining two cycles in $\G$, then $\left|\SSS_{\PP_{\m}}\right|=0$.
\end{Lemma}

\begin{proof}
	Let $\G$ be a leafless cactus graph and $\PP_{\n}$ be a path joining two cycles in $\G$. The path $\PP_{\n}$ can be of the two types shown in the figure \ref{commonpathfigure}.
		\begin{figure}[h]
		\begin{center}
			\tikzset{every picture/.style={line width=0.75pt}} 
			\resizebox{13cm}{3.6cm}{%
				\begin{minipage}[t]{1.0\textwidth}
					\centering				
					\begin{tikzpicture}[x=0.75pt,y=0.75pt,yscale=-1,xscale=1]
						
						\draw  [fill={rgb, 255:red, 0; green, 0; blue, 0 }  ,fill opacity=1 ] (258.08,135.5) .. controls (258.08,132.65) and (260.28,130.34) .. (263,130.34) .. controls (265.72,130.34) and (267.92,132.65) .. (267.92,135.5) .. controls (267.92,138.35) and (265.72,140.66) .. (263,140.66) .. controls (260.28,140.66) and (258.08,138.35) .. (258.08,135.5) -- cycle ;
						\draw  [fill={rgb, 255:red, 0; green, 0; blue, 0 }  ,fill opacity=1 ] (374.98,134.61) .. controls (374.98,131.76) and (377.18,129.45) .. (379.9,129.45) .. controls (382.62,129.45) and (384.82,131.76) .. (384.82,134.61) .. controls (384.82,137.46) and (382.62,139.77) .. (379.9,139.77) .. controls (377.18,139.77) and (374.98,137.46) .. (374.98,134.61) -- cycle ;
						\draw  [dash pattern={on 4.5pt off 4.5pt}] (156,135.5) .. controls (156,105.95) and (179.95,82) .. (209.5,82) .. controls (239.05,82) and (263,105.95) .. (263,135.5) .. controls (263,165.05) and (239.05,189) .. (209.5,189) .. controls (179.95,189) and (156,165.05) .. (156,135.5) -- cycle ;
						\draw  [dash pattern={on 4.5pt off 4.5pt}] (428.87,201.6) .. controls (401.82,201.6) and (379.89,172.07) .. (379.89,135.65) .. controls (379.89,99.23) and (401.82,69.7) .. (428.87,69.7) .. controls (455.93,69.7) and (477.86,99.23) .. (477.86,135.65) .. controls (477.86,172.07) and (455.93,201.6) .. (428.87,201.6) -- cycle ;
						\draw    (263,135.5) -- (379,135.5) ;
						\draw   (136,100.5) .. controls (147.05,100.5) and (156,116.17) .. (156,135.5) .. controls (156,154.83) and (147.05,170.5) .. (136,170.5) .. controls (124.95,170.5) and (116,154.83) .. (116,135.5) .. controls (116,116.17) and (124.95,100.5) .. (136,100.5) -- cycle ;
						\draw  [fill={rgb, 255:red, 0; green, 0; blue, 0 }  ,fill opacity=1 ] (151.08,135.5) .. controls (151.08,132.65) and (153.28,130.34) .. (156,130.34) .. controls (158.72,130.34) and (160.92,132.65) .. (160.92,135.5) .. controls (160.92,138.35) and (158.72,140.66) .. (156,140.66) .. controls (153.28,140.66) and (151.08,138.35) .. (151.08,135.5) -- cycle ;
						\draw   (477.86,135.65) .. controls (477.86,121.84) and (489.05,110.65) .. (502.86,110.65) .. controls (516.67,110.65) and (527.86,121.84) .. (527.86,135.65) .. controls (527.86,149.46) and (516.67,160.65) .. (502.86,160.65) .. controls (489.05,160.65) and (477.86,149.46) .. (477.86,135.65) -- cycle ;
						\draw  [fill={rgb, 255:red, 0; green, 0; blue, 0 }  ,fill opacity=1 ] (472.94,135.65) .. controls (472.94,132.8) and (475.14,130.49) .. (477.86,130.49) .. controls (480.58,130.49) and (482.78,132.8) .. (482.78,135.65) .. controls (482.78,138.5) and (480.58,140.81) .. (477.86,140.81) .. controls (475.14,140.81) and (472.94,138.5) .. (472.94,135.65) -- cycle ;
						
						\draw (124,123.4) node [anchor=north west][inner sep=0.75pt]    {$\C_{\p}$};
						\draw (494,123.4) node [anchor=north west][inner sep=0.75pt]    {$\C_{\q}$};
						\draw (311,109.4) node [anchor=north west][inner sep=0.75pt]    {$\PP_{\n}$};
						
					\end{tikzpicture}

				\end{minipage}
				\begin{minipage}[t]{1.1\textwidth}
					\centering
					\begin{tikzpicture}[x=0.75pt,y=0.75pt,yscale=-1,xscale=1]
						
						\draw  [fill={rgb, 255:red, 0; green, 0; blue, 0 }  ,fill opacity=1 ] (270.08,195.5) .. controls (270.08,192.65) and (272.28,190.34) .. (275,190.34) .. controls (277.72,190.34) and (279.92,192.65) .. (279.92,195.5) .. controls (279.92,198.35) and (277.72,200.66) .. (275,200.66) .. controls (272.28,200.66) and (270.08,198.35) .. (270.08,195.5) -- cycle ;
						\draw  [fill={rgb, 255:red, 0; green, 0; blue, 0 }  ,fill opacity=1 ] (386.98,194.61) .. controls (386.98,191.76) and (389.18,189.45) .. (391.9,189.45) .. controls (394.62,189.45) and (396.82,191.76) .. (396.82,194.61) .. controls (396.82,197.46) and (394.62,199.77) .. (391.9,199.77) .. controls (389.18,199.77) and (386.98,197.46) .. (386.98,194.61) -- cycle ;
						\draw  [dash pattern={on 4.5pt off 4.5pt}] (168,195.5) .. controls (168,165.95) and (191.95,142) .. (221.5,142) .. controls (251.05,142) and (275,165.95) .. (275,195.5) .. controls (275,225.05) and (251.05,249) .. (221.5,249) .. controls (191.95,249) and (168,225.05) .. (168,195.5) -- cycle ;
						\draw  [dash pattern={on 4.5pt off 4.5pt}] (405.75,275.58) .. controls (381.62,263.34) and (375.42,227.09) .. (391.9,194.61) .. controls (408.38,162.13) and (441.29,145.72) .. (465.42,157.96) .. controls (489.55,170.2) and (495.75,206.45) .. (479.27,238.94) .. controls (462.79,271.42) and (429.88,287.82) .. (405.75,275.58) -- cycle ;
						\draw    (275,195.5) -- (333,195.5) ;
						\draw   (148,160.5) .. controls (159.05,160.5) and (168,176.17) .. (168,195.5) .. controls (168,214.83) and (159.05,230.5) .. (148,230.5) .. controls (136.95,230.5) and (128,214.83) .. (128,195.5) .. controls (128,176.17) and (136.95,160.5) .. (148,160.5) -- cycle ;
						\draw  [fill={rgb, 255:red, 0; green, 0; blue, 0 }  ,fill opacity=1 ] (328.08,195.5) .. controls (328.08,192.65) and (330.28,190.34) .. (333,190.34) .. controls (335.72,190.34) and (337.92,192.65) .. (337.92,195.5) .. controls (337.92,198.35) and (335.72,200.66) .. (333,200.66) .. controls (330.28,200.66) and (328.08,198.35) .. (328.08,195.5) -- cycle ;
						\draw   (307.33,29) .. controls (307.33,15.19) and (318.53,4) .. (332.33,4) .. controls (346.14,4) and (357.33,15.19) .. (357.33,29) .. controls (357.33,42.81) and (346.14,54) .. (332.33,54) .. controls (318.53,54) and (307.33,42.81) .. (307.33,29) -- cycle ;
						\draw  [fill={rgb, 255:red, 0; green, 0; blue, 0 }  ,fill opacity=1 ] (327.41,54) .. controls (327.41,51.15) and (329.61,48.84) .. (332.33,48.84) .. controls (335.05,48.84) and (337.26,51.15) .. (337.26,54) .. controls (337.26,56.85) and (335.05,59.16) .. (332.33,59.16) .. controls (329.61,59.16) and (327.41,56.85) .. (327.41,54) -- cycle ;
						\draw  [dash pattern={on 0.84pt off 2.51pt}]  (333,195.5) -- (391.9,194.61) ;
						\draw    (332.33,129) -- (333,195.5) ;
						\draw  [dash pattern={on 4.5pt off 4.5pt}] (252.67,91.5) .. controls (252.67,70.79) and (288.33,54) .. (332.33,54) .. controls (376.33,54) and (412,70.79) .. (412,91.5) .. controls (412,112.21) and (376.33,129) .. (332.33,129) .. controls (288.33,129) and (252.67,112.21) .. (252.67,91.5) -- cycle ;
						\draw  [fill={rgb, 255:red, 0; green, 0; blue, 0 }  ,fill opacity=1 ] (327.41,129) .. controls (327.41,126.15) and (329.61,123.84) .. (332.33,123.84) .. controls (335.05,123.84) and (337.26,126.15) .. (337.26,129) .. controls (337.26,131.85) and (335.05,134.16) .. (332.33,134.16) .. controls (329.61,134.16) and (327.41,131.85) .. (327.41,129) -- cycle ;
						
						\draw (136,183.4) node [anchor=north west][inner sep=0.75pt]    {$\C_{\p}$};
						\draw (324,19.4) node [anchor=north west][inner sep=0.75pt]    {$\C_{\q}$};
						\draw (301,152.4) node [anchor=north west][inner sep=0.75pt]    {$\PP_{\n}$};
						
					\end{tikzpicture}
					
				\end{minipage}
			}
			\caption{Types of Inner Paths in a Leafless Cactus Graph.}
			\label{commonpathfigure}
		\end{center}
	\end{figure}
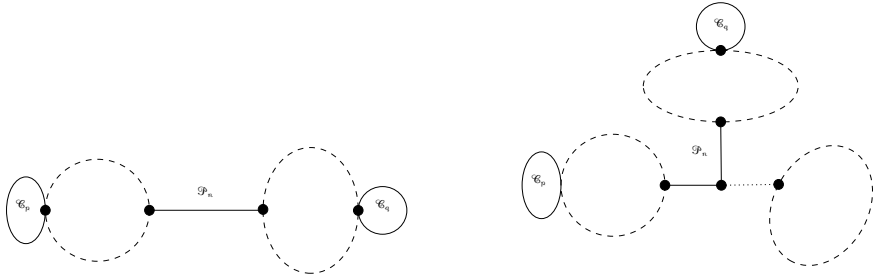
	We see from the figure that in both types, there are always at least two outer cycles, say $\C_{\p}$ and $\C_{\q}$, corresponding to the end points of the path $\PP_{\n}$. By Lemma \ref{outercycle2}, both these cycles contribute two vertices to the fault tolerant metric dimension of $\G$. We can now easily conclude that all the vertices of $\PP_{\n}$ are fault tolerantly resolved by the union of fault tolerant metric generators of $\C_{\p}$ and $\C_{\q}$, and hence $\left|\SSS_{\PP_{\m}}\right|=0$.
\end{proof}

Combining all these results together, we can now easily formulate the fault tolerant metric dimension of cacti graphs without leaves.

\begin{Theorem}\label{finalresult}
	Let $\G$ be a leafless cactus graph with $\n_1$ outer cycles. Let $\n_2$ number of even inner cycles have only two antipodal common vertices, then $\beta'(\G)=2(\n_1+\n_2)$.
\end{Theorem}
\begin{proof}
	Let $\G$ be a cactus graph without leaves with $\n_1$ outer cycles, using Lemma \ref{outercycle2}, every outer cycle contributes two vertices to the fault tolerant metric basis and hence a total of $2\n_1$ vertices are contributed. Similarly, by Lemma \ref{eveninner2}, all even cycles with exactly two antipodal common vertices, contribute two vertices each to the fault tolerant metric basis. Now, by Lemmas \ref{oddcycleincactiresult},\ref{evenacm},\ref{evenmorethan2} and \ref{pathspm}, inner odd cycles, remaining inner even cycles and the paths joining two vertices don't contribute any vertex to the fault tolerant metric basis, since they are resolved by their corresponding outer cycles. Hence $\beta'(\G)=2(\n_1+\n_2)$.
\end{proof}

\section{Cacti Graphs in Supply Chain Logistics}

Let us consider a real world example of a supply chain logistics problem. This supply chain consists of multiple main warehouses, many more smaller warehouses obtaining supplies from the main warehouses or other smaller counterparts if necessary, and the retail outlets which get their supplies from the main or smaller warehouses. This example can be easily seen in real life in the Amazon Delivery Service, FedEx and DHL etc.

The delivery process is optimized in such a way that there is minimal backtracking to save precious resources and time. This gives rise to circular (cyclic) routes and a route is fully optimized if it does not visit any place more than once.

To effectively elaborate this example, let us consider that there are $3$ main warehouses, denoted by $\{\m_1, \m_2, \m_3\}$, $4$ smaller warehouses, denoted by $\{\s_1,\s_2,\s_3,\s_4\}$ and $14$ retail outlets, denoted by $\{\rr_1,\rr_2, \cdots, \rr_{14}\}$. We suppose that the layout of all these entities is as given in the figure \ref{cactiassupplychain}.

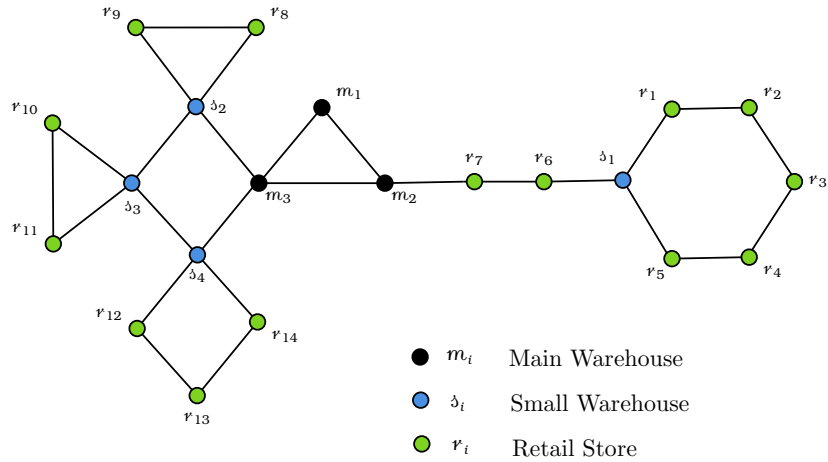
\begin{figure}[H]
	\begin{center}
		\tikzset{every picture/.style={line width=0.75pt}} 
		\resizebox{11cm}{6cm}{
			\begin{tikzpicture}[x=0.75pt,y=0.75pt,yscale=-1,xscale=1]

	\draw  [fill={rgb, 255:red, 0; green, 0; blue, 0 }  ,fill opacity=1 ] (238.83,133.63) .. controls (238.83,131.23) and (240.77,129.28) .. (243.17,129.28) .. controls (245.57,129.28) and (247.51,131.23) .. (247.51,133.63) .. controls (247.51,136.02) and (245.57,137.97) .. (243.17,137.97) .. controls (240.77,137.97) and (238.83,136.02) .. (238.83,133.63) -- cycle ;
	\draw  [fill={rgb, 255:red, 0; green, 0; blue, 0 }  ,fill opacity=1 ] (203.53,91.53) .. controls (203.53,89.13) and (205.47,87.19) .. (207.87,87.19) .. controls (210.27,87.19) and (212.21,89.13) .. (212.21,91.53) .. controls (212.21,93.93) and (210.27,95.87) .. (207.87,95.87) .. controls (205.47,95.87) and (203.53,93.93) .. (203.53,91.53) -- cycle ;
	\draw  [fill={rgb, 255:red, 0; green, 0; blue, 0 }  ,fill opacity=1 ] (168.23,133.63) .. controls (168.23,131.23) and (170.17,129.28) .. (172.57,129.28) .. controls (174.97,129.28) and (176.92,131.23) .. (176.92,133.63) .. controls (176.92,136.02) and (174.97,137.97) .. (172.57,137.97) .. controls (170.17,137.97) and (168.23,136.02) .. (168.23,133.63) -- cycle ;
	\draw    (207.87,91.53) -- (243.17,133.63) ;
	\draw    (446.13,91.53) -- (473.49,132.78) ;
	\draw    (376.42,131.94) -- (402.89,175.72) ;
	\draw    (473.49,132.78) -- (446.13,174.88) ;
	\draw    (446.13,174.88) -- (402.89,175.72) ;
	\draw    (402.89,92.37) -- (376.42,131.94) ;
	\draw    (446.13,91.53) -- (402.89,92.37) ;
	\draw  [fill={rgb, 255:red, 126; green, 211; blue, 33 }  ,fill opacity=1 ] (399.55,92.37) .. controls (399.55,89.97) and (401.49,88.03) .. (403.89,88.03) .. controls (406.29,88.03) and (408.24,89.97) .. (408.24,92.37) .. controls (408.24,94.77) and (406.29,96.72) .. (403.89,96.72) .. controls (401.49,96.72) and (399.55,94.77) .. (399.55,92.37) -- cycle ;
	\draw  [fill={rgb, 255:red, 126; green, 211; blue, 33 }  ,fill opacity=1 ] (442.79,91.53) .. controls (442.79,89.13) and (444.73,87.19) .. (447.13,87.19) .. controls (449.53,87.19) and (451.48,89.13) .. (451.48,91.53) .. controls (451.48,93.93) and (449.53,95.87) .. (447.13,95.87) .. controls (444.73,95.87) and (442.79,93.93) .. (442.79,91.53) -- cycle ;
	\draw  [fill={rgb, 255:red, 126; green, 211; blue, 33 }  ,fill opacity=1 ] (468.15,132.78) .. controls (468.15,130.39) and (470.09,128.44) .. (472.49,128.44) .. controls (474.89,128.44) and (476.83,130.39) .. (476.83,132.78) .. controls (476.83,135.18) and (474.89,137.13) .. (472.49,137.13) .. controls (470.09,137.13) and (468.15,135.18) .. (468.15,132.78) -- cycle ;
	\draw  [fill={rgb, 255:red, 126; green, 211; blue, 33 }  ,fill opacity=1 ] (399.55,175.72) .. controls (399.55,173.32) and (401.49,171.38) .. (403.89,171.38) .. controls (406.29,171.38) and (408.24,173.32) .. (408.24,175.72) .. controls (408.24,178.12) and (406.29,180.06) .. (403.89,180.06) .. controls (401.49,180.06) and (399.55,178.12) .. (399.55,175.72) -- cycle ;
	\draw  [fill={rgb, 255:red, 126; green, 211; blue, 33 }  ,fill opacity=1 ] (442.79,174.88) .. controls (442.79,172.48) and (444.73,170.54) .. (447.13,170.54) .. controls (449.53,170.54) and (451.48,172.48) .. (451.48,174.88) .. controls (451.48,177.28) and (449.53,179.22) .. (447.13,179.22) .. controls (444.73,179.22) and (442.79,177.28) .. (442.79,174.88) -- cycle ;
	\draw    (287.29,132.78) -- (333.18,132.78) ;
	\draw    (333.18,132.78) -- (376.42,131.94) ;
	\draw  [fill={rgb, 255:red, 126; green, 211; blue, 33 }  ,fill opacity=1 ] (327.84,132.78) .. controls (327.84,130.39) and (329.78,128.44) .. (332.18,128.44) .. controls (334.58,128.44) and (336.52,130.39) .. (336.52,132.78) .. controls (336.52,135.18) and (334.58,137.13) .. (332.18,137.13) .. controls (329.78,137.13) and (327.84,135.18) .. (327.84,132.78) -- cycle ;
	\draw  [fill={rgb, 255:red, 74; green, 144; blue, 226 }  ,fill opacity=1 ] (372.08,131.94) .. controls (372.08,129.54) and (374.02,127.6) .. (376.42,127.6) .. controls (378.82,127.6) and (380.76,129.54) .. (380.76,131.94) .. controls (380.76,134.34) and (378.82,136.29) .. (376.42,136.29) .. controls (374.02,136.29) and (372.08,134.34) .. (372.08,131.94) -- cycle ;
	\draw    (243.17,133.63) -- (287.29,132.78) ;
	\draw  [fill={rgb, 255:red, 126; green, 211; blue, 33 }  ,fill opacity=1 ] (288.95,132.78) .. controls (288.95,130.39) and (290.89,128.44) .. (293.29,128.44) .. controls (295.69,128.44) and (297.64,130.39) .. (297.64,132.78) .. controls (297.64,135.18) and (295.69,137.13) .. (293.29,137.13) .. controls (290.89,137.13) and (288.95,135.18) .. (288.95,132.78) -- cycle ;
	\draw    (172.57,133.63) -- (243.17,133.63) ;
	\draw    (172.57,133.63) -- (207.87,91.53) ;
	\draw    (138.2,173.59) -- (171.98,132.53) ;
	\draw    (101.39,133.56) -- (136.93,90.25) ;
	\draw    (138.2,173.59) -- (101.39,133.56) ;
	\draw    (171.98,132.53) -- (136.93,90.25) ;
	\draw    (101.39,133.56) -- (57.05,100.12) ;
	\draw    (57.05,100.12) -- (57.48,167.47) ;
	\draw  [fill={rgb, 255:red, 126; green, 211; blue, 33 }  ,fill opacity=1 ] (52.7,100.12) .. controls (52.7,97.72) and (54.65,95.77) .. (57.05,95.77) .. controls (59.44,95.77) and (61.39,97.72) .. (61.39,100.12) .. controls (61.39,102.52) and (59.44,104.46) .. (57.05,104.46) .. controls (54.65,104.46) and (52.7,102.52) .. (52.7,100.12) -- cycle ;
	\draw    (57.48,167.47) -- (101.39,133.56) ;
	\draw  [fill={rgb, 255:red, 126; green, 211; blue, 33 }  ,fill opacity=1 ] (53.14,167.47) .. controls (53.14,165.07) and (55.08,163.13) .. (57.48,163.13) .. controls (59.88,163.13) and (61.82,165.07) .. (61.82,167.47) .. controls (61.82,169.87) and (59.88,171.81) .. (57.48,171.81) .. controls (55.08,171.81) and (53.14,169.87) .. (53.14,167.47) -- cycle ;
	\draw  [fill={rgb, 255:red, 74; green, 144; blue, 226 }  ,fill opacity=1 ] (97.05,133.56) .. controls (97.05,131.16) and (98.99,129.22) .. (101.39,129.22) .. controls (103.79,129.22) and (105.74,131.16) .. (105.74,133.56) .. controls (105.74,135.96) and (103.79,137.91) .. (101.39,137.91) .. controls (98.99,137.91) and (97.05,135.96) .. (97.05,133.56) -- cycle ;
	\draw    (137.32,90.97) -- (171.02,46.81) ;
	\draw    (171.02,46.81) -- (103.67,46.86) ;
	\draw    (103.67,46.86) -- (137.32,90.97) ;
	\draw  [fill={rgb, 255:red, 74; green, 144; blue, 226 }  ,fill opacity=1 ] (132.98,90.97) .. controls (132.98,88.57) and (134.92,86.63) .. (137.32,86.63) .. controls (139.72,86.63) and (141.67,88.57) .. (141.67,90.97) .. controls (141.67,93.37) and (139.72,95.31) .. (137.32,95.31) .. controls (134.92,95.31) and (132.98,93.37) .. (132.98,90.97) -- cycle ;
	\draw  [fill={rgb, 255:red, 126; green, 211; blue, 33 }  ,fill opacity=1 ] (99.32,46.86) .. controls (99.32,44.46) and (101.27,42.52) .. (103.67,42.52) .. controls (106.07,42.52) and (108.01,44.46) .. (108.01,46.86) .. controls (108.01,49.26) and (106.07,51.21) .. (103.67,51.21) .. controls (101.27,51.21) and (99.32,49.26) .. (99.32,46.86) -- cycle ;
	\draw  [fill={rgb, 255:red, 126; green, 211; blue, 33 }  ,fill opacity=1 ] (166.68,46.81) .. controls (166.68,44.42) and (168.63,42.47) .. (171.02,42.47) .. controls (173.42,42.47) and (175.37,44.42) .. (175.37,46.81) .. controls (175.37,49.21) and (173.42,51.16) .. (171.02,51.16) .. controls (168.63,51.16) and (166.68,49.21) .. (166.68,46.81) -- cycle ;
	\draw    (138.2,173.59) -- (171.83,211) ;
	\draw    (104.41,214.65) -- (138.2,173.59) ;
	\draw    (104.41,214.65) -- (138.05,252.06) ;
	\draw    (138.05,252.06) -- (171.83,211) ;
	\draw  [fill={rgb, 255:red, 74; green, 144; blue, 226 }  ,fill opacity=1 ] (133.85,173.59) .. controls (133.85,171.19) and (135.8,169.25) .. (138.2,169.25) .. controls (140.6,169.25) and (142.54,171.19) .. (142.54,173.59) .. controls (142.54,175.99) and (140.6,177.93) .. (138.2,177.93) .. controls (135.8,177.93) and (133.85,175.99) .. (133.85,173.59) -- cycle ;
	\draw  [fill={rgb, 255:red, 126; green, 211; blue, 33 }  ,fill opacity=1 ] (100.07,214.65) .. controls (100.07,212.25) and (102.01,210.31) .. (104.41,210.31) .. controls (106.81,210.31) and (108.76,212.25) .. (108.76,214.65) .. controls (108.76,217.05) and (106.81,218.99) .. (104.41,218.99) .. controls (102.01,218.99) and (100.07,217.05) .. (100.07,214.65) -- cycle ;
	\draw  [fill={rgb, 255:red, 126; green, 211; blue, 33 }  ,fill opacity=1 ] (167.49,211) .. controls (167.49,208.6) and (169.43,206.66) .. (171.83,206.66) .. controls (174.23,206.66) and (176.18,208.6) .. (176.18,211) .. controls (176.18,213.4) and (174.23,215.34) .. (171.83,215.34) .. controls (169.43,215.34) and (167.49,213.4) .. (167.49,211) -- cycle ;
	\draw  [fill={rgb, 255:red, 126; green, 211; blue, 33 }  ,fill opacity=1 ] (133.71,252.06) .. controls (133.71,249.66) and (135.65,247.72) .. (138.05,247.72) .. controls (140.45,247.72) and (142.39,249.66) .. (142.39,252.06) .. controls (142.39,254.46) and (140.45,256.4) .. (138.05,256.4) .. controls (135.65,256.4) and (133.71,254.46) .. (133.71,252.06) -- cycle ;
	\draw  [fill={rgb, 255:red, 0; green, 0; blue, 0 }  ,fill opacity=1 ] (258.53,230.53) .. controls (258.53,228.13) and (260.47,226.19) .. (262.87,226.19) .. controls (265.27,226.19) and (267.21,228.13) .. (267.21,230.53) .. controls (267.21,232.93) and (265.27,234.87) .. (262.87,234.87) .. controls (260.47,234.87) and (258.53,232.93) .. (258.53,230.53) -- cycle ;
	\draw  [fill={rgb, 255:red, 74; green, 144; blue, 226 }  ,fill opacity=1 ] (258.85,254.59) .. controls (258.85,252.19) and (260.8,250.25) .. (263.2,250.25) .. controls (265.6,250.25) and (267.54,252.19) .. (267.54,254.59) .. controls (267.54,256.99) and (265.6,258.93) .. (263.2,258.93) .. controls (260.8,258.93) and (258.85,256.99) .. (258.85,254.59) -- cycle ;
	\draw  [fill={rgb, 255:red, 126; green, 211; blue, 33 }  ,fill opacity=1 ] (259.49,279) .. controls (259.49,276.6) and (261.43,274.66) .. (263.83,274.66) .. controls (266.23,274.66) and (268.18,276.6) .. (268.18,279) .. controls (268.18,281.4) and (266.23,283.34) .. (263.83,283.34) .. controls (261.43,283.34) and (259.49,281.4) .. (259.49,279) -- cycle ;
	
	\draw (214.21,88.13) node [anchor=south west] [inner sep=0.75pt]  [font=\footnotesize]  {$\m_{1}$};
	\draw (245.17,137.03) node [anchor=north west][inner sep=0.75pt]  [font=\footnotesize]  {$\m_{2}$};
	\draw (174.57,137.03) node [anchor=north west][inner sep=0.75pt]  [font=\footnotesize]  {$\m_{3}$};
	\draw (374.42,124.2) node [anchor=south east] [inner sep=0.75pt]  [font=\footnotesize]  {$\s_{1}$};
	\draw (143.67,90.97) node [anchor=west] [inner sep=0.75pt]  [font=\footnotesize]  {$\s_{2}$};
	\draw (95.67,146.97) node [anchor=west] [inner sep=0.75pt]  [font=\footnotesize]  {$\s_{3}$};
	\draw (131.67,184.97) node [anchor=west] [inner sep=0.75pt]  [font=\footnotesize]  {$\s_{4}$};
	\draw (397.55,88.97) node [anchor=south east] [inner sep=0.75pt]  [font=\footnotesize]  {$\rr_{1}$};
	\draw (453.48,88.13) node [anchor=south west] [inner sep=0.75pt]  [font=\footnotesize]  {$\rr_{2}$};
	\draw (478.83,132.78) node [anchor=west] [inner sep=0.75pt]  [font=\footnotesize]  {$\rr_{3}$};
	\draw (453.48,178.28) node [anchor=north west][inner sep=0.75pt]  [font=\footnotesize]  {$\rr_{4}$};
	\draw (401.89,179.12) node [anchor=north east] [inner sep=0.75pt]  [font=\footnotesize]  {$\rr_{5}$};
	\draw (332.18,125.04) node [anchor=south] [inner sep=0.75pt]  [font=\footnotesize]  {$\rr_{6}$};
	\draw (293.29,125.04) node [anchor=south] [inner sep=0.75pt]  [font=\footnotesize]  {$\rr_{7}$};
	\draw (177.37,43.41) node [anchor=south west] [inner sep=0.75pt]  [font=\footnotesize]  {$\rr_{8}$};
	\draw (97.32,43.46) node [anchor=south east] [inner sep=0.75pt]  [font=\footnotesize]  {$\rr_{9}$};
	\draw (50.7,96.72) node [anchor=south east] [inner sep=0.75pt]  [font=\footnotesize]  {$\rr_{10}$};
	\draw (51.14,164.07) node [anchor=south east] [inner sep=0.75pt]  [font=\footnotesize]  {$\rr_{11}$};
	\draw (98.07,211.25) node [anchor=south east] [inner sep=0.75pt]  [font=\footnotesize]  {$\rr_{12}$};
	\draw (138.05,259.8) node [anchor=north] [inner sep=0.75pt]  [font=\footnotesize]  {$\rr_{13}$};
	\draw (178.18,214.4) node [anchor=north west][inner sep=0.75pt]  [font=\footnotesize]  {$\rr_{14}$};
	\draw (276,225.4) node [anchor=north west][inner sep=0.75pt]    {$\m_{i}$};
	\draw (311,225) node [anchor=north west][inner sep=0.75pt]   [align=left] {Main Warehouse};
	\draw (278,250.4) node [anchor=north west][inner sep=0.75pt]    {$\s_{i}$};
	\draw (312,250) node [anchor=north west][inner sep=0.75pt]   [align=left] {Small Warehouse};
	\draw (279,275.4) node [anchor=north west][inner sep=0.75pt]    {$\rr_{i}$};
	\draw (313,275) node [anchor=north west][inner sep=0.75pt]   [align=left] {Retail Store};

\end{tikzpicture}
}
\caption{Supply Chain Network as Leafless Cactus Graph $\G$}
\label{cactiassupplychain}
\end{center}
\end{figure}

This figure illustrates that the three main warehouses are all interconnected, so that if a warehouse has a shortage, the other two can overcome that. The warehouse $\m_2$ not only serves goods to retail outlets $\rr_6,\rr_7$, but also delivers them to the smaller warehouse $\s_1$. On the other hand, the warehouse $\m_3$ serves three smaller warehouses $\s_2,\s_3,\s_4$, which in turn serve their corresponding retail outlets.

As stated earlier, the supply chain model is designed in such a way that there is minimum overlap, and most of the routes are cyclic.

\subsection{Application of Fault Tolerant Metric Basis in Supply Chain Logistics}

Let us now consider a scenario where a firm wants to replace the delivery workers with automated delivery vehicles and robots. There is a need to designate certain points in the graph as landmarks, which can then be referenced and used by the machines to know the precise location of all other points of interest, as well as their own position. These vehicles and robots must also communicate real time data about their location to one another to either mitigate or increase the chance of multiple coverage, depending on the actual supply and demand.

Metric basis is an effective way to achieve this goal for any graph \cite{khuller1996landmarks}. The problem arises when the landmark itself is out of commission, due to any reason, e.g., the road leading to the store is temporarily closed, or it is unavailable due to renovations etc. The corresponding vertex in this scenario can not be used for reference purposes.

This can be overcome with the introduction of fault tolerant metric basis. For the supply chain graph given in figure \ref{cactiassupplychain}, Theorem \ref{finalresult} implies that $\beta'(G)=2(4+0)=8$. We can actually calculate the fault tolerant metric basis using the results given in Section $4$. Let $\W=\{\rr_1,\rr_2,\rr_8,\rr_9,\rr_{10},\rr_{11},\rr_{12},\rr_{14}\}$ be the set of fault tolerant landmarks for this graph. This subset of landmarks describes the whole graph in terms of distances from $\W$, e.g., $\rr(\m_1|\W)=(5,6,3,3,4,4,3,3)$ and $\rr(\s_2|\W)=(6,7,1,1,2,2,3,3)$.

Let us now consider that the retail store $\rr_2$ is no longer available for referencing. If we again consider the representation of vertices $\m_1$ and $\s_2$, we see that $\rr(\m_1|\W)=(5,3,3,4,4,3,3)$ and $\rr(\s_2|\W)=(6,1,1,2,2,3,3)$. It is obviously clear that both representations are still distinct. In fact, this distinctness can be observed for all the vertices of the supply chain graph.

This example illustrates that the concept of fault tolerant metric basis can easily deal with disruptions in the supply chain network. The unavailability of a landmark does not hinder the effectiveness of automated vehicles and robots. They can still uniquely identify all the relevant points of interests (warehouses and stores), and communicate their unique positions to one another.

\section{Conclusion}
	In this article, we concluded that the fault tolerant metric dimension of Infinity and Kayak Paddle graphs is always constant. We showed that $\beta'(\G)=4$ for both these graphs. We then extended the results to leafless cacti class of connected graphs, and observed that fault tolerant metric dimension of these graphs depends on the number of outer cycles, as well as on the number of even inner cycles with exactly two antipodal vertices. In particular, we showed that $\beta'(\G)=2(\n_1+\n_2)$, where $\n_1$ is the total number of outer cycles of $\G$, while $\n_2$ is the number of even cycles with exactly two opposite common vertices. We then discussed a theoretical application of leafless cacti graph. We established that the concept of fault tolerant metric basis is an effective way to deal with failures in an automated supply chain logistics problem.

\bibliographystyle{ieeetr}
\bibliography{FTMDBicyclic}
\end{document}